\documentclass[journal, 10pt, twoside]{IEEEtran}

\newcommand{\markchanges}{false}

\ifCLASSOPTIONdraftcls
	\newcommand{\defaultfigwidth}{0.65\columnwidth}
\else
  \newcommand{\defaultfigwidth}{\columnwidth}
\fi

\usepackage{times}
\usepackage{ifthen}
\usepackage[english]{babel}
\usepackage{amsmath, amsfonts, bm}
\usepackage{graphicx}
\usepackage{color}
\usepackage[latin1]{inputenc}
\usepackage{flushend}

\usepackage{psfrag}

\usepackage{amsthm}

\newtheorem{lem}{Lemma}
\newtheorem{cor}{Corollary}
\newtheorem{theorem}{Theorem}
\theoremstyle{definition}

\usepackage{bm}
\newcommand{\ma}[1]{\bm{ {#1} }}         
\newcommand{\ten}[1]{\ma{{\mathcal{#1}}}} 

\newcommand{\compl}{\mathbb{C}}        
\newcommand{\real}{\mathbb{R}}         

\newcommand{\eqdef}{\stackrel{.}{=}} 

\renewcommand{\pi}{\uppi} 
\renewcommand{\j}{\jmath} 

\newcommand{\realof}[1]{{\rm Re}\left\{ #1 \right\}}           
\newcommand{\imagof}[1]{{\rm Im}\left\{ #1 \right\}}           
\newcommand{\imagofnlr}[2]{{\rm Im}#2\{ #1 #2\}}           
\newcommand{\vecof}[1]{{\rm vec}\left\{ #1 \right\}}           
\newcommand{\diagof}[1]{{\rm diag}\left( #1 \right)}         
\newcommand{\traceof}[1]{{\rm Tr}\left( #1 \right)}          
\newcommand{\rankof}[1]{{\rm rank}\left( #1 \right)}         
\newcommand{\argof}[1]{{\rm arg}\left( #1 \right)}         

\newcommand{\vecofnlr}[2]{{\rm vec}#2\{ #1 #2\}}           

\newcommand{\normof}[2]{\left\|#1\right\|_{#2}}
\newcommand{\twonorm}[1]{\left\|#1\right\|_2}                  
\newcommand{\fronorm}[1]{\left\|#1\right\|_{\rm F}}            
\newcommand{\honorm}[1]{\left\|#1\right\|_{\rm H}}             
\newcommand{\normofnlr}[3]{#3\|#1#3\|_{#2}}     
\newcommand{\twonormnlr}[2]{\normofnlr{#1}{2}{#2}} 
\newcommand{\froonormnlr}[2]{\normofnlr{#1}{{\rm F}}{#2}} 

\newcommand{\expvof}[1]{\mathbb{E}\left\{ #1 \right\}}         
\newcommand{\expof}[1]{{\rm e}^{#1}}                           

\newcommand{\eigvof}[2]{\mathop{{\rm EV}}_{#2}\left\{#1\right\}} 

\newcommand{\bigO}[1]{{\mathcal O}\left\{#1\right\}}   

\newcommand{\matelem}[3]{\left[#1\right]_{(#2,#3)}} 

\newcommand{\argmin}{\mathop{{\arg \min}}}         

\newcommand{\proj}[1]{\mathbf{P\!r}_{#1}}           
\newcommand{\projp}[1]{\mathbf{P\!r}_{#1}^\perp}    

\newcommand{\sig}[1]{#1^{\rm[s]}}
\newcommand{\sigH}[1]{#1^{\rm[s]^\herm}}
\newcommand{\sigT}[1]{#1^{\rm[s]^\trans}}
\newcommand{\sigC}[1]{#1^{\rm[s]^\conj}}
\newcommand{\siginv}[1]{#1^{\rm[s]^{-1}}}

\newcommand{\noi}[1]{#1^{\rm[n]}}
\newcommand{\noiH}[1]{#1^{\rm[n]^\herm}}
\newcommand{\noiT}[1]{#1^{\rm[n]^\trans}}
\newcommand{\noiC}[1]{#1^{\rm[n]^\conj}}

\newcommand{\fba}[1]{#1^{({\rm fba})}}
\newcommand{\fbaT}[1]{#1^{({\rm fba})^\trans}}
\newcommand{\fbaH}[1]{#1^{({\rm fba})^\herm}}

\newcommand{\fbainv}[1]{#1^{({\rm fba})^{-1}}}
\newcommand{\sigfba}[1]{#1^{{\rm [s]}({\rm fba})}}

\newcommand{\subsel}[1]{#1^{({\rm sel})}}

\newcommand{\outerp}{\circ}  
\newcommand{\krp}{\diamond}  
\newcommand{\kron}{\otimes}  

\newcommand{\pinv}{+}          
\newcommand{\conj}{*}          
\newcommand{\trans}{{\rm T}}   
\newcommand{\herm}{{\rm H}}    

\newcommand{\tconcat}{\mbox{\Huge \textvisiblespace \hspace{-0.2ex}} \;}
\newcommand{\catten}[3]{\left[#1 \; \tconcat_{#3} \; #2\right]}                    
\newcommand{\cattenmult}[4]{\left[#1 \tconcat_{#4} #2 \tconcat_{#4} \ldots \tconcat_{#4} #3\right]}            
\newcommand{\unf}[2]{\left[ \ten{#1} \right]_{(#2)}}
\newcommand{\unfnot}[2]{\left[ #1 \right]_{(#2)}}

\usepackage{color}
\newcommand{\red}[1]{{\color{red} #1}}

\ifthenelse{\equal{\markchanges}{true}}{
\newcommand{\note}[1]{\red{\bfseries #1}}
\newcommand{\TODO}[1]{\par\red{\bfseries ***TODO: #1***}\\}

\newcommand{\reva}[1]{#1}

\newcommand{\revf}[1]{#1}
\newcommand{\revA}[1]{\red{#1}}
\newcommand{\revB}[1]{\blue{#1}}
\newcommand{\temporary}[1]{#1}
}{
\newcommand{\note}[1]{\errmessage{No notes in the final version!}}
\newcommand{\TODO}[1]{\errmessage{No more ToDOs in the final version!}}

\newcommand{\reva}[1]{#1}

\newcommand{\revf}[1]{#1}
\newcommand{\revA}[1]{#1}
\newcommand{\revB}[1]{#1}
\newcommand{\temporary}[1]{}
}

\newcommand{\thetitle}{A framework for the analytical performance assessment of matrix and tensor-based ESPRIT-type algorithms}
\title{\thetitle}
\author{Florian Roemer$^*$~\IEEEmembership{Student Member,~IEEE}
        and
		    Martin Haardt,~\IEEEmembership{Senior Member,~IEEE}
\thanks{
Parts of this paper have been published at the
				{\em IEEE International Conference on Acoustics, Speech, and Signal Processing (ICASSP 2011)},
        Prague, Czech Republic, May 2011, at the
        {\em IEEE International Conference on Acoustics, Speech, and Signal Processing (ICASSP 2010)},
        Dallas, TX, Mar. 2010, and at the {\em 
        Third International Workshop on Computational Advances in Multi-Sensor Adaptive Processing 
        (CAMSAP 2009)},
        Aruba, Dutch Antilles, Dec. 2009.}
\thanks{
The authors 
F. Roemer and M. Haardt are with
Ilmenau University of Technology,
Communications Research Laboratory,
P.~O.~Box 10~05~65, D-98684 Ilmenau, Germany,
e-mail: \{florian.roemer, martin.haardt\}@tu-ilmenau.de,
phone: +49 (3677) 69-2613,
fax: +49 (3677) 69-1195,
WWW: http://www.tu-ilmenau.de/crl.}
\thanks{$*$ corresponding author}
}

\allowdisplaybreaks 

\usepackage[ps2pdf,pdfpagelabels]{hyperref} 
\hypersetup{
   pdftitle={\thetitle},
   pdfauthor={Florian Roemer and Martin Haardt},
   pdfsubject={},
   pdfkeywords={},
   pdfview=FitH,
   pdfstartview=FitH,
   bookmarksnumbered
} 

\begin{document}

\maketitle

\begin{abstract}
\revA{In this paper we present a generic framework for the asymptotic
performance analysis of subspace-based parameter estimation schemes.
It is based on earlier results on an explicit first-order
expansion of the estimation error in the signal subspace obtained
via an SVD of the noisy observation matrix.
We extend these results in a number of aspects.
Firstly, we demonstrate that an explicit first-order expansion of
the Higher-Order SVD (HOSVD)-based subspace estimate can be derived.
Secondly, we show how to obtain 
explicit first-order expansions of the estimation
error of arbitrary ESPRIT-type algorithms and provide the
expressions for $R$-D matrix-based and tensor-based Standard ESPRIT
as well as Unitary ESPRIT. 
Thirdly, we derive closed-form expressions for the mean square error (MSE)
and show that they only depend on the second-order moments of the 
noise. Hence, we only need the noise to be zero mean and possess
finite second order moments. Additional assumptions  such as Gaussianity or
circular symmetry are not needed.
Fourthly, we investigate the effect of using Structured Least Squares (SLS)
to solve the overdetermined shift invariance equations in ESPRIT
and provide an explicit first-order expansion as well as a
closed-form MSE expression.
Finally, we simplify the MSE for the special case of a single source
and compute the asymptotic efficiency of the investigated ESPRIT-type
algorithms in compact closed-form expressions which only depend on the
array size and the effective SNR.

Our results are more general than existing results on the performance
analysis of ESPRIT-type algorithms since (a) we do not need any assumptions
about the noise except for the mean to be zero and the second-order
moments to be finite (in contrast to earlier results that require Gaussianity
and/or second-order circular symmetry); (b) our results are asymptotic in
the effective SNR, i.e., we do not require the number of samples to be large
(in fact we can analyze even the single-snapshot case); (c) we present
a framework that incorporates the SVD-based and the HOSVD-based
subspace estimates as well as Structured Least Squares in one unified manner.}
\end{abstract}


\IEEEpeerreviewmaketitle

\section{Introduction} \label{sec_intro}
\revA{\PARstart{H}{igh} resolution parameter estimation from $R$-dimensional ($R$-D) signals is a
task required for a variety of applications, such as estimating the
multi-dimensional parameters of the dominant multipath components from
MIMO channel measurements~\cite{HTR:04}, which may be used for
geometry-based channel modeling.
Other applications include radar \cite{NS:10},
wireless communications~\cite{LSS:02}, sonar, seismology, and medical
imaging.
In \cite{HRD:08}, we have shown that in the $R$-D case ($R\geq 2$), tensors can be
used to store and manipulate the $R$-D signals in their native multidimensional
form. Based on this idea, we have proposed an enhanced tensor-based signal subspace estimate 
as well as ESPRIT-type algorithms based on tensors in~\cite{HRD:08}.
Their superior performance was shown based on Monte-Carlo simulations. 

In this paper we present a framework for the analytical performance assessment
of subspace-based parameter estimation schemes and apply it to derive a first-order
perturbation expansion for the tensor-based subspace estimate. Moreover, we find
first-order expansions for the estimation error of ESPRIT-type algorithms and
derive {generic} mean square error (MSE) {expressions which only
depend on the second-order moments of the noise
and hence do not require Gaussianity or circular symmetry.} 
This approach allows to assess the gain from using tensors instead of matrices
analytically in order to determine in which scenarios it is particularly pronounced.
We apply the framework for the analysis of $R$-D Standard ESPRIT, $R$-D Unitary ESPRIT,
$R$-D Standard Tensor-ESPRIT and $R$-D Unitary Tensor-ESPRIT. Moreover, we investigate
the effect of using Structured Least Squares (SLS) for the solution of the invariance equations.
Finally, we present simplified MSE expressions for the special case of a single
source impinging on a Uniform Linear Array (ULA) as well as a Uniform Rectangular Array (URA)
and observed under circularly symmetric white noise. These expressions only depend
on the effective Signal to Noise Ratio (SNR) and the array size and the allow
to compute the asymptotic efficiency in closed-form.}



\revA{Analytical performance assessment of subspace-based parameter estimation schemes
has a long standing history in signal processing.
Shortly after the publication of the most prominent candidates, MUSIC \cite{Sch:86} and
ESPRIT \cite{RPK:86}, analytical results on their performance have appeared.
The most frequently cited  papers}
are~\cite{KB:86} for the MUSIC algorithm and~\cite{RH:89} for ESPRIT. However, many follow-up
papers exist which extend the original results, e.g., \cite{PK:89}, \cite{Fri:90}, \cite{MZ:94}, \cite{ZKM:94}, \cite{MHZ:96},
and many others. However, these results have in common that they all go more or less directly 
back to a result on
the distribution of the eigenvectors of a sample covariance matrix from~\cite{And:63,Bri:75}.

In contrast to these results, in~\cite{LLV:93} an entirely different approach was proposed, which provides
an explicit first-order expansion of the subspace of a desired signal component if observed
superimposed by a small additive perturbation. 
This approach has a number of advantages compared to~\cite{Bri:75}. Firstly,~\cite{Bri:75} 
is asymptotic in the sample
size $N$, i.e., the result becomes only accurate as the number of snapshots $N$ is very large, whereas~\cite{LLV:93}
is asymptotic in the effective SNR, i.e., it can be used even for $N=1$ as long as the noise
variance is sufficiently small. Secondly,~\cite{Bri:75} requires strong Gaussianity assumptions, not
only on the perturbation (i.e., the noise), but also on the source symbols. Since~\cite{LLV:93}
is explicit, no assumptions about the statistics of either the desired signal or the perturbation
are needed. \revA{Note that it has recently been shown that \cite{Bri:75} can be extended to the non-Gaussian
and the non-circular case in \cite{DM:11}. However, the large sample size assumption is still needed.}
Thirdly, the covariance expressions from~\cite{Bri:75} are much less intuitive than
the expansion from~\cite{LLV:93} which shows directly how much of the noise subspace ``leaks into''
the signal subspace due to the erroneous estimate. Finally, the expressions involved in~\cite{Bri:75}
are quite complex and tough to handle, whereas~\cite{LLV:93} requires only a few terms which appear
directly as block matrices of the SVD of the noise-free observation matrix.

Due to these advantages we clearly favor~\cite{LLV:93} as a starting point. The authors in~\cite{LLV:93}
have already shown that their results on the perturbation of the subspace can be used to find
a first order expansion for the MUSIC, the Root-MUSIC, the Min-Norm, the State-Space-Realization,
and even the ESPRIT algorithm. However, they only considered 1-D Standard ESPRIT.
We extend their work by considering multiple dimensions ($R$-D ESPRIT), by
incorporating forward-backward-averaging (for Unitary ESPRIT), by considering
the tensor-based subspace estimate (for Standard and Unitary Tensor-ESPRIT), by investigating
the effect of using Structured Least Squares (SLS) to solve the invariance equation instead of the 
Least Squares (LS) solution used
in~\cite{LLV:93}, and by providing generic mean square error (MSE) expressions of the resulting estimation 
errors in these cases. \revA{Note that our MSE expressions depend on the second-order
moments of the noise only. Hence we only assume it to be zero mean (due to the asymptotic
nature of our performance analysis), but do not require it to be Gaussian distributed,
white, or circularly symmetric.}
\revB{This is a particularly attractive feature of our approach with respect
to different types of preprocessing which
alters the noise statistics, e.g., spatial smoothing (which yields spatially correlated noise)
or forward-backward averaging (which annihilates the circular symmetry of the noise).
Since we do not require spatial whiteness or circular symmetry, our MSE expressions
are directly applicable to a wide range of ESPRIT-type algorithms.}

There have been other follow-up papers based on~\cite{LLV:93}. For instance,~\cite{Xu:02}
provides a first-order and second-order perturbation expansion which can be seen as
a generalization of~\cite{LLV:93}. In~\cite{LLM:08} the authors show that there
is also a first-order contribution of the perturbation of the signal subspace
which lies in the signal subspace (which~\cite{Xu:02} and~\cite{LLV:93} have argued to be of second order
and hence negligible). \revA{Note that a perturbation expansions for the signal subspace
and null space projectors based on the sample covariance matrix is provided in \cite{KFP:92} where the
expansion up to an arbitrary order is derived based on a recurrence relation.}
A mean square error expression for Standard ESPRIT is provided 
in~\cite{LV:92}, however, it
does not generalize easily to the tensor case \revA{and it assumes circular symmetry of the noise}.
\revA{Note that the latter assumption implies that it is not applicable to Unitary ESPRIT, since
Forward-Backward Averaging annihilates the circular symmetry of the noise.}
Moreover, other authors have studied
the asymptotical performance of ESPRIT, e.g.,~\cite{SS:91,ESS:93} where harmonic retrieval
from time series is investigated and MSE expressions for a large number of snapshots
as well as MSE expressions for a high SNR are derived. Note that we find MSE expressions
compatible to~\cite{ESS:93} by only assuming a high effective SNR, i.e., either the number
of snapshots or the SNR can tend to infinity. Interestingly,~\cite{SS:91,ESS:93} also
consider the special case for a single source. However, the expressions provided there are specific
to harmonic retrieval from time series and they are not compared to the corresponding Cram\'er-Rao Bound.
Some analytical results on the asymptotic efficiency of MUSIC, Root-MUSIC, ESPRIT, and TLS-ESPRIT are,
among others, presented
in~\cite{PF:88}, \cite{RH:89}, \cite{RH:89b}, and \cite{OVK:91}, respectively.
However, these results are asymptotic in the number of snapshots 
$N$ and sometimes even in the number of sensors $M$.
The asymptotic equivalence of LS-ESPRIT, TLS-ESPRIT, Pro-ESPRIT, and the Matrix Pencil
method has also been shown, see for instance \cite{HS:91}.
Overall, in the matrix case, the number of existing results is quite large, since the underlying
methods have been known for more than two decades. 
However, 
concerning the tensor case and the incorporation of Structured Least Squares,
the existing results are much more scarce.

In the tensor case a first-order expansion for the HOSVD has been proposed in~\cite{LDV:00}.
However, it is not suitable for our application since it does not consider the HOSVD-based subspace
estimate but the subspaces of the separate $n$-mode unfoldings and their singular values. Moreover,
the perturbation is modeled via a single scalar real-valued parameter $\epsilon$. 
A first-order expansion for the best rank-$(R_1,R_2,R_3)$-expansion is provided in~\cite{LdL:04}.
However, again, it is not directly applicable for analyzing the HOSVD-based subspace estimate
as it investigates the approximation error of the entire tensor.
Consequently, our approach to analyze the HOSVD-based subspace estimate based on the link
to the SVD-based subspace estimate via a structured projection is entirely novel. Moreover,
the application of these results to find the analytical performance of Tensor-ESPRIT-type
algorithms is novel as well. It is a particular strength of the framework we use that
many extensions and modifications of ESPRIT are easily incorporated, e.g., Forward-Backward-Averaging
or Structured Least Squares.

This paper is organized as follows: The notation and the data model are introduced
in Sections~\ref{sec_notation} and \ref{sec_dm}, respectively. 
The subsequent
Section \ref{sec_pert_subsp} reviews the first-order perturbation of the matrix-based
subspace estimate and presents the extension to the tensor case.
The performance analysis of ESPRIT-type algorithms is shown in \ref{sec_perf_esprit}.
Numerical results are presented in Section \ref{sec_perf_sims} before drawing the
conclusions in Section~\ref{sec_concl}.
\vspace{-6mm}
\section{Notation} \label{sec_notation}
In order to facilitate the distinction between scalars, matrices, and
tensors, the following notation is used: Scalars are denoted as italic
letters ($a, b, \ldots, A, B, \ldots, \alpha, \beta, \ldots $), column
vectors as lower-case bold-face letters ($\ma{a}, \ma{b}, \ldots$),
matrices as bold-face capitals ($\ma{A}, \ma{B}, \ldots$), and tensors are
written as bold-face calligraphic letters ($\ten{A}, \ten{B}, \ldots$).
Lower-order parts are consistently named: the $(i,j)$-element of the
matrix $\ma{A}$, is denoted as $a_{i,j}$ and the $(i,j,k)$-element of a
third order tensor $\ten{B}$ as $b_{i,j,k}$. 

We use the superscripts
$^\trans, ^\herm,  ^\conj, ^{-1}, ^\pinv$ for transposition, Hermitian transposition,
complex conjugation, matrix inversion, and the Moore-Penrose pseudo
inverse of a matrix,  respectively. 
The trace of a matrix $\ma{A}$ is written as $\traceof{\ma{A}}$.
Moreover, the Kronecker product of two
matrices $\ma{A}$ and $\ma{B}$ is denoted as $\ma{A} \kron \ma{B}$ and
the Khatri-Rao product (column-wise Kronecker product) as $\ma{A} \krp
\ma{B}$. The operator $\vecof{\ma{A}}$ stacks the column of a matrix $\ma{A} 
\in \compl^{M \times N}$ into a column vector of length $M\cdot N \times 1$.
It satisfies the following property
\begin{align}
    \vecof{\ma{A} \cdot \ma{X} \cdot \ma{B}} = \left(\ma{B}^\trans \kron \ma{A}\right) \cdot \vecof{\ma{X}}.
    \label{eqn_veckron}
\end{align}

An $n$-mode vector of an $(I_1 \times I_2 \times \ldots \times
I_N)$-dimensional tensor $\ten{A}$ is an $I_n$-dimensional vector obtained
from $\ten{A}$ by varying the index $i_n$ and keeping the other indices
fixed.
Moreover, a
matrix unfolding of the tensor $\ten{A}$ along the $n$-th mode is denoted
by
$\unf{A}{n}$
and can be understood as a
matrix containing all the $n$-mode vectors of the tensor $\ten{A}$.
The order of the columns is chosen in accordance with~\cite{LDV:00}.

%
The outer product of the tensors
             $\ten{A} \in \compl^{I_1 \times I_2 \times \ldots \times I_N}$ and
             $\ten{B} \in \compl^{J_1 \times J_2 \times \ldots \times J_M}$ is given by
           \begin{equation}
             \begin{array}{l}
                \ten{C} = \ten{A} \outerp \ten{B} \in \compl^{I_1 \times \ldots \times I_N\times J_1 \times \ldots \times J_M}, \quad \mbox{where} \\
                c_{i_1, i_2, \ldots, i_N, j_1, j_2, \ldots, j_M}  =  a_{i_1, i_2, \ldots, i_N}
                        \cdot b_{j_1, j_2, \ldots, j_M}.
             \end{array} \label{eqn_notation_outerp}
           \end{equation}
           In other words, the tensor $\ten{C}$ contains all possible combinations of pairwise products
           between the elements of $\ten{A}$ and $\ten{B}$. This operator is very closely related to
           the Kronecker product defined for matrices.
             %
  
 The $n$-mode product of a tensor $\ten{A} \in \compl^{I_1 \times I_2 \times
             \ldots \times I_N}$ and a matrix $\ma{U} \in \compl^{J_n \times I_n}$
             along the $n$-th
             mode is denoted as $\ten{B} = \ten{A} \times_n \ma{U}$
             and defined via
\begin{equation}
 \ten{B} = \ten{A} \times_n \ma{U} \quad \Leftrightarrow \unf{B}{n} = \ma{U} \cdot \unf{A}{n},
 \label{eqn_notation_nmp}
\end{equation}
             i.e., it may be visualized by multiplying all $n$-mode vectors
             of $\ten{A}$ from the left-hand side by the matrix $\ma{U}$.
Note that the $n$-mode product satisifies
\ifCLASSOPTIONdraftcls
\begin{align}
   & \left(\ten{A} \times_r \ma{U}_r\right)\times_r \ma{V}_r = \ten{A} \times_r \left(\ma{V}_r \cdot \ma{U}_r\right)
   \label{eqn_notation_nmp_twice}  \\
   & \unfnot{\ten{A} \times_1 \ma{U}_1 \ldots \times_R \ma{U}_R}{r}
   = \ma{U}_r \cdot \unf{A}{r} \cdot \left( \ma{U}_{r+1} \kron \ldots \kron \ma{U}_R \kron \ma{U}_1 \kron 
   \ldots
   \kron \ma{U}_{r-1}\right)^\trans
   \label{eqn_notation_nmp_unf}
\end{align}
\else
\begin{align}
   & \left(\ten{A} \times_r \ma{U}_r\right)\times_r \ma{V}_r = \ten{A} \times_r \left(\ma{V}_r \cdot \ma{U}_r\right)
   \label{eqn_notation_nmp_twice}  \\
   & \unfnot{\ten{A} \times_1 \ma{U}_1 \ldots \times_R \ma{U}_R}{r}
   = \ma{U}_r \cdot \unf{A}{r} \cdot \notag \\
   & \quad \left( \ma{U}_{r+1} \kron \ldots \kron \ma{U}_R \kron \ma{U}_1 \kron 
   \ldots
   \kron \ma{U}_{r-1}\right)^\trans
   \label{eqn_notation_nmp_unf}
\end{align}
\fi
for $\ten{A} \in \compl^{I_1 \times I_2 \times \ldots \times I_N}$,
$\ma{U}_r \in \compl^{J_r \times I_r}$ and $\ma{V}_r \in \compl^{K_r \times J_r}$.
  
 The higher-order SVD (HOSVD) \cite{LDV:00} of a tensor
              $\ten{A} \in \compl^{I_1 \times I_2 \times \ldots \times I_N}$ is given by
              \begin{equation}
        \ten{A} = \ten{S} \times_1 \ma{U}_1 \times_2 \ma{U}_2 \ldots \times_N \ma{U}_N,
                   \label{eqn_notation_hosvd}
              \end{equation}
         where $\ten{S} \in \compl^{I_1 \times I_2 \times \ldots \times I_N}$
         is the core tensor which satisfies the all-orthogonality conditions
         \cite{LDV:00} and $\ma{U}_n
         \in \compl^{I_n \times I_n}, \; n= 1, 2, \ldots, N$,
         are the unitary matrices of $n$-mode singular
         vectors. 

We also define the concatenation of two tensors
along the $n$-th mode via the operator $\catten{\ten{A}}{\ten{B}}{n}$.
The Euclidean (vector) norm, the Frobenius (matrix) norm, and the Higher-Order Frobenius (tensor)
norm are denoted by $\twonorm{\ma{a}}$, $\fronorm{\ma{A}}$, and $\honorm{\ten{A}}$, respectively.
All three norms are computed by taking the square-root of the sum of the squared magnitude of all
the elements in their arguments.


The matrix $\ma{K}_{M \times N}$ denotes the commutation matrices \cite{MN:95}
which satisfy
\begin{align}
  \ma{K}_{M \times N}
 \cdot \vecof{\ma{A}^\trans} = \vecof{\ma{A}}     
    \label{eqn_def_commat} \\
   \ma{K}_{M \times P}^\trans
   \cdot \left(\ma{A} \kron \ma{B} \right)
   \cdot \ma{K}_{N \times Q} = \ma{B} \kron  \ma{A}
     \label{eqn_commat_permkron}
\end{align}
for $\ma{A} \in \compl^{M \times N}$, $\ma{B} \in \compl^{P \times Q}$.

\revA{A projection matrix onto the column space of a matrix $\ma{A} \in \compl^{M \times r}$ 
is denoted as $\proj{\ma{A}} \revB{= \ma{A} \cdot \ma{A}^\pinv} \in \compl^{M \times M}$
and its orthogonal complement by $\projp{\ma{A}} = \ma{I}_M - \proj{\ma{A}}$.
Note that for $r=1$, i.e., $\ma{A} = \ma{a}$, this matrix can be computed as
$\proj{\ma{a}} = \frac{\ma{a}\cdot\ma{a}^\herm}{\ma{a}^\herm\cdot\ma{a}}$.}

A $p \times p$ matrix $\ma{Q}_p$ is called left-$\ma{\Pi}$-real if $\ma{\Pi}_p \cdot \ma{Q}_p^\conj
= \ma{Q}_p$, where $\ma{\Pi}_p$ is the $p \times p$ exchange matrix with ones on its antidiagonal
and zeros elsewhere. The special set of unitary sparse left-$\ma{\Pi}$-real matrices 
 introduced in~\cite{HN:95} is denoted as $\ma{Q}_p^{\rm (s)}$. 
Furthermore, a matrix $\ma{X}\in \compl^{M \times N}$ is called centro-Hermitian if 
$\ma{\Pi}_M \cdot \ma{X}^\conj \cdot \ma{\Pi}_N
= \ma{X}$. The vector $\ma{e}_k$ denotes the $k$-th column of an identity matrix.

\section{Data model}\label{sec_dm}
\subsection{Matrix-based and tensor-based data model} \label{sec_dm_matten}


The observations are modeled as a superposition of $d$ undamped exponentials
 sampled on an
 $R$-dimensional grid of size $M_1 \times M_2 \times \ldots \times M_R$ at $N$ subsequent time
 instants \cite{HN:98}. The measurement samples are given by
 \begin{equation}      \label{eqn_dm_element}
    x_{m_1, m_2, \ldots, m_R, t_n}  =  \sum_{i=1}^d s_i(t_n) \prod_{r=1}^R e^{\j \cdot (m_r-1)
                                       \cdot \mu_i^{(r)} }  
     +  n_{m_1, m_2, \ldots, m_R, t_n},
 \end{equation}
   where $m_r = 1, 2, \ldots, M_r$, $n = 1, 2, \ldots, N$, $s_i(t_n)$
   denotes the complex amplitude of the $i$-th exponential at time
   instant $t_n$, $\mu_i^{(r)}$ symbolizes the spatial frequency of the
   $i$-th exponential in the $r$-th mode for $i=1, 2, \ldots, d$ and $r=1, 2, \ldots, R$,
   and $n_{m_1, m_2, \ldots,
   m_R, t_n}$ represents the zero mean additive noise component inherent in the
   measurement process%
%
\footnote{Note that equation~(\ref{eqn_dm_element})
assumes a uniform sampling in the spatial
domain. However, this assumption can be relaxed to more generic geometries
as long as they feature shift invariances and can be constructed as the outer
product of $R$ one-dimensional sampling grids.}.
In the context of array signal processing, each of the $R$-dimensional 
exponentials represents one planar wavefront and
the complex amplitudes $s_i(t_n)$ are the symbols.
It is our goal to estimate the spatial frequencies $\mu_i^{(r)}$
for $r=1, 2, \ldots, R, i=1, 2, \ldots, d$ and their correct
pairing.

In order to arrive at a more compressed formulation of
the data model in~\eqref{eqn_dm_element} we collect the samples 
$x_{m_1,m_2,\ldots,m_R, t_n}$ into one array.
As our signal is referenced by $R+1$ indices, the most natural way of formulating
the model is to employ an $(R+1)$-way array $\ten{X} \in \compl^{M_1 \times M_2 \ldots \times M_R \times N}$
which contains $x_{m_1,m_2,\ldots,m_R,t_n}$ for $m_r = 1, 2, \ldots, M_r$, $r=1, 2, \ldots, R$,
and $n=1, 2, \ldots, N$.
We can then conveniently express $\ten{X}$ as \cite{HRD:08}
\begin{align}
    \ten{X} = \ten{A} \times_{R+1} \ma{S}^\trans + \ten{N},
    \label{eqn_subsb_dm_tensor}
\end{align}
where $\ma{S} \in \compl^{d \times N}$ contains the amplitudes $s_i[n]$
and $\ten{N} \in \compl^{M_1 \times M_2 \ldots \times M_R \times N}$
collects all the noise samples $n_{m_1,m_2,\ldots,m_R,t_n}$ in the same manner as $\ten{X}$.
Finally, $\ten{A} \in \compl^{M_1 \times M_2 \ldots \times M_R \times d}$
is referred to as the ``array steering tensor'' \cite{HRD:08}. It can be expressed
by virtue of the concatenation operator 
via
\begin{align}
   \ten{A} & = \cattenmult{\ten{A}_1}{\ten{A}_2}{\ten{A}_d}{R+1} \label{eqn_subsp_dm_atencat} \\
   \ten{A}_i & = \ma{a}^{(1)}(\mu_i^{(1)}) \outerp \ma{a}^{(2)}(\mu_i^{(2)}) \outerp \ldots \outerp \ma{a}^{(R)}(\mu_i^{(R)})
   \in \compl^{M_1 \times M_2 \times \ldots \times M_R} \notag 
\end{align}
where $\ma{a}^{(r)}(\mu_i^{(r)}) \in \compl^{M_r \times 1}$ represents the array
steering vector of the $i$-th source in the $r$-th mode. 

An alternative expression for the array steering tensor is given by
\begin{align}
   \ten{A} = \ten{I}_{R+1,d} \times_1 \ma{A}^{(1)} \times_2 \ma{A}^{(2)} \ldots \times_R \ma{A}^{(R)},
   \label{eqn_subsp_dm_aten_cp}
\end{align}
where $\ma{A}^{(r)} = 
\begin{bmatrix} \ma{a}^{(r)}(\mu_1^{(r)}) & \ldots & \ma{a}^{(r)}(\mu_d^{(r)})\end{bmatrix}
 \in \compl^{M_r \times d}$
is referred to as the array steering matrix in the $r$-th mode.

The strength of the data model in~\eqref{eqn_subsb_dm_tensor} is that it represents
the signal in its natural multidimensional structure by virtue of the measurement tensor
$\ten{X}$. Before tensor calculus was used in this area, a matrix-based formulation of~\eqref{eqn_subsb_dm_tensor}
was needed. This requires stacking some of the dimensions into rows or columns. A meaningful
definition of a measurement matrix $\ma{X}$ is to apply stacking to all ``spatial'' dimensions
$1, 2, \ldots, R$ along the rows and align the snapshots $n=1, 2, \ldots, N$ as the columns.
Mathematically, we can write $\ma{X} = \unf{X}{R+1}^\trans \in \compl^{M \times N}$, where
$M = \prod_{r=1}^R M_r$. Applying this stacking operation to~\eqref{eqn_subsb_dm_tensor},
we arrive at the matrix-based data model~\cite{HN:98}
\begin{align}
    \ma{X} = \ma{A} \cdot \ma{S} + \ma{N}. \label{eqn_subsb_dm_matrix_rd}
\end{align}
%
Here, $\ma{A} = \unf{A}{R+1}^\trans \in \compl^{M \times d}$ and $\ma{N} = \unf{N}{R+1}^\trans \in \compl^{M \times N}$.
Note that 
$\ma{A}$ is highly structured since it satisfies
\begin{align}
   \ma{A} = \unf{A}{R+1}^\trans = \ma{A}^{(1)} \krp \ma{A}^{(2)} \krp \ldots \krp \ma{A}^{(R)}.
   \label{eqn_subsb_dm_matrix_amatkrp}
\end{align}
%

\subsection{Subspace estimation} \label{sec_dm_subsp}

\revA{The first step in all subspace-based parameter estimation schemes is the estimation
of a basis for the signal subspace from the noisy observations. In the matrix case, this can\revB{,}
for instance\revB{,} be achieved by a truncated SVD of $\ma{X}$. Let $\ma{\hat{U}}_{\rm s} \in \compl^{M \times d}$
be the matrix containing the $d$ dominant left singular vectors of $\ma{X}$. Then the column
space of $\ma{\hat{U}}_{\rm s}$ is an estimate for the signal subspace spanned by the columns
of $\ma{A}$ and we can write
%
    $\ma{A} \approx \ma{\hat{U}}_{\rm s} \cdot \ma{T}$
%
for a non-singular matrix $\ma{T} \in \compl^{d \times d}$.

A tensor-based extension of this subspace estimate was proposed in \cite{HRD:08}. To this end, let
the truncated HOSVD of $\ten{X}$ be given by
\begin{align}
   \ten{X} \approx \sig{\ten{\hat{S}}} \times_1 \sig{\ma{\hat{U}}}_1 \ldots \times_R \sig{\ma{\hat{U}}}_R \times_{R+1} \sig{\ma{\hat{U}}}_{R+1},
   \label{eqn_dm_trnchosvd}
\end{align}
where $\sig{\ten{\hat{S}}} \in \compl^{p_1 \times \ldots \times p_R \times d}$ is the truncated
core tensor and $\sig{\ma{\hat{U}}}_r \in \compl^{M_r \times p_r}$ for $r=1, 2, \ldots, R$, $\sig{\ma{\hat{U}}}_{R+1}
\in \compl^{N \times d}$ are the matrices of dominant $r$-mode singular vectors. 
Here, $p_r = \rankof{\unfnot{\ten{X}_0}{r}}$ represents the $r$-rank of the noise-free
observation tensor $\ten{X}_0 = \ten{A} \times_{R+1} \ma{S}^\trans$.
Based on \eqref{eqn_dm_trnchosvd}, a tensor-based subspace estimate can be defined as
\begin{align}
   \sig{\ten{\hat{U}}} = \sig{\ten{\hat{S}}} \times_1 \sig{\ma{\hat{U}}}_1 \ldots \times_R \sig{\ma{\hat{U}}}_R
   \times_{R+1} \ma{\hat{\Sigma}}_{\rm s}^{-1}. \label{eqn_subsp_usten_def}
\end{align}
Note that the ($R+1$)-mode multiplication with $\ma{\hat{\Sigma}}_{\rm s}^{-1}$ is introduced
in addition to its original definition in \cite{HRD:08} since it simplifies the notation we need at this point
and it has no impact on the subspace estimation accuracy.
Here $\ma{\hat{\Sigma}}_{\rm s}$ refers to the diagonal matrix 
containing the $d$ dominant singular values of $\ma{X}$ on its main diagonal.

Note that $\sig{\ten{\hat{U}}}$ satisfies $\sig{\ten{\hat{U}}} \approx \ten{A} \times_{R+1} \ma{\bar{T}}$
for a non-singular matrix $\ma{\bar{T}} \in \compl^{d \times d}$.
Based on $\sig{\ten{\hat{U}}}$, an improved signal subspace estimate
is given by the matrix $\unfnot{\sig{\ten{\hat{U}}}}{R+1}^\trans \in \compl^{M \times d}$.
}

\section{Perturbations of the subspace estimates} \label{sec_pert_subsp}

\subsection{Review of perturbation results for the SVD}

Let us first review the results from~\cite{LLV:93} which are relevant to the
discussion in this section. Let $\ma{X}_0 = \ma{A} \cdot \ma{S} \in \compl^{M \times N}$ 
be a matrix containing the noise-free observations such that
$	\ma{X} = \ma{X}_0 + \ma{N}$ where $\ma{N}$ represents the undesired
perturbation (noise). 

The SVD of $\ma{X}_0$ can be expressed as
\ifCLASSOPTIONdraftcls
  \begin{align}
   \ma{X}_0 = \begin{bmatrix} \ma{U}_{\rm s} & \ma{U}_{\rm n} \end{bmatrix} \cdot
   						\begin{bmatrix} \ma{\Sigma}_{\rm s} & \ma{0}_{d \times (N-d)} \\ 
   						                \ma{0}_{(M-d) \times d} & \ma{0}_{(M-d) \times (N-d)}\end{bmatrix} \cdot
							\begin{bmatrix} \ma{V}_{\rm s} & \ma{V}_{\rm n} \end{bmatrix}^\herm,
  \end{align}
\else
  \begin{align}
   \ma{X}_0 = \begin{bmatrix} \ma{U}_{\rm s} & \ma{U}_{\rm n} \end{bmatrix} \cdot
   						\begin{bmatrix} \ma{\Sigma}_{\rm s} & \ma{0} \\ 
   						                \ma{0} & \ma{0}\end{bmatrix} \cdot
							\begin{bmatrix} \ma{V}_{\rm s} & \ma{V}_{\rm n} \end{bmatrix}^\herm,
  \end{align}
\fi
where the columns of $\ma{U}_{\rm s} \in \compl^{M \times d}$ provide an orthonormal basis for the signal subspace
which we want to estimate. Moreover 
$\ma{\Sigma}_{\rm s} = \diagof{\begin{bmatrix} \sigma_1, \sigma_2, \ldots, \sigma_d \end{bmatrix}}
\in \real^{d \times d}$ contains the $d$ non-zero singular values on its main diagonal.
We find an estimate for $\ma{U}_{\rm s}$ by computing an SVD
of the noisy observation matrix $\ma{X}$ which can be expressed as
\ifCLASSOPTIONdraftcls
  \begin{align}
   \ma{X}   = \begin{bmatrix} \ma{\hat{U}}_{\rm s} & \ma{\hat{U}}_{\rm n} \end{bmatrix} \cdot
   						\begin{bmatrix} \ma{\hat{\Sigma}}_{\rm s} & \ma{0}_{d \times (N-d)} \\ 
   						                \ma{0}_{(M-d) \times d} & \ma{\hat{\Sigma}}_{\rm n}\end{bmatrix} \cdot
							\begin{bmatrix} \ma{\hat{V}}_{\rm s} & \ma{\hat{V}}_{\rm n} \end{bmatrix}^\herm,
							\label{eqn_svd_X}
  \end{align}
\else
  \begin{align}
   \ma{X}   = \begin{bmatrix} \ma{\hat{U}}_{\rm s} & \ma{\hat{U}}_{\rm n} \end{bmatrix} \cdot
   						\begin{bmatrix} \ma{\hat{\Sigma}}_{\rm s} & \ma{0} \\ 
   						                \ma{0} & \ma{\hat{\Sigma}}_{\rm n}\end{bmatrix} \cdot
							\begin{bmatrix} \ma{\hat{V}}_{\rm s} & \ma{\hat{V}}_{\rm n} \end{bmatrix}^\herm,
							\label{eqn_svd_X}
  \end{align}
\fi
where the ``hat'' denotes the estimated quantities. We can write $\ma{\hat{U}}_{\rm s} = 
\ma{U}_{\rm s} + \Delta \ma{U}_{\rm s}$, where $\Delta \ma{U}_{\rm s}$ represents
the estimation error. 
At this point we are ready to state the main result
on the first order perturbation expansion of $\Delta \ma{U}_{\rm s}$ from~\cite{LLV:93}
\ifCLASSOPTIONdraftcls
\begin{align}    
    \Delta \ma{U}_{\rm s} & = \ma{U}_{\rm n} \cdot \ma{\Gamma}_{\rm n} + \bigO{\Delta^2}, \; \mbox{where}
    \; \Delta = \normof{\ma{N}}{} \; \mbox{and} \; 
    \ma{\Gamma}_{\rm n}  = \ma{U}_{\rm n}^\herm \cdot \ma{N} \cdot \ma{V}_{\rm s} \cdot \ma{\Sigma}_{\rm s}^{-1}
    \in \compl^{(M - d) \times d}
      \label{eqn_perf_linexp_vac}
\end{align}
\else
\begin{align}    
    \Delta \ma{U}_{\rm s} & = \ma{U}_{\rm n} \cdot \ma{\Gamma}_{\rm n} + \bigO{\Delta^2}, \; \mbox{where}
    \; \Delta = \normof{\ma{N}}{} \; \mbox{and} \; \notag \\
    \ma{\Gamma}_{\rm n}  & = \ma{U}_{\rm n}^\herm \cdot \ma{N} \cdot \ma{V}_{\rm s} \cdot \ma{\Sigma}_{\rm s}^{-1}
    \in \compl^{(M - d) \times d}
      \label{eqn_perf_linexp_vac}
\end{align}
\fi
Here $\normof{.}{}$ represents an arbitrary sub-multiplicative\footnote{{A matrix norm is called
submultiplicative if $\normof{\ma{A} \cdot \ma{B}}{} \leq \normof{\ma{A}}{} \cdot \normof{\ma{B}}{}$
for arbitrary matrices $\ma{A}$ and $\ma{B}$.}} norm, e.g., the Frobenius norm.
Equation~\eqref{eqn_perf_linexp_vac} shows the first order expansion of the signal
subspace estimation error $\Delta \ma{U}_{\rm s}$ in terms of the noise subspace $ \ma{U}_{\rm n}$, i.e.,
how much of the noise subspace ``leaks into'' the signal subspace due to the estimation errors
from the perturbation $\ma{N}$. Since it is explicit in $\ma{N}$ it makes no assumptions about
the {\em statistics} of $\ma{N}$, in fact, it is purely deterministic.

The expansion~\eqref{eqn_perf_linexp_vac} only models the leakage of the noise subspace into
the signal subspace. That is to say, the perturbation of the particular basis (the columns
of $\ma{U}_{\rm s}$) is ignored. While for subspace-based parameter estimation schemes
this is indeed sufficient since the particular choice of the basis is irrelevant, there
are other applications where this term matters. For instance, in a communication system
where the channel is decomposed into its individual eigenmodes and one or several of these
eigenmodes are used for transmission, such errors have a major impact. Therefore, other
authors have extended~\eqref{eqn_perf_linexp_vac} to take this error term into account.
For instance, in~\cite{LLM:08} the authors provide the following expansion
%
\begin{align}    
    \Delta \ma{U}_{\rm s} & = \ma{U}_{\rm n} \cdot \ma{\Gamma}_{\rm n} + \ma{U}_{\rm s} \cdot \ma{\Gamma}_{\rm s} + \bigO{\Delta^2}, 
    \; \mbox{where} \label{eqn_perf_linexp_vacma}\\   
    \ma{\Gamma}_{\rm s} & = 
       \ma{D} \odot \left( \ma{U}_{\rm s}^\herm \cdot \ma{N} \cdot \ma{V}_{\rm s} \cdot \ma{\Sigma}_{\rm s}
                               + \ma{\Sigma}_{\rm s} \cdot \ma{V}_{\rm s}^\herm \cdot \ma{N}^\herm \cdot \ma{U}_{\rm s} \right) \in \compl^{p_r \times p_r}. \notag
\end{align}
Here, the matrix $\ma{D}$ is defined as 
\begin{align}
		\matelem{\ma{D}}{k}{\ell} = 
		 \begin{cases}
		     \frac{1}{\sigma_\ell^2 - \sigma_k^2} & k \neq \ell \\
		     0 & k = \ell
		 \end{cases}
		 \quad 
		 \mbox{for} \; k,\ell=1,2,\ldots,d.
\end{align}
Equation~\eqref{eqn_perf_linexp_vacma} additionally shows the perturbation of the individual
singular vectors via the term $\ma{U}_{\rm s} \cdot \ma{\Gamma}_{\rm s}$. This term can be dropped
for the evaluation of subspace-based parameter estimation schemes since for these, the particular
choice of the basis is irrelevant. Therefore we do not consider it in Section~\ref{sec_perf_esprit}
where ESPRIT-type algorithms are investigated. However, we show its impact in
the simulation results in Section~\ref{sec_perf_sims} where the subspace estimation
accuracy is evaluated.

\subsection{Extension to the HOSVD-based subspace estimate} 

As we have shown in Section~\ref{sec_dm_subsp}, in the multidimensional case, an improved
signal subspace estimate can be computed via the HOSVD of the measurement tensor $\ten{X}$.
Since the HOSVD is computed via SVDs of the unfoldings, we can apply the same framework
to find a perturbation expansion of the HOSVD-based subspace estimate. In order to distinguish
unperturbed from estimated (perturbed) quantities we express $\ten{X}$ as
$\ten{X} = \ten{X}_0 + \ten{N}$,
where $\ten{X}_0 = \ten{A} \times_{R+1} \ma{S}^\trans$ is the unperturbed observation tensor.
%
The SVD of the $r$-mode unfoldings
of $\ten{X}$ and $\ten{X}_0$ are then given by
\ifCLASSOPTIONdraftcls
  \begin{align}
   \unfnot{\ten{X}_0}{r} & = 
   						\begin{bmatrix} \sig{\ma{U}}_r & \noi{\ma{U}}_r \end{bmatrix} \cdot
   						\begin{bmatrix} \sig{\ma{\Sigma}}_r & \ma{0}_{d \times (M\cdot N/M_r-d)} \\ 
   						                \ma{0}_{(M_r-d) \times d} & \ma{0}_{(M_r-d) \times (M\cdot N/M_r-d)}\end{bmatrix} \cdot
							\begin{bmatrix} \sig{\ma{V}}_r & \noi{\ma{V}}_r \end{bmatrix}^\herm \label{eqn_perf_unfsvd_x0} \\
	 \unfnot{\ten{X}}{r} & = 
   						\begin{bmatrix} \sig{\ma{\hat{U}}}_r & \noi{\ma{\hat{U}}}_r \end{bmatrix} \cdot
   						\begin{bmatrix} \sig{\ma{\hat{\Sigma}}}_r & \ma{0}_{d \times (M\cdot N/M_r-d)} \\ 
   						                \ma{0}_{(M_r-d) \times d} & \noi{\ma{\hat{\Sigma}}}_r\end{bmatrix} \cdot
							\begin{bmatrix} \sig{\ma{\hat{V}}}_r & \noi{\ma{\hat{V}}}_r \end{bmatrix}^\herm.
	\label{eqn_perf_unfsvd_x}
  \end{align}
\else
  \begin{align}
   \unfnot{\ten{X}_0}{r} & = 
   						\begin{bmatrix} \sig{\ma{U}}_r & \noi{\ma{U}}_r \end{bmatrix} \cdot
   						\begin{bmatrix} \sig{\ma{\Sigma}}_r & \ma{0} \\ 
   						                \ma{0} & \ma{0}\end{bmatrix} \cdot
							\begin{bmatrix} \sig{\ma{V}}_r & \noi{\ma{V}}_r \end{bmatrix}^\herm \label{eqn_perf_unfsvd_x0} \\
	 \unfnot{\ten{X}}{r} & = 
   						\begin{bmatrix} \sig{\ma{\hat{U}}}_r & \noi{\ma{\hat{U}}}_r \end{bmatrix} \cdot
   						\begin{bmatrix} \sig{\ma{\hat{\Sigma}}}_r & \ma{0} \\ 
   						                \ma{0} & \noi{\ma{\hat{\Sigma}}}_r\end{bmatrix} \cdot
							\begin{bmatrix} \sig{\ma{\hat{V}}}_r & \noi{\ma{\hat{V}}}_r \end{bmatrix}^\herm
	\label{eqn_perf_unfsvd_x}
  \end{align}
\fi  
where $\sig{\ma{\Sigma}}_r = \diagof{[\sigma_1^{(r)}, \sigma_2^{(r)}, \ldots, \sigma_d^{(r)}]}$
and $r=1, 2, \ldots, R$.
Note that since~\eqref{eqn_perf_unfsvd_x0} and~\eqref{eqn_perf_unfsvd_x} are in fact SVDs,
we can apply~\eqref{eqn_perf_linexp_vacma} and find
$\sig{\ma{\hat{U}}}_r = \sig{\ma{U}}_r +\Delta \sig{\ma{U}}_r$ where
\ifCLASSOPTIONdraftcls
\begin{align}
   \Delta \sig{\ma{U}}_r & = \noi{\ma{U}}_r \cdot \noi{\ma{\Gamma}}_{r} + \sig{\ma{U}}_r \cdot \sig{\ma{\Gamma}}_{r} + \bigO{\Delta^2}, \quad 
   \label{eqn_perf_unfexp} \\
   \noi{\ma{\Gamma}}_{r} & = \noiH{\ma{U}}_r \cdot \unf{N}{r} \cdot \sig{\ma{V}}_r \cdot \siginv{\ma{\Sigma}}_r,
   \notag \\
   \sig{\ma{\Gamma}}_{r} & = \ma{D}_r \odot \left( 
       \sigH{\ma{U}}_r \cdot \unf{N}{r} \cdot \sig{\ma{V}}_r \cdot \sig{\ma{\Sigma}}_r
     + \sig{\ma{\Sigma}}_r \cdot \sigH{\ma{V}}_r \cdot \unf{N}{r}^\herm \cdot \sig{\ma{U}}_r
                                \right) \notag \\
   \matelem{\ma{D}_r}{k}{\ell} & = 
		 \begin{cases}
		     \frac{1}{\sigma_\ell^{(r)^2} - \sigma_k^{(r)^2}} & k \neq \ell \\
		     0 & k = \ell
		 \end{cases}
		 \quad 
		 \mbox{for} \; k,\ell=1,2,\ldots,d.   \notag
\end{align} 
\else
\begin{align}
   \Delta \sig{\ma{U}}_r & = \noi{\ma{U}}_r \cdot \noi{\ma{\Gamma}}_{r} + \sig{\ma{U}}_r \cdot \sig{\ma{\Gamma}}_{r} + \bigO{\Delta^2}, \quad 
   \label{eqn_perf_unfexp} \\
   \noi{\ma{\Gamma}}_{r} & = \noiH{\ma{U}}_r \cdot \unf{N}{r} \cdot \sig{\ma{V}}_r \cdot \siginv{\ma{\Sigma}}_r,
   \notag \\
   \sig{\ma{\Gamma}}_{r} & = \ma{D}_r \odot \Big( 
       \sigH{\ma{U}}_r \cdot \unf{N}{r} \cdot \sig{\ma{V}}_r \cdot \sig{\ma{\Sigma}}_r
       \notag \\ &
     + \sig{\ma{\Sigma}}_r \cdot \sigH{\ma{V}}_r \cdot \unf{N}{r}^\herm \cdot \sig{\ma{U}}_r
                                \Big)  \notag \\
   \matelem{\ma{D}_r}{k}{\ell} & = 
		 \begin{cases}
		     \frac{1}{\sigma_\ell^{(r)^2} - \sigma_k^{(r)^2}} & k \neq \ell \\
		     0 & k = \ell
		 \end{cases}
		 \quad 
		 \mbox{for} \; k,\ell=1,2,\ldots,d.   \notag
\end{align} 
\fi
%

Our goal is to use the perturbation of the $r$-mode unfoldings to find a corresponding
expansion for the HOSVD-based subspace estimate introduced in Section~\ref{sec_dm_subsp}.
%
This is facilitated by the following theorem:

\begin{theorem} \label{thm_perf_reltenmat}
The HOSVD-based subspace
estimate $\unfnot{\sig{\ten{\hat{U}}}}{R+1}^\trans$ defined in 
\eqref{eqn_subsp_usten_def} is linked to 
the SVD-based subspace estimate $\ma{\hat{U}}_{\rm s}$ 
via the following algebraic relation
 %
    \begin{align}
       \unfnot{\sig{\ten{\hat{U}}}}{R+1}^\trans = \left(
           \hat{\ma{T}}_1 \otimes \hat{\ma{T}}_2 \otimes \ldots \otimes \hat{\ma{T}}_R
           \right) \cdot \ma{\hat{U}}_{\rm s}, \label{eqn_perf_reltenmat2}
    \end{align}
    where $\ma{\hat{T}}_r \in \compl^{M_r \times M_r}$ represent estimates of the projection
    matrices onto the $r$-spaces of $\ten{X}_0$, which are computed via
    $\ma{\hat{T}}_r = \sig{\ma{\hat{U}}}_r \sigH{\ma{\hat{U}}}_r$.
\end{theorem}

Proof: Relation \eqref{eqn_perf_reltenmat2} was shown in \cite{RBHW:09} for $R=2$. The proof
for $R>2$ proceeds in an analogous manner and is presented in Appendix \ref{sec_app_proof_perf_reltenmat}.

Equation~\eqref{eqn_perf_reltenmat2} shows that a perturbation expansion for $\unfnot{\sig{\ten{\hat{U}}}}{R+1}^\trans$
can be developed based on the subspaces of all $R+1$ unfoldings, as the core tensor is
not needed for its computation.
The result is shown in the following theorem:
\begin{theorem} \label{thm_perf_exptensub}
  The HOSVD-based signal subspace estimate can be written as
  $\unfnot{\sig{\ten{\hat{U}}}}{R+1}^\trans = 
   \ma{U}_{\rm s} + 
   \unfnot{\Delta\sig{\ten{\hat{U}}}}{R+1}^\trans$,
   where 
\ifCLASSOPTIONdraftcls    
    \begin{align}
        \unfnot{\Delta\sig{\ten{\hat{U}}}}{R+1}^\trans
        & = 
           \left(
           {\ma{T}}_1 \kron {\ma{T}}_2 \kron \ldots \kron {\ma{T}}_R
           \right) \cdot \Delta \ma{{U}}_{\rm s} 
            + 
           \left(
           \left[\Delta\sig{\ma{U}}_1 \cdot \sigH{\ma{U}}_1\right]
           \kron {\ma{T}}_2 \kron \ldots \kron {\ma{T}}_R
           \right) \cdot \ma{{U}}_{\rm s} \notag \\&
            + 
           \left(
           \ma{T}_1 \kron 
           \left[\Delta\sig{\ma{U}}_2 \cdot \sigH{\ma{U}}_2\right]
           \kron \ldots \kron {\ma{T}}_R
           \right) \cdot \ma{{U}}_{\rm s} 
           + \ldots \notag \\&
           + 
           \left(
           {\ma{T}}_1
           \kron {\ma{T}}_2 \kron \ldots 
           \kron \left[\Delta\sig{\ma{U}}_R \cdot \sigH{\ma{U}}_R\right]
           \right) \cdot \ma{{U}}_{\rm s} 
            + \bigO{\Delta^2}, \label{eqn_perf_exptensub}
    \end{align}    
\else
    \begin{align}
        \unfnot{\Delta\sig{\ten{\hat{U}}}}{R+1}^\trans
        & = 
           \left(
           {\ma{T}}_1 \kron {\ma{T}}_2 \kron \ldots \kron {\ma{T}}_R
           \right) \cdot \Delta \ma{{U}}_{\rm s} \notag \\
           & + 
           \left(
           \left[\Delta\sig{\ma{U}}_1 \cdot \sigH{\ma{U}}_1\right]
           \kron {\ma{T}}_2 \kron \ldots \kron {\ma{T}}_R
           \right) \cdot \ma{{U}}_{\rm s} \notag \\
           & + 
           \left(
           \ma{T}_1 \kron 
           \left[\Delta\sig{\ma{U}}_2 \cdot \sigH{\ma{U}}_2\right]
           \kron \ldots \kron {\ma{T}}_R
           \right) \cdot \ma{{U}}_{\rm s} \notag \\
           & + \ldots  \notag \\
           & + 
           \left(
           {\ma{T}}_1
           \kron {\ma{T}}_2 \kron \ldots 
           \kron \left[\Delta\sig{\ma{U}}_R \cdot \sigH{\ma{U}}_R\right]
           \right) \cdot \ma{{U}}_{\rm s} \notag \\
           & + \bigO{\Delta^2}, \label{eqn_perf_exptensub}
    \end{align}
\fi    
    the SVD-based signal subspace perturbation 
    $\Delta \ma{U}_{\rm s}$ is given by~\eqref{eqn_perf_linexp_vacma} 
    and the perturbation of the $r$-space can be computed via
    \begin{align}
     \Delta \sig{\ma{U}}_r & = \noi{\ma{U}}_r\cdot\noi{\ma{\Gamma}}_{r} = \noi{\ma{U}}_r \cdot \noiH{\ma{U}}_r \cdot \unf{N}{r} \cdot \sig{\ma{V}}_r \cdot \siginv{\ma{\Sigma}}_r.
    \end{align}
\end{theorem}

Proof: cf. Appendix~\ref{sec_app_proof_perf_exptensub}.

Note that while $\Delta \ma{U}_{\rm s}$ in general contains both perturbation terms 
$\ma{U}_{\rm n} \cdot \ma{\Gamma}_{\rm n}$
and $\ma{U}_{\rm s}\cdot \ma{\Gamma}_{\rm s}$, for $\Delta \sig{\ma{U}}_r$ the term 
$\sig{\ma{U}}_r\cdot \sig{\ma{\Gamma}}_r$
cancels. This is not surprising since the $r$-mode subspaces enter~\eqref{eqn_perf_reltenmat2} only
via projection matrices for which the choice of the particular basis is irrelevant.

\section{Asymptotical analysis of the parameter estimation accuracy} \label{sec_perf_esprit}

\subsection{Review of perturbation results for the 1-D Standard ESPRIT}

In~\cite{LLV:93} the authors point out that once a first order expansion of the subspace
estimation error is available it can be used to find a corresponding first order
expansion of the estimation error of a suitable parameter estimation scheme.
One of the examples the authors show is the 1-D Standard ESPRIT algorithm using LS, which
we use as a starting point to discuss various ESPRIT-type algorithms in this section.
In the noise-free case,
the shift invariance equation for 1-D Standard ESPRIT can be expressed
as
\begin{align}
    \ma{J}_1 \cdot \ma{U}_{\rm s} \cdot \ma{\Psi} = \ma{J}_2 \cdot \ma{U}_{\rm s} \label{eqn_perf_esprit_1dsie}
\end{align}
where $\ma{J}_1, \ma{J}_2 \in \real^{\subsel{M} \times M}$
are the selection matrices that select the $\subsel{M}$ elements from the $M$ antenna
elements which correspond to the first and the second subarray, respectively. 
Moreover, $\ma{\Psi} = \ma{Q} \cdot \ma{\Phi} \cdot \ma{Q}^{-1}$, where $\ma{\Phi}
= \diagof{\begin{bmatrix} \expof{\j \mu_1}, & \ldots, & \expof{\j \mu_d}\end{bmatrix}}
\in \compl^{d \times d}$ contains the spatial frequencies $\mu_k$, $k=1, 2, \ldots, d$
that we want to estimate. Therefore, $\mu_k = \argof{\eigvof{\ma{\Psi}}{k}}$, i.e.,
the $k$-th spatial frequency is obtained from the phase of the $k$-th eigenvalue 
($\eigvof{\cdot}{k}$) of $\ma{\Psi}$.

In presence of noise,
we only have an estimate $\ma{\hat{U}}_{\rm s}$ of the signal subspace $\ma{U}_{\rm s}$.
Consequently, \eqref{eqn_perf_esprit_1dsie} does in general not have an exact
solution anymore. A simple way of finding an approximate $\ma{\hat{\Psi}}$ is given by the LS
solution which can be expressed as
\begin{align}
    \ma{\hat{\Psi}}_{\rm LS} = \left(\ma{J}_1 \cdot \ma{\hat{U}}_{\rm s}\right)^+
    \cdot \ma{J}_2 \cdot \ma{\hat{U}}_{\rm s} \label{eqn_perf_esprit_1dlssol}
\end{align}
To simplify the notation we skip the index ``LS'' for the remainder of this section
(since only LS is considered) and pick it up again in the next section where we
expand the discussion to SLS.

For the estimation error of the $k$-th spatial frequency corresponding to the LS solution
from~\eqref{eqn_perf_esprit_1dlssol}, \cite{LLV:93} provides the following expansion
\ifCLASSOPTIONdraftcls
  \begin{align}
    \Delta \mu_k = \imagof{
       \ma{p}_k^\trans \cdot \left( \ma{J}_1 \cdot \ma{U}_{\rm s} \right)^+
       \cdot
       \left[ \ma{J}_2 / \lambda_k -\ma{J}_1 \right]
       \cdot
       \Delta \ma{U}_{\rm s} \cdot \ma{q}_k
    } + \bigO{\Delta^2} \label{eqn_perf_1dse_expl}
  \end{align}
\else
  \begin{align}
    \Delta \mu_k = & {\rm Im}\Big\{
       \ma{p}_k^\trans \cdot \left( \ma{J}_1 \cdot \ma{U}_{\rm s} \right)^+
       \cdot
       \left[ \ma{J}_2 / \lambda_k -\ma{J}_1 \right]
       \notag \\ &
       \cdot
       \Delta \ma{U}_{\rm s} \cdot \ma{q}_k
    \Big\} + \bigO{\Delta^2} \label{eqn_perf_1dse_expl}
  \end{align}
\fi
where $\lambda_k = \expof{\j \mu_k}$ and $\ma{q}_k$ is the $k$-th column of $\ma{Q}$.
Moreover, $\ma{p}_k^\trans$ represents the $k$-th row vector of the matrix $\ma{P} = \ma{Q}^{-1}$.
Note that $\Delta \ma{U}_{\rm s}$ can be expanded in terms of the perturbation term $\ma{N}$
by directly using the expansion~\eqref{eqn_perf_linexp_vac}. The additional term from~\eqref{eqn_perf_linexp_vacma}
is not needed as it is irrelevant for the performance of ESPRIT.


\subsection{\texorpdfstring{Extension to $R$-D Standard (Tensor-)ESPRIT}{Extension to R-D Standard (Tensor-)ESPRIT}}

The previous result from~\cite{LLV:93} on the first order perturbation expansion
of 1-D Standard ESPRIT using LS is easily generalized to the $R$-D case. The reason
is that for $R$-D LS-based ESPRIT, the $R$ shift invariance equations are solved
independently from each other\footnote{If the shift invariance equations are solved
completely independently, the correct pairing of the parameters across dimensions
has to be found in a subsequent step. This is often avoided by computing the LS solutions
for $\ma{\Psi}^{(r)}$ independently but then performing a joint eigendecomposition
of all $R$ dimensions to yield $\ma{\Phi}^{(r)}$. This step is not included in the performance
analysis presented in this section, since no performance results on joint eigendecompositions
are available and it appears to be a very difficult task. Moreover, 
this step has indeed no impact on the asymptotic estimation error of the spatial frequencies 
for high SNRs since the eigenvectors become asymptotically equal.
As shown in~\cite{LT:78}, the impact of the perturbation of the eigenvectors
is of second-order and can hence be ignored in a first-order perturbation analysis.}.
Hence, the arguments from~\cite{LLV:93} are readily
applied to all modes individually and we directly obtain a first order expansion
for the estimation error of the $k$-th spatial frequency in the $r$-th mode
\ifCLASSOPTIONdraftcls
  \begin{align}
    \Delta \mu_k^{(r)} = \imagof{
       \ma{p}_k^\trans \cdot \left( \ma{\tilde{J}}_1^{(r)} \cdot \ma{U}_{\rm s} \right)^+
       \cdot
       \left[ \ma{\tilde{J}}_2^{(r)} / \lambda_k^{(r)} -\ma{\tilde{J}}_1^{(r)} \right]
       \cdot
       \Delta \ma{U}_{\rm s} \cdot \ma{q}_k
    } + \bigO{\Delta^2} \label{eqn_perf_expl_rdse}
  \end{align}
\else
  \begin{align}
    \Delta \mu_k^{(r)} = & {\rm Im}\Big\{
       \ma{p}_k^\trans \cdot \left( \ma{\tilde{J}}_1^{(r)} \cdot \ma{U}_{\rm s} \right)^+
       \cdot
       \left[ \ma{\tilde{J}}_2^{(r)} / \lambda_k^{(r)} -\ma{\tilde{J}}_1^{(r)} \right]
       \notag \\ & 
       \cdot
       \Delta \ma{U}_{\rm s} \cdot \ma{q}_k
    \Big\} + \bigO{\Delta^2} \label{eqn_perf_expl_rdse}
  \end{align}
\fi
where $\ma{\tilde{J}}_1^{(r)}, \ma{\tilde{J}}_2^{(r)} \in \real^{\frac{M}{M_r} \cdot \subsel{M_r} \times M}$
are the effective $R$-D selection matrix for the 
first and the second subarray in the $r$-th mode, respectively.
They can be expressed as $\ma{\tilde{J}}_\ell^{(r)} = \ma{I}_{\prod_{n=1}^{r-1} M_n} \kron 
\ma{J}_\ell^{(r)} \kron \ma{I}_{\prod_{n=r+1}^R M_n}$, for $\ell=1,2$ and $r=1, 2, \ldots, R$, where 
$\ma{J}_\ell^{(r)} \in \real^{\subsel{M_r} \times M_r}$ are the selection matrices which
select the $\subsel{M_r}$ elements belonging to the first and the second subarray in the $r$-th mode,
respectively.

Since this expansion for $R$-D Standard ESPRIT is explicit in the perturbation of the subspace
estimate and $R$-D Standard Tensor-ESPRIT only differs in the fact that it uses the enhanced
HOSVD-based subspace estimate, we immediately conclude that a first order perturbation
expansion for $R$-D Standard Tensor-ESPRIT is given by
\ifCLASSOPTIONdraftcls
  \begin{align}
    \Delta \mu_k^{(r)} = \imagof{
       \ma{p}_k^\trans \cdot \left( \ma{\tilde{J}}_1^{(r)} \cdot \ma{U}_{\rm s} \right)^+
       \cdot
       \left[ \ma{\tilde{J}}_2^{(r)} / \lambda_k^{(r)} -\ma{\tilde{J}}_1^{(r)} \right]
       \cdot
       \unfnot{\Delta\sig{\ten{\hat{U}}}}{R+1}^\trans \cdot \ma{q}_k
    } + \bigO{\Delta^2}. \label{eqn_perf_epxl_ste}
  \end{align}
\else
  \begin{align}
    \Delta \mu_k^{(r)} = & {\rm Im}\Big\{
       \ma{p}_k^\trans \cdot \left( \ma{\tilde{J}}_1^{(r)} \cdot \ma{U}_{\rm s} \right)^+
       \cdot
       \left[ \ma{\tilde{J}}_2^{(r)} / \lambda_k^{(r)} -\ma{\tilde{J}}_1^{(r)} \right]
       \notag \\ &
       \cdot
       \unfnot{\Delta\sig{\ten{\hat{U}}}}{R+1}^\trans \cdot \ma{q}_k
    \Big\} + \bigO{\Delta^2}. \label{eqn_perf_epxl_ste}
  \end{align}
\fi
An explicit expansion of $\Delta \mu_k^{(r)}$
in terms of the noise tensor $\ten{N}$ is obtained
by inserting the previous result~\eqref{eqn_perf_exptensub}.

\subsection{Mean Square Errors}

As it has been said above, the advantage of the first order perturbation expansion
we have discussed so far is that it is explicit in the perturbation term $\ma{N}$
and hence makes no assumptions about its distribution. 
However, it is often also desirable to know the mean square error if a specific
distribution is assumed and the ensemble average over all possible noise realizations
is computed.

\revA{In the sequel we show that the mean square error only depends on the second-order
moments of the noise samples. Hence, we can derive the MSE as a function of the covariance
matrix and the pseudo-covariance matrix only assuming the noise to be zero mean.
We neither need to assume Gaussianity nor circular symmetry.}
For simplicity we consider the special case $R=2$
for Standard Tensor-ESPRIT, however,
a generalization to a larger number of dimensions is quite straightforward. 

\begin{theorem} \label{thm_perf_mse}
	  Assume that the entries of the perturbation term $\ma{N}$ or $\ten{N}$ are
	  \revA{
	  zero mean random variables with finite second-order moments described
	  by the covariance matrix 
	  $\ma{R}_{\rm nn} = \expvof{\ma{n} \cdot \ma{n}^\herm}$
	  and the complementary covariance matrix
	  $\ma{C}_{\rm nn} = \expvof{\ma{n} \cdot \ma{n}^\trans}$
	  for $\ma{n} = \vecof{\ma{N}} = \vecof{\unf{N}{3}^\trans}$.}
    Then, the first-order
    approximation of the mean square estimation error for the $k$-th spatial frequency in the $r$-th mode
    is given by
    \revA{
\ifCLASSOPTIONdraftcls
    \begin{align}
          \expvof{ \left(\Delta \mu_k^{(r)}\right)^2}
   & = \frac{1}{2} \Big(\ma{r}_k^{(r)^\herm} \cdot \ma{W}_{\rm mat}^\conj \cdot \ma{R}_{\rm nn}^\trans \cdot
   \ma{W}_{\rm mat}^\trans \cdot \ma{r}_k^{(r)} \notag \\ &
    - \realof{\ma{r}_k^{(r)^\trans} \cdot \ma{W}_{\rm mat} \cdot \ma{C}_{\rm nn} \cdot 
    \ma{W}_{\rm mat}^\trans \cdot \ma{r}_k^{(r)} }
   \Big)
      + \bigO{\traceof{\ma{R}_{\rm nn}}^2}
   \label{eqn_subsp_perf_mse_se}
    \end{align}
\else
	\begin{align}
          &\expvof{ \left(\Delta \mu_k^{(r)}\right)^2}
   = \frac{1}{2} \Big(\ma{r}_k^{(r)^\herm} \cdot \ma{W}_{\rm mat}^\conj \cdot \ma{R}_{\rm nn}^\trans \cdot
   \ma{W}_{\rm mat}^\trans \cdot \ma{r}_k^{(r)} - \notag \\ &
    \realof{\ma{r}_k^{(r)^\trans} \cdot \ma{W}_{\rm mat} \cdot \ma{C}_{\rm nn} \cdot 
    \ma{W}_{\rm mat}^\trans \cdot \ma{r}_k^{(r)} }
   \Big)
      + \bigO{\traceof{\ma{R}_{\rm nn}}^2}
   \label{eqn_subsp_perf_mse_se}
    \end{align}
\fi}
		for $R$-D Standard ESPRIT and
    \revA{
\ifCLASSOPTIONdraftcls    
    \begin{align}
          \expvof{ \left(\Delta \mu_k^{(r)}\right)^2}
   = & \frac{1}{2} \Big(\ma{r}_k^{(r)^\herm} \cdot \ma{W}_{\rm ten}^\conj \cdot \ma{R}_{\rm nn}^\trans \cdot
   \ma{W}_{\rm ten}^\trans \cdot \ma{r}_k^{(r)} \notag \\ &
    - \realof{\ma{r}_k^{(r)^\trans} \cdot \ma{W}_{\rm ten} \cdot \ma{C}_{\rm nn} \cdot 
    \ma{W}_{\rm ten}^\trans \cdot \ma{r}_k^{(r)} }
   \Big)
      + \bigO{\traceof{\ma{R}_{\rm nn}}^2}
    \label{eqn_subsp_perf_mse_ste}
    \end{align}
\else
   \begin{align}
       &  \expvof{ \left(\Delta \mu_k^{(r)}\right)^2}
   =  \frac{1}{2} \Big(\ma{r}_k^{(r)^\herm} \cdot \ma{W}_{\rm ten}^\conj \cdot \ma{R}_{\rm nn}^\trans \cdot
   \ma{W}_{\rm ten}^\trans \cdot \ma{r}_k^{(r)} - \notag \\ &
    \realof{\ma{r}_k^{(r)^\trans} \cdot \ma{W}_{\rm ten} \cdot \ma{C}_{\rm nn} \cdot 
    \ma{W}_{\rm ten}^\trans \cdot \ma{r}_k^{(r)} }
   \Big)
      + \bigO{\traceof{\ma{R}_{\rm nn}}^2}
    \label{eqn_subsp_perf_mse_ste}
    \end{align}
\fi}
		%
    for 2-D Standard Tensor-ESPRIT\footnote{The reason that this result is specific for $R=2$
    is not~\eqref{eqn_subsp_perf_mse_ste} (which applies to arbitrary $R$) but~\eqref{eqn_subsp_perf_wten_r2}
    which we develop only for $R=2$ here.}. The vector $\ma{r}_k^{(r)}$ and the matrices 
    $\ma{W}_{\rm mat}$ and $\ma{W}_{\rm ten}$ are given by
\ifCLASSOPTIONdraftcls
    \begin{align}
         \ma{r}_k^{(r)} &  = \ma{q}_k \kron 
         \left(\left[ \left(\ma{\tilde{J}}_1^{(r)}  \ma{U}_{\rm s} \right)^+ 
                \left(\ma{\tilde{J}}_2^{(r)}/\expof{\j \cdot \mu_k^{(r)}} - \ma{\tilde{J}}_1^{(r)}\right)
         \right]^\trans \cdot \ma{p}_k\right) \notag \\ 
        \ma{W}_{\rm mat}  & = 
   \left(\ma{\Sigma}_{\rm s}^{-1} \cdot \ma{V}_{\rm s}^\trans \right) \kron \left( {\ma{U}}_{\rm n}  \cdot {\ma{U}}_{\rm n}^\herm \right) \label{eqn_subsp_perf_wmat}
   \\
    \ma{W}_{\rm ten}
    & = 
    \left( \siginv{\ma{\Sigma}_3} \sigH{\ma{U}}_3 \right)\kron \left(\left[ \ma{T}_1 \kron \ma{T}_2\right]
    \noiC{\ma{V}_3} \noiT{\ma{V}_3} \right) \notag \\
    &
     + \left( \ma{U}_{\rm s}^\trans \kron \ma{I}_M\right) \ma{\bar{T}}_2 
        \left( \sigC{\ma{U}}_1 \siginv{\ma{\Sigma}_1} \sigT{\ma{V}}_1
        \kron \noi{\ma{U}_1} \noiH{\ma{U}_1} \right)
        \cdot \ma{K}_{M_2 \times (M_1 \cdot N)} \notag \\
    &
     + \left( \ma{U}_{\rm s}^\trans \kron \ma{I}_M\right) \ma{\bar{T}}_1 
        \left( \sigC{\ma{U}}_2 \siginv{\ma{\Sigma}_2} \sigT{\ma{V}}_2
        \kron \noi{\ma{U}_2} \noiH{\ma{U}_2} \right) \quad \mbox{where} \label{eqn_subsp_perf_wten_r2}\\
\ma{\bar{T}}_1  & = 
\left[
\begin{array}{c}
\ma{I}_{M_2} \kron \ma{t}_{1,1}  \\
\vdots \\
\ma{I}_{M_2} \kron \ma{t}_{1,M_1} 
\end{array}
\right] \kron \ma{I}_{M_2},   \; \ma{\bar{T}}_2  = 
\ma{I}_{M_1} \kron \left[
\begin{array}{c}
\ma{I}_{M_1} \kron \ma{t}_{2,1}  \\
\vdots \\
\ma{I}_{M_1} \kron \ma{t}_{2,M_2}
\end{array}
\right], \notag
    \end{align}
\else
{\small
	 \begin{align}
        & \ma{r}_k^{(r)} = \ma{q}_k \kron 
         \left(\left[ \left(\ma{\tilde{J}}_1^{(r)}  \ma{U}_{\rm s} \right)^+ 
                \left(\ma{\tilde{J}}_2^{(r)}/\expof{\j \cdot \mu_k^{(r)}} - \ma{\tilde{J}}_1^{(r)}\right)
         \right]^\trans \cdot \ma{p}_k\right) \notag \\ 
        & \ma{W}_{\rm mat}  = 
   \left(\ma{\Sigma}_{\rm s}^{-1} \cdot \ma{V}_{\rm s}^\trans \right) \kron \left( {\ma{U}}_{\rm n}  \cdot {\ma{U}}_{\rm n}^\herm \right) \label{eqn_subsp_perf_wmat}
   \\
   & \ma{W}_{\rm ten}
     = 
    \left( \siginv{\ma{\Sigma}_3} \sigH{\ma{U}}_3 \right)\kron \left(\left[ \ma{T}_1 \kron \ma{T}_2\right]
    \noiC{\ma{V}_3} \noiT{\ma{V}_3} \right) \notag \\
    &
     + \left( \ma{U}_{\rm s}^\trans \kron \ma{I}_M\right) \ma{\bar{T}}_2 
        \left( \sigC{\ma{U}}_1 \siginv{\ma{\Sigma}_1} \sigT{\ma{V}}_1
        \kron \noi{\ma{U}_1} \noiH{\ma{U}_1} \right)
        \cdot \ma{K}_{M_2 \times (M_1 \cdot N)} \notag \\
    &
     + \left( \ma{U}_{\rm s}^\trans \kron \ma{I}_M\right) \ma{\bar{T}}_1 
        \left( \sigC{\ma{U}}_2 \siginv{\ma{\Sigma}_2} \sigT{\ma{V}}_2
        \kron \noi{\ma{U}_2} \noiH{\ma{U}_2} \right) \label{eqn_subsp_perf_wten_r2} \\ 
        & \ma{\bar{T}}_1   = 
\begin{bmatrix}\ma{I}_{M_2} \kron \ma{t}_{1,1}  \\
\vdots \\
\ma{I}_{M_2} \kron \ma{t}_{1,M_1} 
\end{bmatrix} \kron \ma{I}_{M_2},  
 \ma{\bar{T}}_2  = 
\ma{I}_{M_1} \kron \begin{bmatrix}
\ma{I}_{M_1} \kron \ma{t}_{2,1}  \\
\vdots \\
\ma{I}_{M_1} \kron \ma{t}_{2,M_2} \end{bmatrix}, \notag
    \end{align}
    }
\fi
    and $\ma{t}_{r,m}$ is the $m$-th column of $\ma{T}_r$. Finally, $\ma{K}_{p,q}$ is
    the commutation matrix from~\eqref{eqn_def_commat}.
\end{theorem}

Proof: cf. Appendix~\ref{sec_app_proof_perf_mse}.

Note that for 1-D Standard ESPRIT, this MSE expression agrees with the one
shown in~\cite{LV:92}. However,~\cite{LV:92} does not directly generalize
to the tensor case. This is the advantage of the MSE 
expressions~\eqref{eqn_subsp_perf_mse_se} and~\eqref{eqn_subsp_perf_mse_ste}
 where we only need to replace $\ma{W}_{\rm mat}$
by $\ma{W}_{\rm ten}$ to account for the enhanced signal subspace estimate. 
\revA{Furthermore, note that the special case of circularly symmetric
white noise corresponds to $\ma{R}_{\rm nn} = \sigma_{\rm n}^2 \cdot \ma{I}_{MN}$
and $\ma{C}_{\rm nn}  = \ma{0}_{MN \times MN}$.}

\subsection{Incorporation of Forward-Backward-Averaging}\label{subsec_subsp_perf_fba}

So far we have shown the explicit expansion and the MSE expressions for $R$-D Standard
ESPRIT and $R$-D Standard Tensor-ESPRIT. In order to extend these results to Unitary-ESPRIT-type
algorithms we need to incorporate the mandatory preprocessing for Unitary ESPRIT
which is given by Forward-Backward-Averaging. The second step in Unitary ESPRIT is the
transformation on the real-valued domain. However, it can be shown that this step has no impact on the performance
for high SNRs. Therefore, the asymptotic performance of Unitary-ESPRIT-type algorithms
is found once Forward-Backward-Averaging is taken into account.


Forward-Backward-Averaging augments the $N$ observations of the sampled $R$-D harmonics
by $N$ new ``virtual'' observations which are a conjugated and row- as well as column-flipped version
of the original ones~\cite{HN:95}. This can be expressed in matrix form as
%
\begin{align}
   \fba{\ma{X}} = \begin{bmatrix} \ma{X} & \ma{\Pi}_M \cdot \ma{X}^\conj \cdot \ma{\Pi}_N \end{bmatrix}
   \in \compl^{M \times 2N} \label{eqn_fba_directdata}
\end{align}
Inserting $\ma{X} = \ma{X}_0 + \ma{N}$ we find
\ifCLASSOPTIONdraftcls
\begin{align}
   \fba{\ma{X}} & = \begin{bmatrix} \ma{X}_0, & \ma{\Pi}_M \cdot \ma{X}_0^\conj \cdot \ma{\Pi}_N \end{bmatrix}
   + \begin{bmatrix} \ma{N}, & \ma{\Pi}_M \cdot \ma{N}^\conj \cdot \ma{\Pi}_N \end{bmatrix} 
   = \fba{\ma{X}_0} + \fba{\ma{N}}. \label{eqn_fba_directdata_exp}
\end{align}
\else
\begin{align}
   \fba{\ma{X}} & = \begin{bmatrix} \ma{X}_0, & \ma{\Pi}_M \cdot \ma{X}_0^\conj \cdot \ma{\Pi}_N \end{bmatrix}
   + \begin{bmatrix} \ma{N}, & \ma{\Pi}_M \cdot \ma{N}^\conj \cdot \ma{\Pi}_N \end{bmatrix} \notag \\
   & = \fba{\ma{X}_0} + \fba{\ma{N}}. \label{eqn_fba_directdata_exp}
\end{align}
\fi
%
%
However, the latter relation shows that we are interested in the perturbation of the subspace
of a matrix $\fba{\ma{X}_0}$ superimposed by an additive perturbation $\fba{\ma{N}}$, which is small.
Since the explicit perturbation expansion we have used up to this point requires no additional
assumptions, the surprisingly simple answer is that we do not need to change anything but we
can apply the previous results directly. All we need to do is to replace all exact (noise-free)
subspaces of $\ma{X}_0$ by the corresponding subspaces of $\fba{\ma{X}_0}$.
{From~\eqref{eqn_perf_expl_rdse},} we immediately obtain the following explicit first-order expansion which is valid
for $R$-D Unitary ESPRIT
 \ifCLASSOPTIONdraftcls
\begin{align}
    \Delta \mu_k^{(r)} = \imagof{
       \fbaT{\ma{p}_k} \cdot \left( \ma{\tilde{J}}_1^{(r)} \cdot \fba{\ma{U}}_{\rm s} \right)^+
       \cdot
       \left[ \ma{\tilde{J}}_2^{(r)} / \lambda_k^{(r)} -\ma{\tilde{J}}_1^{(r)} \right]
       \cdot
       \Delta \fba{\ma{U}_{\rm s}} \cdot \fba{\ma{q}_k}
    } + \bigO{\Delta^2}
\end{align}
 \else
\begin{align}
    \Delta \mu_k^{(r)} = & {\rm Im}\Big\{
       \fbaT{\ma{p}_k} \cdot \left( \ma{\tilde{J}}_1^{(r)} \cdot \fba{\ma{U}}_{\rm s} \right)^+
       \cdot
       \left[ \ma{\tilde{J}}_2^{(r)} / \lambda_k^{(r)} -\ma{\tilde{J}}_1^{(r)} \right] \\ \notag
       &
       \cdot
       \Delta \fba{\ma{U}_{\rm s}} \cdot \fba{\ma{q}_k}
    \Big\} + \bigO{\Delta^2}
\end{align} 
 \fi
where $\Delta \fba{\ma{U}_{\rm s}}$ is given by 
\begin{align}
    \Delta \fba{\ma{U}_{\rm s}} = \fba{\ma{U}}_{\rm n} \cdot \fbaH{\ma{U}}_{\rm n}
   \cdot \fba{\ma{N}} \cdot \fba{\ma{V}}_{\rm s} \cdot \fbainv{\ma{\Sigma}}_{\rm s}
\end{align}
and $\fba{\ma{U}_{\rm s}}$, $\fba{\ma{U}_{\rm n}}$, $\fba{\ma{V}}_{\rm s}$, $\fba{\ma{\Sigma}}_{\rm s}$
correspond to the signal subspace, the noise subspace, the row space, and the singular values of $\fba{\ma{X}}_0$, 
respectively. Likewise, $\fba{\ma{q}_k}$ and $\fba{\ma{p}_k}$ 
represent the corresponding versions of $\ma{q}_k$ and $\ma{p}_k$ if ${\ma{U}}_{\rm s}$
is replaced by $\fba{\ma{U}}_{\rm s}$ in the shift invariance equations.

With the same reasoning, an explicit expansion for $R$-D Unitary Tensor-ESPRIT is obtained
by consistently replacing $\ten{X}_0$ by $\fba{\ten{X}}_0$ in~\eqref{eqn_perf_epxl_ste}, i.e.,
 \ifCLASSOPTIONdraftcls
\begin{align}
    \Delta \mu_k^{(r)} = \imagof{
       \fbaT{\ma{p}_k} \cdot \left( \ma{\tilde{J}}_1^{(r)} \cdot \fba{\ma{U}}_{\rm s} \right)^+
       \cdot
       \left[ \ma{\tilde{J}}_2^{(r)} / \lambda_k^{(r)} -\ma{\tilde{J}}_1^{(r)} \right]
       \cdot
       \unfnot{\Delta\sigfba{\ten{\hat{U}}}}{R+1}^\trans \cdot \fba{\ma{q}_k}
    } + \bigO{\Delta^2}. \label{eqn_perf_epxl_ute}
\end{align}
 \else
\begin{align}
    \Delta \mu_k^{(r)} = &{\rm Im}\Big\{
       \fbaT{\ma{p}_k} \cdot \left( \ma{\tilde{J}}_1^{(r)} \cdot \fba{\ma{U}}_{\rm s} \right)^+
       \cdot
       \left[ \ma{\tilde{J}}_2^{(r)} / \lambda_k^{(r)} -\ma{\tilde{J}}_1^{(r)} \right]
       \notag \\
       & 
       \cdot
       \unfnot{\Delta\sigfba{\ten{\hat{U}}}}{R+1}^\trans \cdot \fba{\ma{q}_k}
    \Big\} + \bigO{\Delta^2}. \label{eqn_perf_epxl_ute}
\end{align} 
 \fi
%

\revA{Similarly, Theorem~\ref{thm_perf_mse} can be applied to compute the MSE since
we only assumed the noise to be zero mean and possess finite second order moments,
which is still true after forward-backward
averaging. The following theorem summarizes the results for $R$-D Unitary ESPRIT
and 2-D Unitary Tensor-ESPRIT:}

\begin{theorem}\label{thm_perf_mse_fba}
    For the case where $\ma{N}$ or $\ten{N}$ contain
\revA{
	  zero mean random variables with finite second-order moments described
	  by the covariance matrix 
	  $\ma{R}_{\rm nn} = \expvof{\ma{n} \cdot \ma{n}^\herm}$
	  and the complementary covariance matrix
	  $\ma{C}_{\rm nn} = \expvof{\ma{n} \cdot \ma{n}^\trans}$
	  for $\ma{n} = \vecof{\ma{N}} = \vecof{\unf{N}{3}^\trans}$,}
%
%
	  \revA{the MSE for $R$-D Unitary ESPRIT and 2-D Unitary Tensor-ESPRIT
	  are given by \eqref{eqn_subsp_perf_mse_se}
	  and \eqref{eqn_subsp_perf_mse_ste} if we replace 
		$\ma{r}_k^{{(r)}}$ by $\ma{r}_k^{\fba{(r)}}$,
		${\ma{W}}_{\rm mat}$ by $\fba{\ma{W}}_{\rm mat}$,
		${\ma{W}}_{\rm ten}$ by $\fba{\ma{W}}_{\rm ten}$,
		and $\ma{R}_{\rm nn}$ as well as $\ma{C}_{\rm nn}$
		by $\fba{\ma{R}_{\rm nn}}$ and $\fba{\ma{C}_{\rm nn}}$.}
		Here, 
     $\ma{r}_k^{\fba{(r)}}$, $\fba{\ma{W}}_{\rm mat}$, and $\fba{\ma{W}}_{\rm ten}$
    are computed as in~\eqref{eqn_subsp_perf_mse_se} and~\eqref{eqn_subsp_perf_mse_ste}
    by consistently replacing
    all quantities by their forward-backward-averaged equivalents.
    \revA{Moreover, $\fba{\ma{R}_{\rm nn}}$ and $\fba{\ma{C}_{\rm nn}}$ represent the covariance
    and the pseudo-covariance matrix of the forward-backward averaged noise, which are
    given by
\ifCLASSOPTIONdraftcls    
    \begin{align}
       \fba{\ma{R}_{\rm nn}}
       =
       \begin{bmatrix}
           \ma{R}_{\rm nn} & \ma{C}_{\rm nn} \cdot \ma{\Pi}_{MN} \\
           \ma{\Pi}_{MN} \cdot \ma{C}_{\rm nn}^\conj & 
               \ma{\Pi}_{MN} \cdot \ma{R}_{\rm nn}^\conj \cdot \ma{\Pi}_{MN}
       \end{bmatrix},
       \quad
       \fba{\ma{C}_{\rm nn}}
       =
       \begin{bmatrix}
           \ma{C}_{\rm nn} & \ma{R}_{\rm nn} \cdot \ma{\Pi}_{MN} \\
           \ma{\Pi}_{MN} \cdot \ma{R}_{\rm nn}^\conj & 
               \ma{\Pi}_{MN} \cdot \ma{C}_{\rm nn}^\conj \cdot \ma{\Pi}_{MN}
       \end{bmatrix}.
       \notag
    \end{align}
\else
    \begin{align}
       \fba{\ma{R}_{\rm nn}}
       &=
       \begin{bmatrix}
           \ma{R}_{\rm nn} & \ma{C}_{\rm nn} \cdot \ma{\Pi}_{MN} \\
           \ma{\Pi}_{MN} \cdot \ma{C}_{\rm nn}^\conj & 
               \ma{\Pi}_{MN} \cdot \ma{R}_{\rm nn}^\conj \cdot \ma{\Pi}_{MN}
       \end{bmatrix}
       \notag \\
       \fba{\ma{C}_{\rm nn}}
       &=
       \begin{bmatrix}
           \ma{C}_{\rm nn} & \ma{R}_{\rm nn} \cdot \ma{\Pi}_{MN} \\
           \ma{\Pi}_{MN} \cdot \ma{R}_{\rm nn}^\conj & 
               \ma{\Pi}_{MN} \cdot \ma{C}_{\rm nn}^\conj \cdot \ma{\Pi}_{MN}
       \end{bmatrix}.
       \notag
    \end{align}
\fi    
    }
\end{theorem}    

Proof: cf. Appendix~\ref{sec_app_proof_perf_mse_fba}.

\revA{Note that in the special case where the noise is circularly symmetric and white
we have $\fba{\ma{R}_{\rm nn}} = \sigma_{\rm n}^2 \cdot \ma{I}_{2MN}$
and $\fba{\ma{C}_{\rm nn}} = \sigma_{\rm n}^2 \cdot \ma{\Pi}_{2MN}$.}

{It is important to note that results on Unitary ESPRIT in this section
relate to the LS solution only. If {TLS} is used instead, the equivalence
of Standard ESPRIT with Forward-Backward-Averaging and Unitary ESPRIT is shown
in~\cite{HN:95}.}

\subsection{Extension to other ESPRIT-type algorithms}\label{subsec_subsp_perf_nc}

In a similar manner as in the previous section, other ESPRIT-type algorithms can be
analyzed. For instance, the NC Standard ESPRIT  and NC Unitary ESPRIT algorithm
for strict-sense non-circular sources are based on a different kind of preprocessing
where instead of augmenting the columns we augment the rows of the measurement matrix.
Yet, the explicit first order perturbation expansion still applies since the result
can be written as a noise-free (augmented) measurement matrix superimposed by a small (augmented)
perturbation matrix. Consequently, for the explicit expansion we only need to consistently
replace the quantities originating from the SVD of $\ma{X}_0$ by the corresponding quantities
from the appropriately preprocessed measurement matrix $\ma{X}_0^{({\rm nc})}$.
\revB{Likewise, the MSE expressions are directly applicable since we only require
the noise to be zero mean and possess finite second-order moments.}

{Another possible extension is to incorporate spatial smoothing. 
If sources are mutually
coherent, preprocessing must be applied to the data to decorrelate the sources prior
to any subspace-based parameter estimation scheme. 
Via Forward-Backward-Averaging, two sources can be decorrelated. However, if more than
two sources are coherent (or if FBA cannot be applied), additional preprocessing is needed. 
For spatial smoothing
we divide the array into a number of identical displaced subarrays and average the spatial covariance matrix
over these subarrays. Since the number of subarrays we choose is a design parameter
{influencing the performance},
investigating its effect by virtue of an analytical performance assessment would
be desirable. Note that the spatial averaging 
introduces a correlation into the noise. Therefore, the presented framework is particularly attractive
since for the explicit expansion, no assumptions about the noise statistics are needed.
A further extension is the performance assessment of tensor-based schemes for spatial smoothing.
{We have introduced a tensor-based formulation of spatial smoothing for $R$-D signals in \cite{HRD:08}.
Moreover, a tensor-based spatial smoothing technique for 1-D damped and undamped
harmonic retrieval with a single snapshot is shown in \cite{THG:09}.
The extension to multiple snapshots is introduced in \cite{THRG:10}
{and an $R$-D extension is shown in \cite{THG:09b}}.
A major advantage of \cite{THRG:10,THG:09b} is that the performance of the ESPRIT-type
parameter estimates is almost independent of the choice of the subarray size.
This {could} be verified by analytical results if the performance analysis
is extended accordingly.}
}

\subsection{Incorporation of Structured Least Squares (SLS)}\label{sec_perf_sls}

So far, all performance results are based on ESPRIT using LS, i.e., the
overdetermined shift invariance equations are solved using LS only. 
However, the LS solution to the shift invariance
equation is in general suboptimal as errors on both sides of the equations
need to be taken into account. Even more so, since for overlapping subarrays, the
shift invariance equation has a specific structure resulting in common error terms
on both sides of the equations, this structure should be taken into account when
solving them. This has led to the development of the SLS algorithm \cite{Haa:97}.
Since it has been shown that the resulting ESPRIT algorithm using SLS
outperforms ESPRIT using LS and TLS
for overlapping subarrays~\cite{Haa:97},
it is desirable to extend our performance analysis results to SLS-based ESPRIT
as well.

Due to the fact that our analysis is asymptotic in the SNR we can make the following simplifying assumptions
for SLS. Firstly, we consider only a single iteration, as proposed in~\cite{Haa:97}.
This is optimal for high SNRs, since the underlying cost function
is quadratic but actually asymptotically linear (the quadratic term vanishes against the linear
terms for high SNRs). Secondly, we do not consider the optional regularization term in SLS
(i.e., we set the corresponding regularization parameter $\alpha$ to infinity) as regularization
is typically not needed for high SNRs.

Under these conditions we can show the following theorem:

\begin{theorem} \label{thm_perf_mse_sls}
  A first order expansion of the estimation error of 1-D Standard ESPRIT using SLS
 is given by
    \begin{align}
        \Delta \mu_{k,{\rm SLS}} & = \imagof{\ma{r}_{k,{\rm SLS}}^\trans \cdot \vecof{\Delta \ma{U}_{\rm s}}} + \bigO{\Delta^2} \label{eqn_subsp_perf_sls_u} \\
        & = \imagof{\ma{r}_{k,{\rm SLS}}^\trans \cdot \ma{W}_{\rm mat} \cdot \vecof{\ma{N}}} + \bigO{\Delta^2}
        \label{eqn_subsp_perf_sls}
    \end{align}
    where $\ma{W}_{\rm mat}$ is defined in~\eqref{eqn_subsp_perf_wmat} 
    and $\ma{r}_{k,{\rm SLS}}^\trans$ is given by
   \ifCLASSOPTIONdraftcls
    \begin{align}
     \ma{r}_{k,{\rm SLS}}^\trans & =   \ma{q}_k^\trans \kron \left[\ma{p}_k^\trans \cdot \left(\ma{J}_1 \cdot \ma{U}_s\right)^+
                                  \cdot \left( \frac{\ma{J}_2}{\expof{\j \mu_k}} - \ma{J}_1 \right) \right] 
                                  \notag \\ &
                       -  \left(\ma{q}_k^\trans \kron \left[ \ma{p}_k^\trans\cdot 
                       \frac{(\ma{J}_1 \cdot \ma{U}_{\rm s})^\herm}{\expof{\j \mu_k}} \right] \right)
                                     \cdot \left(\ma{F}_{\rm SLS}\cdot\ma{F}_{\rm SLS}^\herm\right)^{-1}
                                     \cdot \ma{W}_{\rm R,U} \notag \\
\ma{W}_{\rm R,U} & = 
    \left(
 \ma{\Psi}^\trans \kron \ma{J}_1 \right)
     + \ma{I}_d \kron \left(\ma{J}_1\cdot \ma{U}_{\rm s} \left(\ma{J}_1\cdot \ma{U}_{\rm s}\right)^+\cdot\ma{J}_2\right)
      - \ma{\Psi}^\trans \kron \left(\ma{J}_1\cdot \ma{U}_{\rm s} \left(\ma{J}_1\cdot \ma{U}_{\rm s}\right)^+\cdot\ma{J}_1\right)  - {\left(  \ma{I}_d \kron \ma{J}_2
        \right)} \notag \\
    \ma{F}_{\rm SLS} & = {\left[  
         \ma{I}_d \kron \left(\ma{J}_1 \cdot \ma{U}_{\rm s}\right), \;
         \left(\ma{\Psi}^{\rm T} \kron \ma{J}_1\right) - \left(\ma{I}_d \kron \ma{J}_2\right)
         \right]} \notag
    \end{align}
  \else
  \begin{align}
     \ma{r}_{k,{\rm SLS}}^\trans & =   \ma{q}_k^\trans \kron \left[\ma{p}_k^\trans \cdot \left(\ma{J}_1 \cdot \ma{U}_s\right)^+
                                  \cdot \left( \frac{\ma{J}_2}{\expof{\j \mu_k}} - \ma{J}_1 \right) \right] 
                                  \notag \\ &
                       -  \left(\ma{q}_k^\trans \kron \left[ \ma{p}_k^\trans\cdot 
                       \frac{(\ma{J}_1 \cdot \ma{U}_{\rm s})^\herm}{\expof{\j \mu_k}} \right] \right)
                                     \cdot \left(\ma{F}_{\rm SLS}\cdot\ma{F}_{\rm SLS}^\herm\right)^{-1}
                                     \cdot \ma{W}_{\rm R,U} \notag \\
\ma{W}_{\rm R,U} & = 
    \left(
 \ma{\Psi}^\trans \kron \ma{J}_1 \right)
     + \ma{I}_d \kron \left(\ma{J}_1\cdot \ma{U}_{\rm s} \left(\ma{J}_1\cdot \ma{U}_{\rm s}\right)^+\cdot\ma{J}_2\right)
     \notag \\ &
      - \ma{\Psi}^\trans \kron \left(\ma{J}_1\cdot \ma{U}_{\rm s} \left(\ma{J}_1\cdot \ma{U}_{\rm s}\right)^+\cdot\ma{J}_1\right)  - {\left(  \ma{I}_d \kron \ma{J}_2
        \right)} \notag \\
    \ma{F}_{\rm SLS} & = {\left[  
         \ma{I}_d \kron \left(\ma{J}_1 \cdot \ma{U}_{\rm s}\right), \;
         \left(\ma{\Psi}^{\rm T} \kron \ma{J}_1\right) - \left(\ma{I}_d \kron \ma{J}_2\right)
         \right]} \notag
    \end{align}
  \fi
    for $k=1, 2, \ldots, d$.
    \revA{The MSE for zero mean noise samples
    can then be computed via}
    %
     \revA{
\ifCLASSOPTIONdraftcls
    \begin{align}
          \expvof{ \left(\Delta \mu_{k,{\rm SLS}}\right)^2}
   & = \frac{1}{2} \Big(\ma{r}_{k,{\rm SLS}}^\herm \cdot \ma{W}_{\rm mat}^\conj \cdot 
   \ma{R}_{\rm nn}^\trans \cdot
   \ma{W}_{\rm mat}^\trans \cdot \ma{r}_{k,{\rm SLS}} \notag \\ &
    - \realof{\ma{r}_{k,{\rm SLS}}^\trans \cdot \ma{W}_{\rm mat} \cdot \ma{C}_{\rm nn} \cdot 
    \ma{W}_{\rm mat}^\trans \cdot \ma{r}_{k,{\rm SLS}} }
   \Big)
      + \bigO{\traceof{\ma{R}_{\rm nn}}^2}
   \label{eqn_subsp_perf_sls_mse}
    \end{align}
\else
	\begin{align}
          &\expvof{ \left(\Delta \mu_{k,{\rm SLS}}\right)^2}
   = \frac{1}{2} \Big(\ma{r}_{k,{\rm SLS}}^\herm \cdot \ma{W}_{\rm mat}^\conj \cdot 
   \ma{R}_{\rm nn}^\trans \cdot
   \ma{W}_{\rm mat}^\trans \cdot \ma{r}_{k,{\rm SLS}} - \notag \\ &
    \realof{\ma{r}_{k,{\rm SLS}}^\trans \cdot \ma{W}_{\rm mat} \cdot \ma{C}_{\rm nn} \cdot 
    \ma{W}_{\rm mat}^\trans \cdot \ma{r}_{k,{\rm SLS}} }
   \Big)
      + \bigO{\traceof{\ma{R}_{\rm nn}}^2}
   \label{eqn_subsp_perf_sls_mse}
    \end{align}
\fi}

\end{theorem}

Proof: Equation \eqref{eqn_subsp_perf_sls} is shown in Appendix~\ref{sec_app_proof_perf_sls}. Since the
explicit expansion of~\eqref{eqn_subsp_perf_sls} has the same form as the explicit expansion 
in~\eqref{eqn_perf_1dse_expl}, the MSE expression~\eqref{eqn_subsp_perf_sls_mse} is shown
analogously to Theorem~\ref{thm_perf_mse} as presented in Appendix~\ref{sec_app_proof_perf_mse}.

\subsection{Special case: Single source}

So far we have found closed-form expressions for the first-order approximate MSE
%
for different kinds of ESPRIT-type
algorithms.
As they are deterministic, they can be plotted
for varying system parameters without performing Monte-Carlo simulations and one can learn
from these plots under which conditions the performance changes how much.

However, it would be desirable to find expressions that are even more insightful. The biggest
disadvantage of the MSE expressions in their current form is that they are formulated
in terms of the subspaces of the noise-free observation matrix and not in terms of the
actual parameters with a physical significance, such as, the number of sensors or the positions
of the sources. 

Finding such a formulation in the general case seems to be impossible given the complicated
algebraic nature in which the MSE expressions depend on the physical parameters. However, it 
becomes much easier if some special cases are considered.
Therefore we present one example of such a special case in this section, namely, the case
of a single source captured by a uniform linear array (ULA)
and a uniform rectangular array (URA) \revA{and circularly symmetric white noise}.
Although this is a very trivial case, it serves as an example which types
of insights such an analytical performance assessment can provide.
For the 1-D case we have the following theorem:

\begin{theorem} \label{thm_perf_mse_singsrc_1d}
For the case of an $M$-element ULA (1-D) and a single source ($d=1$) we
can show that the mean square 
    estimation error of the spatial frequency for Standard ESPRIT and
    for Unitary ESPRIT is given by
    \begin{align}
          \expvof{ (\Delta \mu)^2} =
          \frac{1}{\hat{\rho}} \cdot \frac{1}{(M-1)^2}
          + \bigO{\frac{1}{\hat{\rho}^2}}
    \end{align}      
    Moreover, the deterministic Cram\'er-Rao Bound can be simplified into
    \begin{align}
   		{\rm CRB} = 
		\frac{1}{\hat{\rho}} \cdot \frac{6}{M \cdot (M^2-1)}  
		\end{align}
		Consequently, the asymptotic efficiency is given by
		\begin{align}
			\eta = \lim_{\hat{\rho} \rightarrow \infty} \frac{{\rm CRB}}{\expvof{ (\Delta \mu)^2} }
			= \frac{6 (M-1)}{M(M+1)}.
		\end{align}
		Here, $\hat{\rho}$ represents the effective SNR given by
    $\hat{\rho} = \frac{\hat{P}_{\rm T} \cdot N}{\sigma_{\rm n}^2}$,
    where $\hat{P}_{\rm T}$ is the empirical transmit power given by
    $\hat{P}_{\rm T} = \fronorm{\ma{S}}^2/N$ if $\ma{S}$ is the matrix
    of source symbols.
\end{theorem}

Proof: cf. Appendix~\ref{sec_app_proof_perf__mse_singsrc_1d}.

Note that~\cite{RH:89} provide an MSE expression for ESPRIT for the case of
a single source which scales with $1/M^2$ and is derived under the assumption
of high ``array SNR'' $P \cdot M / \sigma_{\rm n}^2$, i.e., it is asymptotic also in $M$.
The result presented here is accurate for small values of $M$ as well and only asymptotic
in the effective SNR $N \cdot P_{\rm T} / \sigma_{\rm n}^2$.
Also note that analytical expression for the {\em stochastic} Cram\'er-Rao Bound for
one and two sources are available in~\cite{Smi:05}.

We can simplify the MSE expression for ESPRIT using SLS shown in Section~\ref{sec_perf_sls}
in a similar manner, as shown in the following theorem:

\begin{theorem} \label{thm_perf_mse_singsrc_1d_sls}
The mean square 
    estimation error of the spatial frequency for Standard ESPRIT using SLS 
    on an $M$-element ULA is given by
    \begin{align}
	   	\expvof{ (\Delta \mu)^2} = \frac{6}{\hat{\rho}} 
	   	\cdot \frac{M^4-2M^3+24M^2-22M+23}{M(M^2+11)^2(M-1)^2}
	   	+ \bigO{\frac{1}{\hat{\rho}^2}}
	  \end{align}
	  Consequently, the asymptotic efficiency is given by
		\begin{align}
			\eta 
			& = \frac{(M^2+11)^2(M-1)}{(M+1)(M^4-2M^3+24M^2-22M+23)} \notag \\
& = \frac{M^5 - M^4 + 22 M^3 - 22 M^2 + 121 M - 121}{M^5 - M^4 + 22 M^3 + 2 M^2 + M + 23 \quad\quad\;\,}
		\end{align}
\end{theorem}

Proof: cf. Appendix~\ref{sec_app_proof_perf__mse_singsrc_1d_sls}.

Finally, for the 2-D case, we have the following result:

\begin{theorem} \label{thm_perf_mse_singsrc_2d}
For a uniform rectangular array
of $M_1 \times M_2$ sensors (2-D) and a single source, the mean square 
    estimation error of the spatial frequency for 2-D Standard ESPRIT,
    2-D Standard Tensor-ESPRIT, 2-D Unitary ESPRIT, and
    2-D Unitary Tensor-ESPRIT can be simplified into
  \ifCLASSOPTIONdraftcls
    \begin{align}
          \expvof{ (\Delta \mu^{(1)})^2 + (\Delta \mu^{(2)})^2} =
          \frac{1}{\hat{\rho}}  \cdot
   \left( \frac{1}{(M_1-1)^2 \cdot M_2}  + \frac{1}{ M_1 \cdot (M_2-1)^2} \right) 
   + \bigO{\frac{1}{\hat{\rho}^2}}
    \end{align}      
  \else
    \begin{align}
          & \expvof{ (\Delta \mu^{(1)})^2 + (\Delta \mu^{(2)})^2} \\ & = 
          \frac{1}{\hat{\rho}}  \cdot
   \left( \frac{1}{(M_1-1)^2 \cdot M_2}  + \frac{1}{ M_1 \cdot (M_2-1)^2} \right) 
   + \bigO{\frac{1}{\hat{\rho}^2}} \notag
    \end{align}      
  \fi
    Finally, the deterministic Cram\'er-Rao Bound for a URA can be written as
    \begin{align}
   		{\rm CRB} = \traceof{\ma{C}} = 
		\frac{1}{\hat{\rho}}  \cdot  \left( \frac{6}{M \cdot (M_1^2-1)} 
+ \frac{6}{M \cdot (M_2^2-1)} \right).
		\end{align}
\end{theorem}		

Proof: in Appendix~\ref{sec_app_proof_perf__mse_singsrc_2d} we derive MSE
expressions for $R$-D Standard ESPRIT, $R$-D Unitary ESPRIT, and the Cramér-Rao Bound
for the more general $R$-D case. From these, this theorem follows by
setting $R=2$. Moreover, we show the identity of $R$-D Standard Tensor-ESPRIT
and $R$-D Unitary Tensor-ESPRIT with $R$-D Standard ESPRIT for $R=2$.

These MSE expressions provide some interesting insights. Firstly, they show that
for a single source there is neither an improvement in terms of the estimation accuracy
from applying Forward-Backward-Averaging nor from the HOSVD-based subspace estimate.
This is surprising at first sight since the HOSVD-based subspace estimate itself
is more accurate also for a single source.

Moreover, they show that the asymptotic efficiency can be explicitly computed and that 
it is only a function of the array geometry, i.e., the number of sensors in the array.
Unfortunately, the outcome of this analysis is that ESPRIT-type algorithms using LS 
\revf{are asymptotically efficient for $M=2,3$ in the 1-D case and $M_1 \in [2,3]$, $M_2 \in [2,3]$
in the 2-D case. However, they}
become
less and less efficient when the number of sensors grows, in fact, for $M\rightarrow \infty$
we even have $\eta \rightarrow 0$. 
\revf{A possible
explanation for this phenomenon could be that an $M$-sensor ULA offers not only the one shift
invariance used in LS (the first and last $M-1$ sensors) but multiple invariances \cite{SORK:92},
which are not fully exploited by LS.}
However, for ESPRIT based on SLS, the asymptotic efficiency is in general higher, in fact,
we have $\eta \rightarrow 1$ for $M=2$, $M=3$ and $M\rightarrow \infty$ for a single source. Moreover, even
for limited $M$, $\eta$ is never far away from 1.
\revf{As we show in the simulations below, we have $\eta = 1$ for $M=2, 3$ and the smallest value
of $\eta$ is obtained for $M=5$ where $\eta = 36/37 \approx 0.973$ {for $d=1$}.}

\section{Simulation results} \label{sec_perf_sims}

In this section we show numerical results to demonstrate the asymptotic
behavior of the analytical performance assessment presented in this chapter.
We first investigate the subspace estimation accuracy in order to verify~\eqref{eqn_perf_exptensub}.
Note that the analytical results
for the subspace estimates are explicit expansions in terms of the perturbation
(i.e., the additive noise). Therefore, we repeat the experiment with
a number of randomly generated realizations of the noise and perform
Monte-Carlo averaging over the analytical expansions. These ``{\bf semi-analytical}''
results are then compared with purely {\bf empirical} results where we estimate
the subspace via an SVD or a HOSVD and compute the estimation error
compared to the true signal subspace.

The subsequent numerical results demonstrate the performance of
ESPRIT-type parameter estimation schemes. Here, we compute the mean
square estimation error in three different ways. Firstly, {\bf analytically}, via 
the MSE expressions provided
in Theorem \ref{thm_perf_mse}, Theorem \ref{thm_perf_mse_fba}, and Theorem \ref{thm_perf_mse_sls}
(eqn.~\eqref{eqn_subsp_perf_sls_mse}), respectively.
Secondly, {\bf semi-analytically}, by performing Monte-Carlo averaging over the explicit first-order
expansions provided in equation~\eqref{eqn_perf_epxl_ste},~\eqref{eqn_perf_epxl_ute},
and~\eqref{eqn_subsp_perf_sls}, respectively.
Thirdly, {\bf empirically}, by estimating the spatial frequencies via the corresponding
ESPRIT-type algorithms and comparing the estimates to the true spatial frequencies.

For all the simulations we assume a known number of planar wavefronts impinging
on an antenna array of $M$ istrotropic sensor elements. We assume uniform $\lambda/2$
spacing in all dimensions, i.e., an $M$-element uniform linear array (ULA) in the 1-D
case and an $M_1 \times M_2$ uniform rectangular array (URA) in the 2-D case.
The sources emit narrow-band waveforms $s_i(t)$ modeled as complex Gaussian distributed symbols $s_i(t)$ 
and we observe $N$ subsequent snapshots $t=1, 2, \ldots, N$. 
All sources are assumed to have unit power, i.e., $\expvof{|s_i(t)|^2} = 1$.
In the case where source correlation
is investigated we generate the symbols $s_i(t)$ such that $\expvof{s_i(t) \cdot s_j(t)^\conj} = \rho \cdot
\expof{\j \varphi_{i,j}}$ for $i\neq j = 1, 2, \ldots, d$, where $\rho$ is the correlation coefficient
between each pair of sources and $\varphi_{i,j}$ is a uniformly distributed correlation phase.
The additive noise is generated according to a circularly symmetric complex Gaussian distribution
with zero mean and variance $\sigma_{\rm n}^2$ and noise samples are assumed to be mutually independent.
Therefore, the Signal to Noise Ratio (SNR) is defined as $1 / \sigma_{\rm n}^2$.

\subsection{Subspace estimation accuracy}

We evaluate the subspace estimation accuarcy by computing the Frobenius norm of the subspace estimation
error, i.e., $\fronorm{\Delta \ma{U}_{\rm s}}^2$ in the matrix case and $\honorm{\unfnot{\Delta \sig{\ten{U}}}{R+1}^\trans}^2$
in the tensor case. 

In order to find the estimation error empirically, we obtain a subspace
estimate $\ma{\hat{U}}_{\rm s}$ via an SVD of the noisy observation and then compare
it to the true subspace $\ma{U}_{\rm s}$ column by column. The estimation error of the $n$-th
column is computed via
\begin{align}
   \Delta\ma{u}_n = \ma{\hat{u}}_n \cdot \frac{\ma{\hat{u}}_n^\herm \ma{u}_n}{\left|\ma{\hat{u}}_n^\herm \ma{u}_n\right|} - \ma{u}_n, \quad n=1, 2, \ldots, d
\end{align}
to account for the inherent phase ambiguity in each column of the SVD, cf.~\cite{LLM:08}.

For the analytical estimation error we calculate $\Delta \ma{U}_{\rm s}$ via the first-order
expansion $\Delta \ma{U}_{\rm s} \approx \ma{U}_{\rm n} \cdot \noi{\ma{\Gamma}}$
 provided in~\eqref{eqn_perf_linexp_vac} and
 the expansion $\Delta \ma{U}_{\rm s} \approx \ma{U}_{\rm n} \cdot \noi{\ma{\Gamma}}+\ma{U}_{\rm s} \cdot \sig{\ma{\Gamma}}$
 provided in~\eqref{eqn_perf_linexp_vacma}, respectively.
 Note that the latter is more accurate since it additionally considers the perturbation of the individual singular vectors,
 i.e., the particular choice of the basis for the signal subspace.
  However, this contribution is irrelevant for the performance of
 ESPRIT-type algorithms. 

\begin{figure}%
   \begin{center}
		\includegraphics[width=\defaultfigwidth]{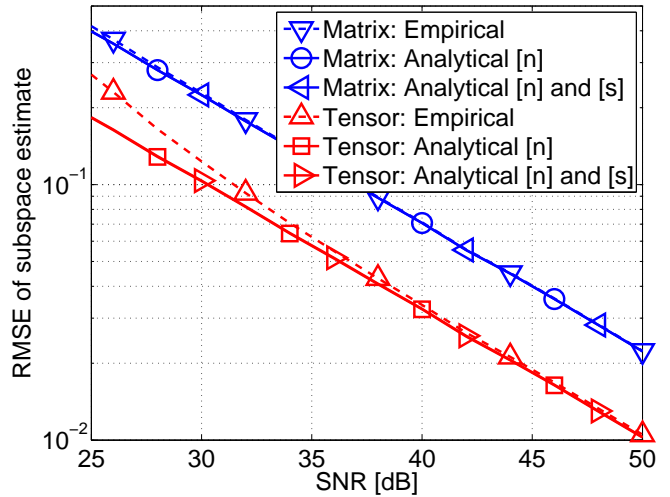}%
	 \end{center}
		\caption{Subspace estimation accuracy using $\noi{\ma{\Gamma}}$ only vs. using 
		$\noi{\ma{\Gamma}}$ and $\sig{\ma{\Gamma}}$. Scenario: $d=3$ correlated sources
		($\rho = 0.97$) at $\mu_1^{(1)} = 0.7, \mu_2^{(1)} = 0.9, \mu_3^{(1)} = 1.1, 
\mu_1^{(2)} = -0.1, \mu_2^{(2)} = -0.3, \mu_3^{(2)} = -0.5$, an $8 \times 8$ URA,
and $N=20$ snapshots.}%
		\label{fig_perf_subsp_3scorr}%
\end{figure}
\begin{figure}
   \begin{center}
		\includegraphics[width=\defaultfigwidth]{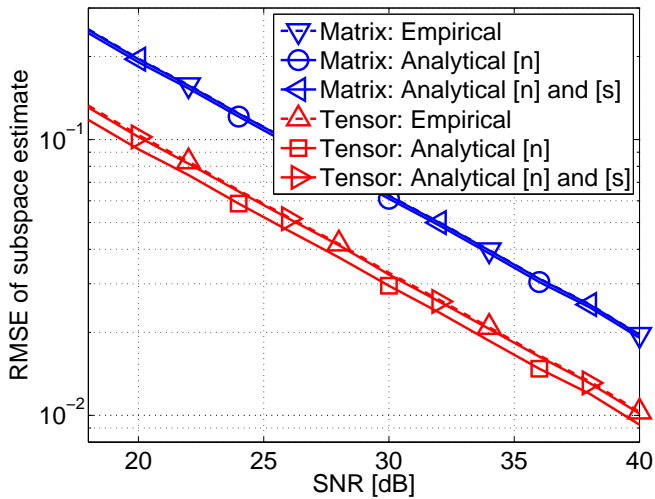}%
	 \end{center}
		\caption{Subspace estimation accuracy using $\noi{\ma{\Gamma}}$ only vs. using 
		$\noi{\ma{\Gamma}}$ and $\sig{\ma{\Gamma}}$. Scenario: $d=4$ uncorrelated sources
		at $\mu_1^{(1)} = -1.5$, $\mu_2^{(1)} = 0.5$, $\mu_3^{(1)} = 1.0$,
   $\mu_4^{(1)} = -0.3$, $\mu_1^{(2)} = 1.3$, $\mu_2^{(2)} = -0.2$, $\mu_3^{(2)} = 0.7$, $\mu_4^{(2)} = -1.5$,
   an $8 \times 8$ URA, and $N=5$ snapshots.}%
		\label{fig_perf_subsp_4sunc}%
\end{figure}

In Figure~\ref{fig_perf_subsp_3scorr} we have $d=3$ sources
positioned at $\mu_1^{(1)} = 0.7, \mu_2^{(1)} = 0.9, \mu_3^{(1)} = 1.1, 
\mu_1^{(2)} = -0.1, \mu_2^{(2)} = -0.3, \mu_3^{(2)} = -0.5$
and mutually correlated with a correlation coefficient of $\rho = 0.97$. Moreover, the array
size is increased to an $8 \times 8$ URA.
For the simulation result shown in Figure~\ref{fig_perf_subsp_4sunc} we consider $d=4$
uncorrelated sources located at $\mu_1^{(1)} = -1.5$, $\mu_2^{(1)} = 0.5$, $\mu_3^{(1)} = 1.0$,
   $\mu_4^{(1)} = -0.3$, $\mu_1^{(2)} = 1.3$, $\mu_2^{(2)} = -0.2$, $\mu_3^{(2)} = 0.7$, $\mu_4^{(2)} = -1.5$
and $N=5$ snapshots. 

Both simulations show that the empirical estimation errors agree with the analytical
results as the SNR tends to infinity. Therefore, the improvement obtained by the HOSVD-based
subspace estimate can reliably be predicted via the analytical expressions. In general, it is
particularly pronounced for correlated sources and for a small number of snapshots.
Moreover, while for three correlated sources, the impact
of the additional term $\ma{U}_{\rm s} \cdot \sig{\ma{\Gamma}}$ is negligibly small, it is
clearly visible for the four uncorrelated sources shown in Figure~\ref{fig_perf_subsp_4sunc}.

\subsection{\texorpdfstring{$R$-D Tensor-ESPRIT}{R-D Tensor-ESPRIT}}

The following set of simulation results demonstrates the performance of $R$-D matrix-based
and tensor-based ESPRIT. As explained in the beginning of this section,
for the analytical results we use~\eqref{eqn_subsp_perf_mse_se}
and~\eqref{eqn_subsp_perf_mse_ste} for $R$-D Standard ESPRIT
and $R$-D Standard Tensor-ESPRIT, respectively, and their
extensions for Unitary ESPRIT and $R$-D Unitary Tensor-ESPRIT
as discussed in Theorem \ref{thm_perf_mse_fba}.
Likewise, the semi-analytical results are obtained by Monte-Carlo averaging
of the explicit first-order expansion provided in~\eqref{eqn_perf_epxl_ste}
and~\eqref{eqn_perf_epxl_ute}, respectively.

\begin{figure}%
   \begin{center}
		\includegraphics[width=\defaultfigwidth]{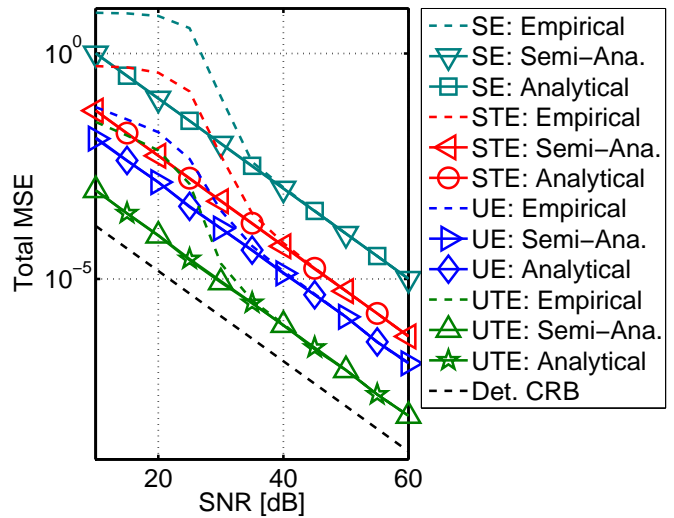}%
	 \end{center}
		\caption{Performance of 2-D SE, STE, UE, UTE for $d=2$ highly correlated sources ($\rho = 0.9999$)
		located at $\mu_1^{(1)} = 1$, $\mu_2^{(1)} = -0.5$, $\mu_1^{(2)} = -0.5$, and
   $\mu_2^{(2)} = 1$, a $5 \times 6$ URA, and $N=20$ snapshots.}%
		\label{fig_perf_esprit_2scorr}%
\end{figure}
\begin{figure}	
   \begin{center}
		\includegraphics[width=\defaultfigwidth]{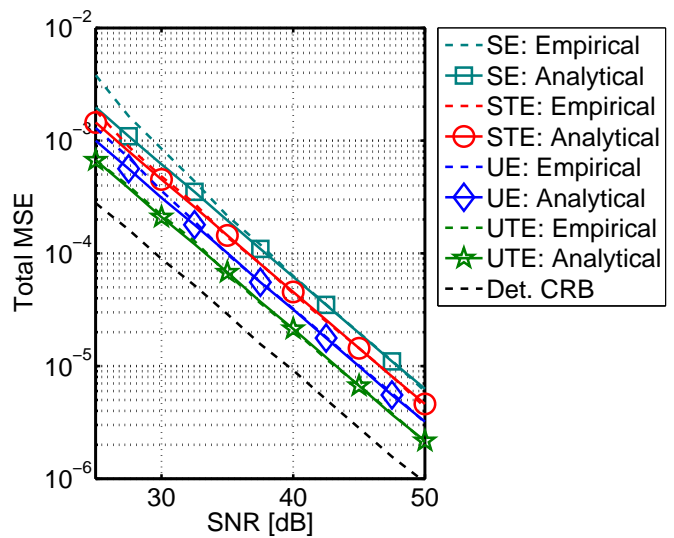}%
	 \end{center}
		\caption{Performance of 2-D SE, STE, UE, UTE for $d=3$ correlated sources ($\rho = 0.97$)
		positioned at $\mu_1^{(1)} = 0.7, \mu_2^{(1)} = 0.9, \mu_3^{(1)} = 1.1, 
\mu_1^{(2)} = -0.1, \mu_2^{(2)} = -0.3, \mu_3^{(2)} = -0.5$, an $8 \times 8$ URA, and $N=20$ snapshots.}%
		\label{fig_perf_esprit_3scorr}%
\end{figure}

For Figure~\ref{fig_perf_esprit_2scorr} we employ a $5 \times 6$ URA
and collect $N=20$ snapshots from two sources located at
   $\mu_1^{(1)} = 1$, $\mu_2^{(1)} = -0.5$, $\mu_1^{(2)} = -0.5$, and
   $\mu_2^{(2)} = 1$. The sources are highly correlated with
   a correlation of $\rho = 0.9999$.
On the other hand, for Figure~\ref{fig_perf_esprit_3scorr} we increase the number
of sources to $d=3$ and the correlation coefficient to $\rho = 0.97$.
Moreover, the spatial frequencies of the sources are given by
$\mu_1^{(1)} = 0.7, \mu_2^{(1)} = 0.9, \mu_3^{(1)} = 1.1, 
\mu_1^{(2)} = -0.1, \mu_2^{(2)} = -0.3, \mu_3^{(2)} = -0.5$
and we use an $8 \times 8$ URA.

To enhance the legibility, we show the semi-analytical estimation errors only
in Figure~\ref{fig_perf_esprit_2scorr} since they always agree with the analytical
results, as expected. Moreover, the empirical estimation errors agree with the
analytical ones for high SNRs. This is also expected as the performance analysis
framework presented here is asymptotically accurate for high effective SNRs.
We conclude that the improvement in terms of estimation accuracy for Tensor-ESPRIT-type
parameter estimation schemes can reliably be predicted via the analytical
expressions we have derived.

\subsection{Structured Least Squares}

The next set of simulation results illustrates the analytical expressions
for ESPRIT using SLS. The semi-analytical MSE is obtained
by Monte-Carlo averaging over the explicit expansion provided in~\eqref{eqn_subsp_perf_sls}
and the analytical MSE is computed via~\eqref{eqn_subsp_perf_sls_mse}.
For the empirical estimation errors we perform a single iteration of the
Structured Least Squares algorithm and do not use regularization (i.e., the regularization
parameter $\alpha$ is set to $\infty$).

\begin{figure}%
   \begin{center}
		\includegraphics[width=\defaultfigwidth]{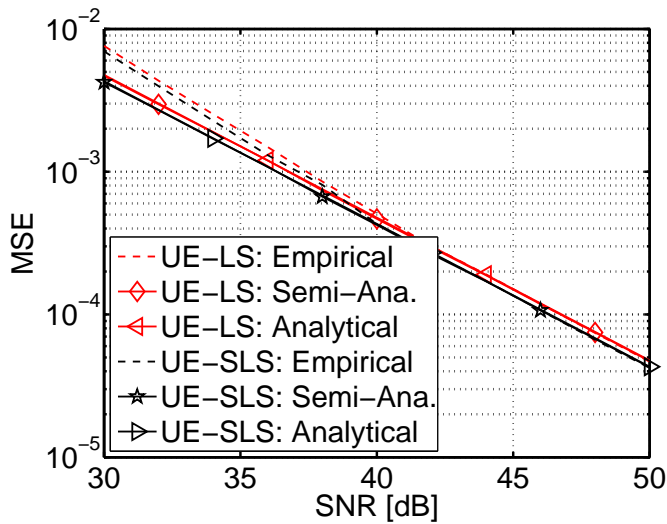}%
	 \end{center}
		\caption{Performance of LS-ESPRIT vs. SLS-ESPRIT for 4 sources at $\mu_1 = 1.0, \mu_2 = 0.7, \mu_3 = -0.6, \mu_4 = -0.3$,
		an $M=8$ ULA, $N=3$ shapshots.}%
		\label{fig_perf_sls_4s}%
\end{figure}
\begin{figure}		
   \begin{center}
		\includegraphics[width=\defaultfigwidth]{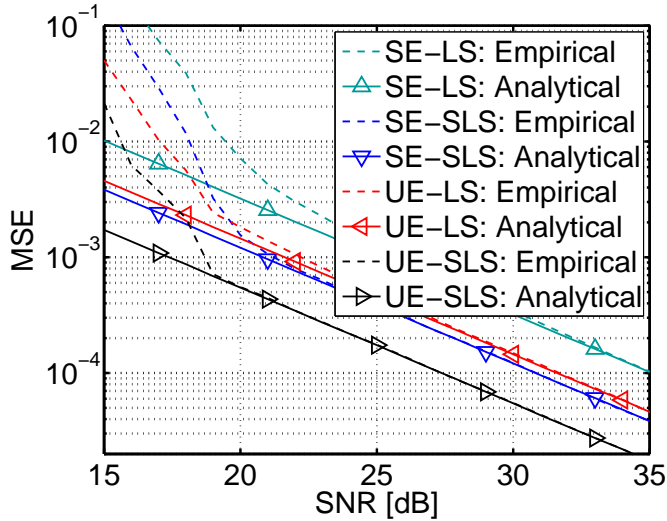}%
	 \end{center}
		\caption{Performance of LS-ESPRIT vs. SLS-ESPRIT for $d=3$ correlated sources ($\rho = 0.99$)
		at $\mu_1 = 1$, $\mu_2 = 0$, $\mu_3 = -1$, a $M=12$ ULA and $N=10$ shapshots. }%
		\label{fig_perf_sls_3scorr}%
\end{figure}

The first simulation result is shown in Figure~\ref{fig_perf_sls_4s}.
Here we consider $N=3$ snapshots from $d=4$ uncorrelated sources captured by an $M=8$ element
uniform linear array. The sources' spatial frequencies are given by
$\mu_1 = 1.0, \mu_2 = 0.7, \mu_3 = -0.6, \mu_4 = -0.3$. Note that since $N<d$, we cannot apply
Standard ESPRIT, therefore, only Unitary ESPRIT is used.
On the other hand, in the second scenario we consider $N=10$ snapshots from $d=3$ sources that
are mutually correlated with a correlation coefficient of $\rho = 0.99$. The sources are
located at $\mu_1 = 1$, $\mu_2 = 0$, $\mu_3 = -1$ and a $M=12$ element ULA is used.
The corresponding estimation errors are shown in Figure~\ref{fig_perf_sls_3scorr}.

As before, the empirical results agree with the analytical results for high SNRs. Moreover,
the improvement in MSE obtained via SLS is particularly pronounced for the correlated sources.
However, even the very slight improvement which is present for four uncorrelated sources
can reliably be predicted via the analytical MSE expressions we have derived.

\subsection{Asymptotic efficiency for a single source}

The final set of simulation results demonstrates the special case of a single source, in which
case the MSE expressions can be simplified to very compact closed-form expressions which
only depend on the physical parameters, i.e., the array size and the SNR.

\begin{figure}%
   \begin{center}
		\includegraphics[width=\defaultfigwidth]{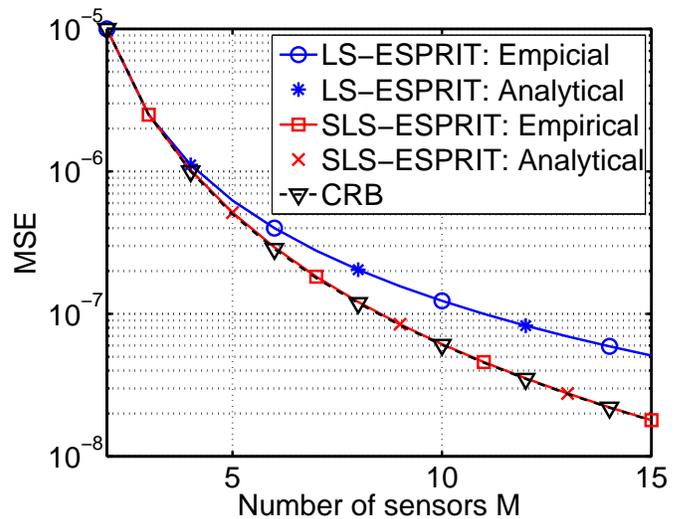}%
	 \end{center}
		\caption{Performance of LS-ESPRIT and SLS-ESPRIT for a single source
		 vs. the number of sensors $M$ ($M$-ULA) at an effective SNR of 25~dB ($P_{\rm T} = 1, \sigma_{\rm n}^2 = 0.032, N = 10$).}%
		\label{fig_perf_sls_singsrc_mse}%
\end{figure}
\begin{figure}
   \begin{center}
		\includegraphics[width=\defaultfigwidth]{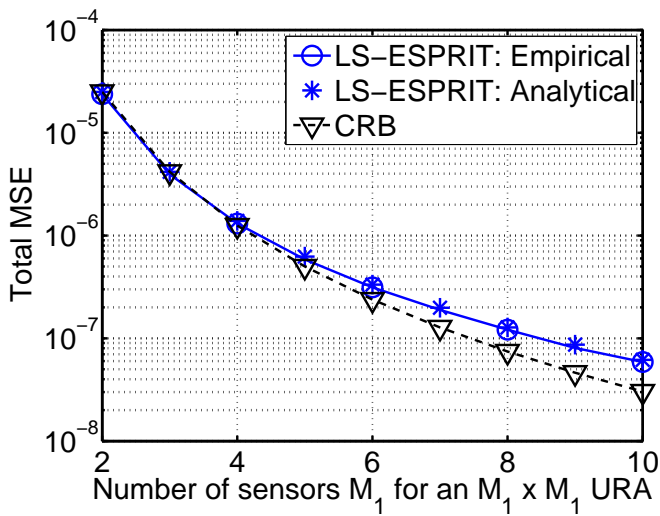}%
	 \end{center}
		\caption{Performance of LS-ESPRIT for a single source vs. $M_1$ using an $M_1 \times M_1$ URA
		at an effective SNR of 46~dB ($P_{\rm T} = 1, \sigma_{\rm n}^2 = 10^{-4}, N = 4$).}%
		\label{fig_perf_sls_singsrc2d_mse}%
\end{figure}

\ifCLASSOPTIONdraftcls
\begin{figure*}%
   \begin{center}
		\includegraphics[width=\columnwidth]{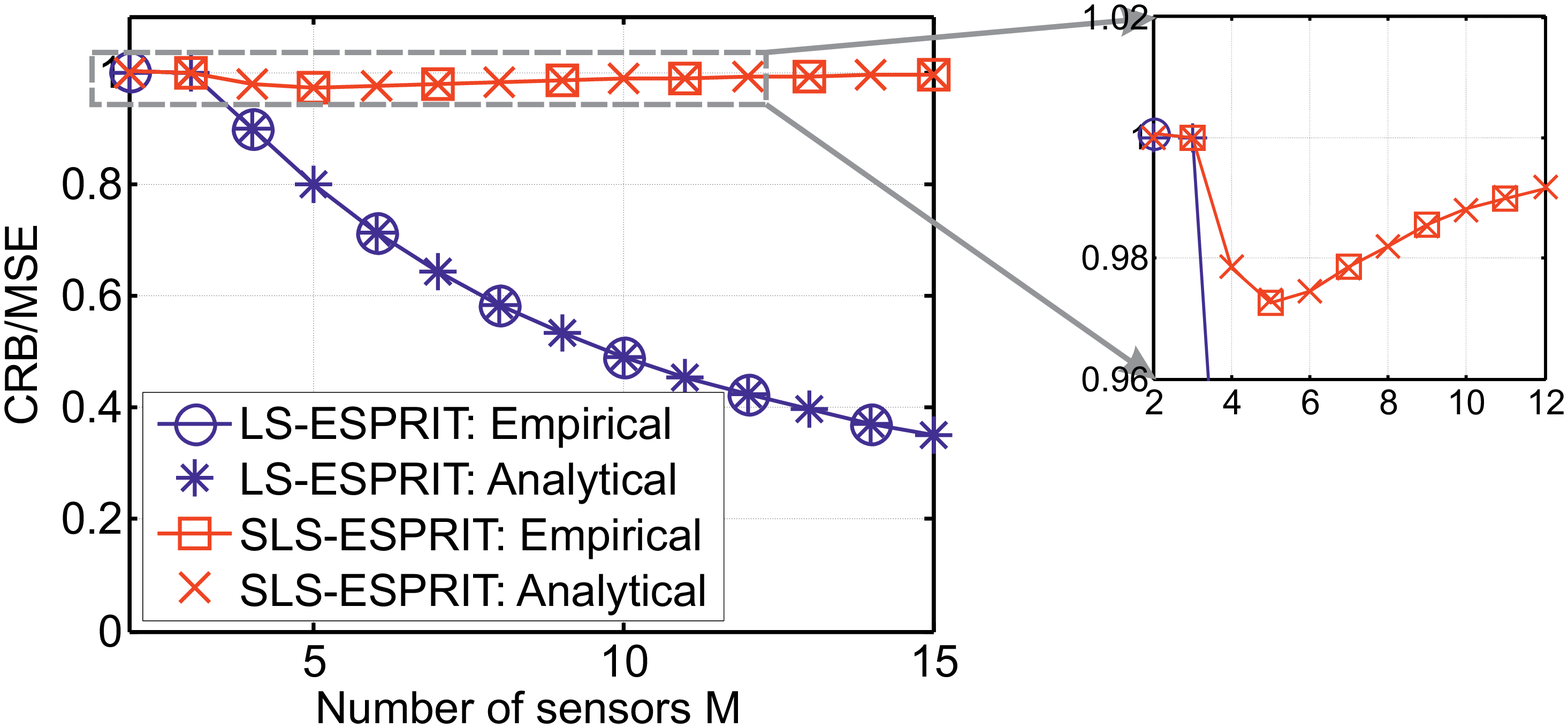}%
	 \end{center}
		\caption{Asymptotic efficiency of LS-ESPRIT vs. SLS-ESPRIT vs. $M$. Same scenario as in Figure~\ref{fig_perf_sls_singsrc_mse}. The left-hand side shows a zoom.}%
		\label{fig_perf_sls_singsrc_eta}%
\end{figure*}
\else
\begin{figure*}%
   \begin{center}
		\includegraphics[width=0.65\linewidth]{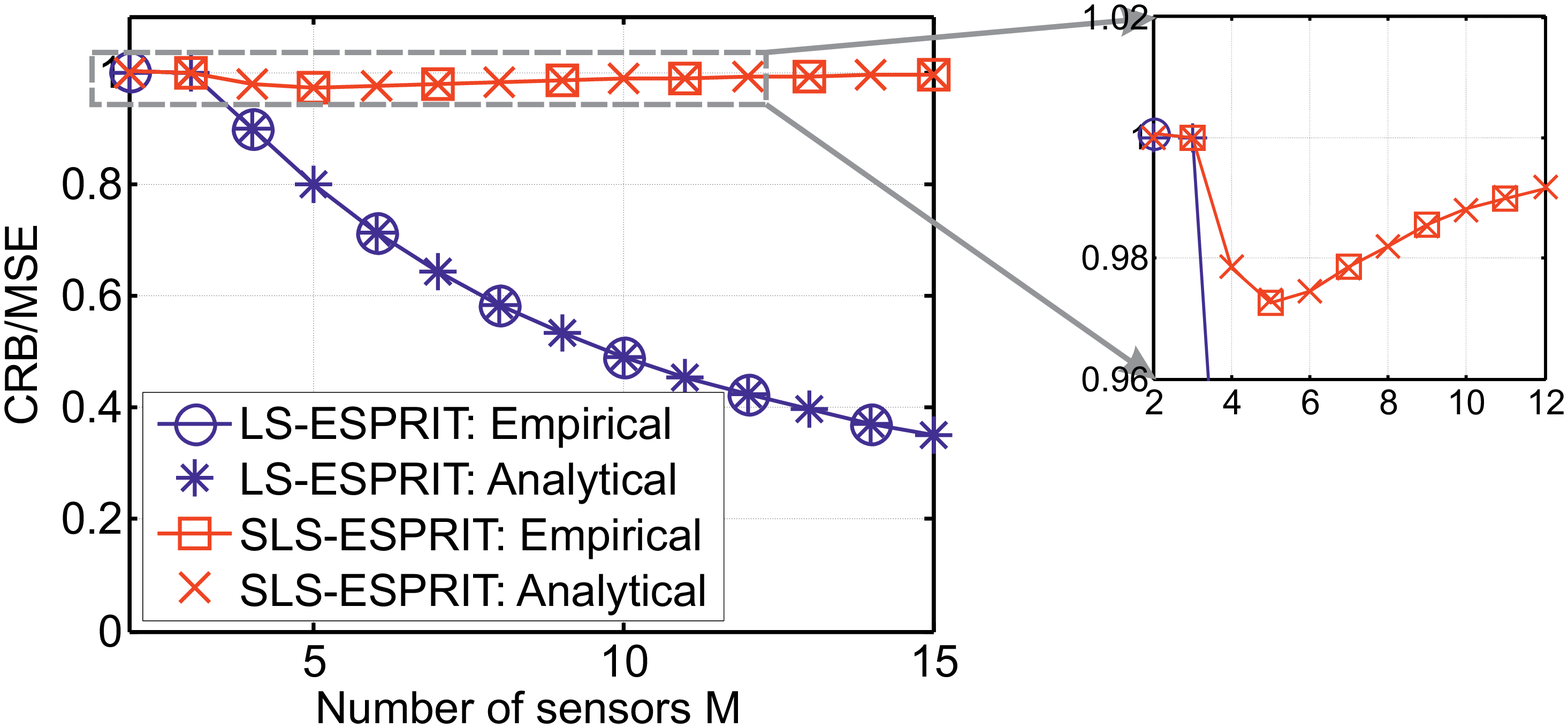}%
	 \end{center}
		\caption{Asymptotic efficiency of LS-ESPRIT vs. SLS-ESPRIT vs. $M$. Same scenario as in Figure~\ref{fig_perf_sls_singsrc_mse}. The left-hand side shows a zoom.}%
		\label{fig_perf_sls_singsrc_eta}%
\end{figure*}
\fi

Figures~\ref{fig_perf_sls_singsrc_mse} and~\ref{fig_perf_sls_singsrc2d_mse} show the MSE
vs. the number of sensors $M$ for a (1-D) Uniform Linear Array and 
vs.~$M_1$ for a (2-D) $M_1 \times M_1$ Uniform
Rectangular Array, respectively. For both scenarios, the spatial frequencies of the single
source were drawn randomly (note that they have no impact on the MSE). The effective
SNR was set to 25~dB for Figure~\ref{fig_perf_sls_singsrc_mse} ($P_{\rm T} = 1, \sigma_{\rm n}^2 = 0.032, N = 10$)
and to 46~dB for Figure~\ref{fig_perf_sls_singsrc2d_mse} ($P_{\rm T} = 1, \sigma_{\rm n}^2 = 10^{-4}, N = 4$),
respectively.

For both plots we observe that LS-ESPRIT is asymptotically efficient for $M=2$ and $M=3$ (which, in the
2-D case means, a $3 \times 3$ URA) and then becomes increasingly inefficient as the array size grows.
Moreover, for the 1-D case we see that SLS-based ESPRIT is in fact very close to the Cram\'er-Rao Bound,
which may, at first sight, lead to believe that the asymptotic efficiency is in fact 1 for all $M$.
However, as we have shown it is in fact slightly lower than one. Therefore, we provide two additional
figures where we depict the ``asymptotic efficiency'', i.e., we divide the CRB by the corresponding
value of the MSE. The resulting efficiency plot is shown in Figure~\ref{fig_perf_sls_singsrc_eta}.
This plot shows more clearly that LS-ESPRIT becomes increasingly inefficient for $M>3$, whereas
SLS-ESPRIT approaches $\eta = 1$ for large $M$. The worst efficiency is found for $M=5$ where we have
$\eta = 36 / 37 \approx 0.973$.

\section{Conclusions}\label{sec_concl}

In this paper we have discussed a framework for analytical performance assessment
of subspace-based parameter estimation schemes. It is based on earlier results
on an explicit first-order expansion of the SVD and its application to
1-D versions of subspace-based parameter estimation schemes, e.g., ESPRIT.
We have extended this framework in a number of ways. Firstly, we have
derived an explicit first-order expansion of the HOSVD-based subspace estimate
which is the basis for Tensor-ESPRIT-type algorithms. 
%
Secondly, we have shown that the first-order expansion for 1-D Standard ESPRIT
can be extended to other ESPRIT-type algorithms, e.g., $R$-D Standard ESPRIT, 
$R$-D Unitary ESPRIT, $R$-D Standard Tensor-ESPRIT, or $R$-D Unitary Tensor-ESPRIT.
%
Thirdly, we have derived a corresponding first-order expansion for Structured Least
Squared (SLS)-based ESPRIT-type algorithms.

All these expansions have in common that they are explicit, i.e., no assumption
about the statistics of either desired signal or additive perturbation need to be
made. We only require the perturbation to be small compared to the desired signal.
{We also do not need the number of snapshots to be large, i.e., they even apply
to the single snapshot case ($N=1$).}
Our fourth contribution is to show that the mean square error can readily be computed
in closed-form \revA{and that it depends only on the second-order moments of the noise.}
\revA{Consequently, for the MSE expressions we only need the noise
to be zero mean and its second order moments to be finite. Neither Gaussianity nor
circular symmetry is required.}
\revB{This is a particularly attractive feature of our approach with respect
to different types of preprocessing which
alters the noise statistics, e.g., spatial smoothing (which yields spatially correlated noise)
or forward-backward averaging (which annihilates the circular symmetry of the noise).
Since we do not require spatial whiteness or circular symmetry, our MSE expressions
are directly applicable.}
The resulting
MSE expressions are asymptotic in the effective SNR, i.e., they become accurate as either
the noise variance goes to zero or the number of observations goes to infinity.

As a final contribution we have investigated the special case of a single source,
\revA{circularly symmetric white noise,}
and uniform linear (1-D) or uniform rectangular (2-D) arrays.
In this case we have been able to show analytically, that $R$-D Standard ESPRIT, $R$-D Unitary
ESPRIT, and (for $R=2$) $R$-D Standard Tensor-ESPRIT as well as $R$-D Unitary Tensor-ESPRIT
yield the same MSE, which only depends on the effective SNR and the number of antenna elements.
We have also shown that 1-D Standard ESPRIT using SLS has a lower MSE which is also expressed
explicitly as a function of the effective SNR and the number of antenna elements. 

\appendices

\section{Proof of Theorem~\ref{thm_perf_reltenmat}}\label{sec_app_proof_perf_reltenmat}

As shown in~\eqref{eqn_subsp_usten_def}, the estimated signal subspace tensor can be computed via
\begin{align}
   \sig{\ten{\hat{U}}} = \sig{\ten{\hat{S}}} \times_1 \sig{\ma{\hat{U}}}_1 \ldots \times_R \sig{\ma{\hat{U}}}_R
   \times_{R+1} \ma{\hat{\Sigma}}_{\rm s}^{-1}. \label{eqn_app_proof_perf_subsptendef}
\end{align}
Here, $\sig{\ten{\hat{S}}}$ represents the truncated version core tensor $\ten{\hat{S}}$
from the HOSVD of $\ten{X}$. In order to eliminate $\sig{\ten{\hat{S}}}$ we require the following
Lemma:

\begin{lem} \label{lem_app_proof_perf_lemtrnccore}
	 The truncated core tensor $\sig{\ten{\hat{S}}}$ can be computed from $\ten{X}$ directly via
	 \begin{align}
   		\ten{\hat{S}} = \ten{X} \times_1 \sigH{\ma{\hat{U}}}_1 \ldots \times_{R+1} \sigH{\ma{\hat{U}}}_{R+1}.
   		\label{eqn_app_proof_perf_coretentrunc}
	\end{align}
\end{lem}

\begin{proof}
   To show~\eqref{eqn_app_proof_perf_coretentrunc} we insert the HOSVD of $\ten{X}$ given by
   $\ten{X} = \ten{\hat{S}} \times_1 \ma{\hat{U}}_1 \ldots \times_{R+1} \ma{\hat{U}}_{R+1}$. 
   Using 
   \eqref{eqn_notation_nmp_twice},
   we obtain
   \begin{align}
       \ten{\hat{S}} = \ten{S} \times_1 \left(\sigH{\ma{\hat{U}}}_1\cdot \ma{\hat{U}}_1\right) \ldots 
       \times_{R+1} \left(\sigH{\ma{\hat{U}}}_{R+1} \cdot \ma{\hat{U}}_{R+1} \right).
   \end{align}
   However, since the matrices of $r$-mode singular vectors $\ma{\hat{U}}_r$ are unitary, they satisfy
   $\sigH{\ma{\hat{U}}}_r \cdot \ma{\hat{U}}_r = \left[ \ma{I}_{p_r}, \; \ma{0}_{p_r \times (M_r - p_r)}\right]$.
   Therefore, $\sig{\ten{\hat{S}}}$ computed via~\eqref{eqn_app_proof_perf_coretentrunc} contains the first
   $p_r$ elements of $\ten{\hat{S}}$ in the $r$-th mode, which shows that it is indeed the truncated
   core tensor.
\end{proof}

Next, we use Lemma~\ref{lem_app_proof_perf_lemtrnccore} to eliminate $\ten{\hat{S}}$
in~\eqref{eqn_app_proof_perf_subsptendef}. We obtain
\ifCLASSOPTIONdraftcls
\begin{align}
   \sig{\ten{\hat{U}}} & = \ten{X} \times_1 \left(\sig{\ma{\hat{U}}}_1\cdot \sigH{\ma{\hat{U}}}_1\right) 
   \ldots \times_R \left(\sig{\ma{\hat{U}}}_R\cdot \sigH{\ma{\hat{U}}}_R\right)
   \times_{R+1} \left(\ma{\hat{\Sigma}}_{\rm s}^{-1} \cdot \sigH{\ma{\hat{U}}}_{R+1}\right)  \\
   & = \ten{X} \times_1 \ma{\hat{T}}_1 
   \ldots \times_R \ma{\hat{T}}_R
   \times_{R+1} \left(\ma{\hat{\Sigma}}_{\rm s}^{-1} \cdot \sigH{\ma{\hat{U}}}_{R+1}\right), \label{eqn_proof_perf_usten_ir}
\end{align}
\else
\begin{align}
   \sig{\ten{\hat{U}}} = &\ten{X} \times_1 \left(\sig{\ma{\hat{U}}}_1\cdot \sigH{\ma{\hat{U}}}_1\right) 
   \ldots \times_R \left(\sig{\ma{\hat{U}}}_R\cdot \sigH{\ma{\hat{U}}}_R\right) \notag \\
   & \times_{R+1} \left(\ma{\hat{\Sigma}}_{\rm s}^{-1} \cdot \sigH{\ma{\hat{U}}}_{R+1}\right)  \\
    = &\ten{X} \times_1 \ma{\hat{T}}_1 
   \ldots \times_R \ma{\hat{T}}_R
   \times_{R+1} \left(\ma{\hat{\Sigma}}_{\rm s}^{-1} \cdot \sigH{\ma{\hat{U}}}_{R+1}\right), \label{eqn_proof_perf_usten_ir}
\end{align}
\fi
where we have introduced the short hand notation $\ma{\hat{T}}_r = \sig{\ma{\hat{U}}}_r\cdot \sigH{\ma{\hat{U}}}_r$.
The next step is to compute the matrix $\unfnot{\sig{\ten{\hat{U}}}}{R+1}^\trans$.
Inserting~\eqref{eqn_proof_perf_usten_ir} and using 
\eqref{eqn_notation_nmp_unf},
we obtain
\begin{align}
   \unfnot{\sig{\ten{\hat{U}}}}{R+1}^\trans & = 
   \left( \ma{\hat{T}}_1 \otimes \ldots \otimes \ma{\hat{T}}_R\right) \cdot 
   \unfnot{\ten{X}}{R+1}^\trans \cdot 
   \sigC{\ma{\hat{U}}}_{R+1} \cdot 
   \ma{\hat{\Sigma}}_{\rm s}^{-1}. \label{eqn_proof_perf_tensubspunf_ir2}
\end{align}
%

As pointed out in Section~\ref{sec_dm_matten}, the link between the measurement matrix
$\ma{X}$ and the measurement tensor $\ten{X}$ is given by
$\ma{X} = \unfnot{\ten{X}}{R+1}^\trans$. Therefore, their SVDs (cf.~\eqref{eqn_svd_X} and~\eqref{eqn_perf_unfsvd_x}) 
are linked through the following identities
\begin{align}
    \ma{\hat{U}}_{\rm s} = \sigC{\ma{\hat{V}}}_{R+1}, \; \ma{\hat{U}}_{\rm n} = \noiC{\ma{\hat{V}}}_{R+1}, \;
    \ma{\hat{V}}_{\rm s} = \sigC{\ma{\hat{U}}}_{R+1}, \; \ma{\hat{V}}_{\rm n} = \noiC{\ma{\hat{U}}}_{R+1}.
    \notag
\end{align}
Consequently we can write 
\ifCLASSOPTIONdraftcls
\begin{align}
   \unfnot{\ten{X}}{R+1}^\trans \cdot 
   \sigC{\ma{\hat{U}}}_{R+1} \cdot 
   \ma{\hat{\Sigma}}_{\rm s}^{-1}
   = 
   \ma{X} \cdot 
   \ma{\hat{V}}_{\rm s} \cdot 
   \ma{\hat{\Sigma}}_{\rm s}^{-1}   
   = \ma{\hat{U}} \cdot \ma{\hat{\Sigma}} \cdot \ma{\hat{V}}^\herm  \cdot 
   \ma{\hat{V}}_{\rm s} \cdot 
   \ma{\hat{\Sigma}}_{\rm s}^{-1}   
   = \ma{\hat{U}}_{\rm s} \cdot \ma{\hat{\Sigma}}_s \cdot \ma{\hat{\Sigma}}_{\rm s}^{-1}   
   = \ma{\hat{U}}_{\rm s}. \label{eqn_proof_perf_xunfus}
\end{align}
\else
\begin{align}
  \begin{split}
   & \unfnot{\ten{X}}{R+1}^\trans \cdot 
   \sigC{\ma{\hat{U}}}_{R+1} \cdot 
   \ma{\hat{\Sigma}}_{\rm s}^{-1}
   = 
   \ma{X} \cdot 
   \ma{\hat{V}}_{\rm s} \cdot 
   \ma{\hat{\Sigma}}_{\rm s}^{-1} \\
   = & \ma{\hat{U}} \cdot \ma{\hat{\Sigma}} \cdot \ma{\hat{V}}^\herm  \cdot 
   \ma{\hat{V}}_{\rm s} \cdot 
   \ma{\hat{\Sigma}}_{\rm s}^{-1}   
   = \ma{\hat{U}}_{\rm s} \cdot \ma{\hat{\Sigma}}_s \cdot \ma{\hat{\Sigma}}_{\rm s}^{-1}   
   = \ma{\hat{U}}_{\rm s}. 
  \end{split} \label{eqn_proof_perf_xunfus}
\end{align}
\fi
%
Finally, inserting~\eqref{eqn_proof_perf_xunfus} into~\eqref{eqn_proof_perf_tensubspunf_ir2}
yields
\begin{align}
       \unfnot{\sig{\ten{\hat{U}}}}{R+1}^\trans =
           \left( \ma{\hat{T}}_1 \otimes \ldots \otimes \ma{\hat{T}}_R\right) \cdot 
           \ma{\hat{U}}_{\rm s},
    \end{align}
which is the desired result. \qed

\begin{cor} \label{cor_app_proofs_perf_kronprojsubsp}
   A corollary which follows from this theorem is that the exact subspace $\ma{U}_{\rm s}$
   satisfies the following identity
   \begin{align}
       \ma{U}_{\rm s} = \left( \ma{T}_1 \otimes \ldots \otimes \ma{T}_R\right) \cdot 
           \ma{U}_{\rm s}.
   \end{align}
\end{cor}
\begin{proof}
   The corollary follows by considering the special case where $\ten{X} = \ten{X}_0$ and
   hence $\ma{\hat{T}}_r = \ma{T}_r$ as well as $\ma{\hat{U}}_{\rm s} = \ma{U}_{\rm s}$.
   For this case we also have $\unfnot{\sig{\ten{\hat{U}}}}{R+1}^\trans 
   = \unfnot{\sig{\ten{{U}}}}{R+1}^\trans = \ma{U}_{\rm s}$, where the last identity
   resembles the fact that in the noise-free case, the HOSVD-based subspace estimate
   coincides with the SVD-bases subspace estimate.
\end{proof}

\section{Proof of Theorem~\ref{thm_perf_exptensub}}\label{sec_app_proof_perf_exptensub}

We start by inserting $\ma{\hat{U}}_{\rm s} = \ma{U}_{\rm s} + \Delta \ma{U}_{\rm s}$
and $\ma{\hat{T}}_r = \ma{T}_r + \Delta \ma{T}_r$ into~\eqref{eqn_perf_reltenmat2}.
{Then} we obtain
\ifCLASSOPTIONdraftcls
\begin{align}
       \unfnot{\sig{\ten{\hat{U}}}}{R+1}^\trans &= \left[
           \left(\ma{{T}}_1+\Delta\ma{T}_1\right) \kron 
           \ldots 
           \kron
           \left(\ma{{T}}_R+\Delta\ma{T}_R\right)\right]
           \cdot \left( \ma{{U}}_{\rm s}+\Delta\ma{U}_{\rm s}\right) \notag \\
		& = \underbrace{\left[\ma{T}_1 \kron \ldots \kron \ma{T}_R \right] \cdot \ma{U}_{\rm s}}_{\ma{U}_{\rm s}}
		+ \left[\ma{T}_1 \kron \ldots \kron \ma{T}_R \right] \cdot \Delta \ma{U}_{\rm s}  \label{eqn_app_proof_perttenexp_ir1} \\
		& + \left[\Delta\ma{T}_1 \kron \ma{T}_2 \kron \ldots \kron \ma{T}_R \right]
           \cdot \ma{U}_{\rm s}
    + \ldots 
    + \left[\ma{T}_1 \kron \ma{T}_2 \kron \ldots \kron \Delta\ma{T}_R \right]
           \cdot \ma{U}_{\rm s}
    + \bigO{\Delta^2}, \notag
\end{align}
\else
\begin{align}
       \unfnot{\sig{\ten{\hat{U}}}}{R+1}^\trans &= \left[
           \left(\ma{{T}}_1+\Delta\ma{T}_1\right) \kron 
           \ldots 
           \kron
           \left(\ma{{T}}_R+\Delta\ma{T}_R\right)\right]
           \cdot \left( \ma{{U}}_{\rm s}+\Delta\ma{U}_{\rm s}\right) \notag\\
		& = \underbrace{\left[\ma{T}_1\kron \ldots \kron \ma{T}_R \right] \cdot \ma{U}_{\rm s}}_{\ma{U}_{\rm s}}
		+ \left[\ma{T}_1 \kron \ldots \kron \ma{T}_R \right] \cdot \Delta \ma{U}_{\rm s}  \notag \\
		& + \left[\Delta\ma{T}_1 \kron \ma{T}_2 \kron \ldots \kron \ma{T}_R \right]
           \cdot \ma{U}_{\rm s}
    + \ldots \label{eqn_app_proof_perttenexp_ir1} \\
    & + \left[\ma{T}_1 \kron \ma{T}_2 \kron \ldots \kron \Delta\ma{T}_R \right]
           \cdot \ma{U}_{\rm s}
    + \bigO{\Delta^2}, \notag
\end{align}
\fi
since all terms that contain more than one perturbation term can be absorbed into 
$\bigO{\Delta^2}$.
The first term in~\eqref{eqn_app_proof_perttenexp_ir1} represents the exact signal subspace
(cf.~Corollary~\ref{cor_app_proofs_perf_kronprojsubsp}),
hence the remaining terms \reva{are} the first order expansion of 
$\unfnot{\Delta\sig{\ten{\hat{U}}}}{R+1}^\trans$. 
As the first term of this expansion already agrees with Theorem~\ref{thm_perf_exptensub},
we still need to show that for the
remaining terms we have for $r=1, 2, \ldots, R$
\ifCLASSOPTIONdraftcls
\begin{align}
	\left[\ma{T}_1 \kron \ldots \kron \Delta\ma{T}_r \kron \ldots \ma{T}_R \right]
           \cdot \ma{U}_{\rm s}
  =  \left[\ma{T}_1 \kron \ldots \kron (\noi{\ma{U}}_r \cdot \noi{\ma{\Gamma}}_r\cdot\sigH{\ma{U}}_r) \kron \ldots \ma{T}_R \right]
           \cdot \ma{U}_{\rm s} + \bigO{\Delta^2}. \label{eqn_app_proof_perttenexp_ir1p5}
\end{align}
\else
\begin{align}
	&\left[\ma{T}_1 \kron \ldots \kron \Delta\ma{T}_r \kron \ldots \ma{T}_R \right]
           \cdot \ma{U}_{\rm s}
  =  \label{eqn_app_proof_perttenexp_ir1p5} \\
  & \left[\ma{T}_1 \kron \ldots \kron (\noi{\ma{U}}_r \cdot \noi{\ma{\Gamma}}_r\cdot\sigH{\ma{U}}_r) \kron \ldots \ma{T}_R \right]
           \cdot \ma{U}_{\rm s} + \bigO{\Delta^2}. \notag
\end{align}
\fi
As a first step, we expand the left-hand side of~\eqref{eqn_app_proof_perttenexp_ir1p5}
by applying Corollary~\ref{cor_app_proofs_perf_kronprojsubsp}
\ifCLASSOPTIONdraftcls
\begin{align}
	\left[\ma{T}_1 \kron \ldots \kron \Delta\ma{T}_r \kron \ldots \ma{T}_R \right]
           \cdot \ma{U}_{\rm s}
  & = \left[\ma{T}_1 \kron \ldots \kron \Delta\ma{T}_r \kron \ldots \ma{T}_R \right]
  \cdot
     \left[\ma{T}_1 \kron \ldots \kron \ma{T}_r \kron \ldots \ma{T}_R \right]
           \cdot \ma{U}_{\rm s} \notag \\
	& = \left[(\ma{T}_1\cdot \ma{T}_1) \kron \ldots \kron (\Delta\ma{T}_r\cdot \ma{T}_r) \kron \ldots (\ma{T}_R\cdot \ma{T}_R) \right]
           \cdot \ma{U}_{\rm s}  \notag \\   
	& = \left[\ma{T}_1 \kron \ldots \kron (\Delta\ma{T}_r\cdot \ma{T}_r) \kron \ldots \ma{T}_R \right]
           \cdot \ma{U}_{\rm s},           
\end{align}
\else
\begin{align}
	& \left[\ma{T}_1 \kron \ldots \kron \Delta\ma{T}_r \kron \ldots \ma{T}_R \right]
           \cdot \ma{U}_{\rm s} \notag \\
   = & \left[\ma{T}_1 \kron \ldots \kron \Delta\ma{T}_r \kron \ldots \ma{T}_R \right]
  \cdot
     \left[\ma{T}_1 \kron \ldots \kron \ma{T}_r \kron \ldots \ma{T}_R \right]
           \cdot \ma{U}_{\rm s} \notag \\
	 = & \left[(\ma{T}_1\cdot \ma{T}_1) \kron \ldots \kron (\Delta\ma{T}_r\cdot \ma{T}_r) \kron \ldots (\ma{T}_R\cdot \ma{T}_R) \right]
           \cdot \ma{U}_{\rm s}  \notag \\   
	 = & \left[\ma{T}_1 \kron \ldots \kron (\Delta\ma{T}_r\cdot \ma{T}_r) \kron \ldots \ma{T}_R \right]
           \cdot \ma{U}_{\rm s},           
\end{align}
\fi
where we have used the fact that the matrices $\ma{T}_r$ are projection matrices and hence idempotent, i.e.,
$\ma{T}_r \cdot \ma{T}_r = \ma{T}_r$.
What remains to be shown is that $\Delta\ma{T}_r\cdot \ma{T}_r = \noi{\ma{U}}_r \cdot \noi{\ma{\Gamma}}_r\cdot\sigH{\ma{U}}_r + \bigO{\Delta^2}$.
%
Since $\ma{\hat{T}}_r = \sig{\ma{\hat{U}}}_r \cdot \sigH{\ma{\hat{U}}}_r$
and $\sig{\ma{\hat{U}}}_r = \sig{\ma{{U}}}_r + \Delta \sig{\ma{{U}}}_r$, a first order expansion
for $\Delta\ma{T}_r$ is obtained via
\begin{align}
   \ma{\hat{T}}_r  & = \left(\sig{\ma{{U}}_r} + \Delta \sig{\ma{{U}}_r}\right)
             \cdot  \left(\sigH{\ma{U}_r} + \Delta \sigH{\ma{{U}}_r}\right) \notag \\
              & = \ma{{T}}_r + \sig{\ma{{U}}_r} \cdot \Delta \sigH{\ma{{U}}_r}
             + \Delta \sig{\ma{{U}}_r} \cdot \sigH{\ma{{U}}_r}
             + \bigO{\Delta^2} \notag \\
	\Rightarrow \Delta\ma{T}_r & = 
	\sig{\ma{{U}}_r} \cdot \Delta \sigH{\ma{{U}}_r}
             + \Delta \sig{\ma{{U}}_r} \cdot \sigH{\ma{{U}}_r} + \bigO{\Delta^2},
             \label{eqn_app_proof_perttenexp_ir2}
\end{align}
where in general we have $\Delta \sig{\ma{{U}}_r} = \noi{\ma{U}}_r \cdot \noi{\ma{\Gamma}}_r + \sig{\ma{U}}_r \cdot \sig{\ma{\Gamma}}_r + \bigO{\Delta^2}$ (cf.~\eqref{eqn_perf_unfexp}).
Using this expansion in~\eqref{eqn_app_proof_perttenexp_ir2} we obtain
\ifCLASSOPTIONdraftcls
\begin{align}
	\Delta\ma{T}_r \cdot \ma{T}_r & = 
	\sig{\ma{{U}}_r} \cdot \left( \noi{\ma{U}}_r \cdot \noi{\ma{\Gamma}}_r + \sig{\ma{U}}_r \cdot \sig{\ma{\Gamma}}_r \right)^\herm
	\cdot \ma{T}_r
	+
   \left( \noi{\ma{U}}_r \cdot \noi{\ma{\Gamma}}_r + \sig{\ma{U}}_r \cdot \sig{\ma{\Gamma}}_r \right)
    \cdot \sigH{\ma{{U}}_r}\cdot \ma{T}_r + \bigO{\Delta^2} \notag \\
    & = 
    \sig{\ma{{U}}_r} \cdot \noiH{\ma{\Gamma}}_r \cdot \noiH{\ma{U}}_r\cdot \ma{T}_r +
    \sig{\ma{{U}}_r} \cdot \sigH{\ma{\Gamma}}_r \cdot \sigH{\ma{U}}_r\cdot \ma{T}_r \notag \\ & +
    \sig{\ma{{U}}_r} \cdot \sig{\ma{\Gamma}}_r \cdot \sigH{\ma{U}}_r\cdot \ma{T}_r +
    \noi{\ma{{U}}_r} \cdot \noi{\ma{\Gamma}}_r \cdot \sigH{\ma{U}}_r\cdot \ma{T}_r +
    \bigO{\Delta^2} \notag \\
    & = 
    \sig{\ma{{U}}_r} \cdot \noiH{\ma{\Gamma}}_r \cdot \underbrace{\noiH{\ma{U}}_r\cdot \ma{T}_r}_{\ma{0}_{M_r-p_r \times M_r}} +
    \sig{\ma{{U}}_r} \cdot \underbrace{\left(\sig{\ma{\Gamma}}_r + \sigH{\ma{\Gamma}}_r \right)}_{\ma{0}_{p_r \times p_r}} \cdot \sigH{\ma{U}}_r \cdot \ma{T}_r+    
    \noi{\ma{{U}}_r} \cdot \noi{\ma{\Gamma}}_r \cdot \underbrace{\sigH{\ma{U}}_r \cdot \ma{T}_r}_{\sigH{\ma{U}}_r} +
    \bigO{\Delta^2} \notag \\
    & =     
    \noi{\ma{{U}}_r} \cdot \noi{\ma{\Gamma}}_r \cdot \sigH{\ma{U}}_r +
    \bigO{\Delta^2}, \label{eqn_app_proof_perttenexp_ir3}
\end{align}
\else
\begin{align}
	& \Delta\ma{T}_r \cdot \ma{T}_r \notag \\ = &
	\sig{\ma{{U}}_r} \cdot \left( \noi{\ma{U}}_r \cdot \noi{\ma{\Gamma}}_r 
	+ \sig{\ma{U}}_r \cdot \sig{\ma{\Gamma}}_r \right)^\herm
	\cdot \ma{T}_r \notag \\ &
	+
   \left( \noi{\ma{U}}_r \cdot \noi{\ma{\Gamma}}_r + \sig{\ma{U}}_r \cdot \sig{\ma{\Gamma}}_r \right)
    \cdot \sigH{\ma{{U}}_r}\cdot \ma{T}_r + \bigO{\Delta^2} \notag \\
     = &
    \sig{\ma{{U}}_r} \cdot \noiH{\ma{\Gamma}}_r \cdot \noiH{\ma{U}}_r\cdot \ma{T}_r +
    \sig{\ma{{U}}_r} \cdot \sigH{\ma{\Gamma}}_r \cdot \sigH{\ma{U}}_r\cdot \ma{T}_r \notag \\ & +
    \sig{\ma{{U}}_r} \cdot \sig{\ma{\Gamma}}_r \cdot \sigH{\ma{U}}_r\cdot \ma{T}_r +
    \noi{\ma{{U}}_r} \cdot \noi{\ma{\Gamma}}_r \cdot \sigH{\ma{U}}_r\cdot \ma{T}_r +
    \bigO{\Delta^2} \notag \\
     = &
    \sig{\ma{{U}}_r} \cdot \noiH{\ma{\Gamma}}_r \cdot \underbrace{\noiH{\ma{U}}_r\cdot \ma{T}_r}_{\ma{0}_{M_r-p_r \times M_r}} +
    \sig{\ma{{U}}_r} \cdot \underbrace{\left(\sig{\ma{\Gamma}}_r + \sigH{\ma{\Gamma}}_r \right)}_{\ma{0}_{p_r \times p_r}} \cdot \sigH{\ma{U}}_r \cdot \ma{T}_r
    \notag \\ &
    +    
    \noi{\ma{{U}}_r} \cdot \noi{\ma{\Gamma}}_r \cdot \underbrace{\sigH{\ma{U}}_r \cdot \ma{T}_r}_{\sigH{\ma{U}}_r} +
    \bigO{\Delta^2} \notag \\
     =  &
    \noi{\ma{{U}}_r} \cdot \noi{\ma{\Gamma}}_r \cdot \sigH{\ma{U}}_r +
    \bigO{\Delta^2}, \label{eqn_app_proof_perttenexp_ir3}
\end{align}
\fi
which is the desired result. 
Note that 
$\sig{\ma{\Gamma}}_r + \sigH{\ma{\Gamma}}_r = \ma{0}_{p_r \times p_r}$ follows from the fact that $\sig{\ma{\Gamma}}_r$
is a skew-Hermitian matrix (which is apparent from its definition shown in~\eqref{eqn_perf_unfexp}).
This completes the proof of the theorem. \qed

\section{Proof of Theorem~\ref{thm_perf_mse}}\label{sec_app_proof_perf_mse}

For $R$-D Standard ESPRIT, the explicit first-order expansion of the estimation error 
for the $k$-th spatial frequency in the $r$-th mode 
in terms of the signal subspace estimation error $\Delta \ma{U}_{\rm s}$ is 
given by~\eqref{eqn_perf_1dse_expl}.
This error can be expressed in terms of the perturbation (noise) matrix $\ma{N}$
by inserting \eqref{eqn_perf_linexp_vac}.
We obtain
%
%
\ifCLASSOPTIONdraftcls
\begin{align}
    \Delta \mu_k^{(r)} = & \imagof{
       \ma{r}_k^{(r)^\trans}
       \cdot
       \vecof{
       \Delta \ma{U}_{\rm s}}
    } + \bigO{\Delta^2}
    =  \imagof{
       \ma{r}_k^{(r)^\trans}
       \cdot
       \ma{W}_{\rm mat} \cdot 
       \vecof{\ma{N}}
    } + \bigO{\Delta^2} \label{eqn_app_proof_pert_mse_ir1} \\
\ma{r}_k^{(r)^\trans} = & \ma{q}_k^{(r)^\trans} \kron \left(
\ma{p}_k^{(r)^\trans} \cdot \left( \ma{\tilde{J}}_1^{(r)} \cdot \ma{U}_{\rm s} \right)^+
       \cdot
       \left[ \ma{\tilde{J}}_2^{(r)} / \lambda_k^{(r)} -\ma{\tilde{J}}_1^{(r)} \right]\right)  \notag \\
\ma{W}_{\rm mat} = &
\left(\ma{\Sigma}_{\rm s}^{-1} \cdot \ma{V}_{\rm s}^\trans \right) \kron \left( {\ma{U}}_{\rm n}  \cdot {\ma{U}}_{\rm n}^\herm \right)  \notag
\end{align}
\else
\begin{align}
    \Delta \mu_k^{(r)} = & \imagof{
       \ma{r}_k^{(r)^\trans}
       \cdot
       \vecof{
       \Delta \ma{U}_{\rm s}}
    } + \bigO{\Delta^2} \notag \\
    = & \imagof{
       \ma{r}_k^{(r)^\trans}
       \cdot
       \ma{W}_{\rm mat} \cdot 
       \vecof{\ma{N}}
    } + \bigO{\Delta^2} \label{eqn_app_proof_pert_mse_ir1} \\
\ma{r}_k^{(r)^\trans} = & \ma{q}_k^{(r)^\trans} \kron \left(
\ma{p}_k^{(r)^\trans} \cdot \left( \ma{\tilde{J}}_1^{(r)} \cdot \ma{U}_{\rm s} \right)^+
       \cdot
       \left[ \ma{\tilde{J}}_2^{(r)} / \lambda_k^{(r)} -\ma{\tilde{J}}_1^{(r)} \right]\right)  \notag \\
\ma{W}_{\rm mat} = &
\left(\ma{\Sigma}_{\rm s}^{-1} \cdot \ma{V}_{\rm s}^\trans \right) \kron \left( {\ma{U}}_{\rm n}  \cdot {\ma{U}}_{\rm n}^\herm \right),   \notag    
\end{align}
\fi
which follows directly by applying property~\eqref{eqn_veckron} 
to \eqref{eqn_perf_1dse_expl} and to \eqref{eqn_perf_linexp_vac}.
%
%
In order to expand $\expvof{(\Delta \mu_k^{(r)})^2}$ using~\eqref{eqn_app_proof_pert_mse_ir1},
\reva{we} observe that for arbitrary complex vectors $\ma{z}_1,\ma{z}_2$ we have
%
   $\imagof{\ma{z}_1^\trans \cdot \ma{z}_2}  = \imagof{\ma{z}_1}^\trans \cdot \realof{\ma{z}_2} 
   + \realof{\ma{z}_1}^\trans \cdot \imagof{\ma{z}_2}$ 
   and hence
\ifCLASSOPTIONdraftcls
\begin{align}   
   \imagof{\ma{z}_1^\trans \cdot \ma{z}_2}^2 
     & = \imagof{\ma{z}_1}^\trans \cdot \realof{\ma{z}_2} \cdot \realof{\ma{z}_2}^\trans \cdot \imagof{\ma{z}_1} 
       + \realof{\ma{z}_1}^\trans \cdot \imagof{\ma{z}_2} \cdot \imagof{\ma{z}_2}^\trans \cdot \realof{\ma{z}_1} \notag \\ &
       + \imagof{\ma{z}_1}^\trans \cdot \realof{\ma{z}_2} \cdot \imagof{\ma{z}_2}^\trans \cdot \realof{\ma{z}_1} 
       + \realof{\ma{z}_1}^\trans \cdot \imagof{\ma{z}_2} \cdot \realof{\ma{z}_2}^\trans \cdot \imagof{\ma{z}_1} 
       \label{eqn_app_proof_perf_imz1z2sq}
\end{align}
\else
   \begin{align}   
    \begin{split}
   \imagof{\ma{z}_1^\trans \cdot \ma{z}_2}^2 
     & = \imagof{\ma{z}_1}^\trans \cdot \realof{\ma{z}_2} \cdot \realof{\ma{z}_2}^\trans \cdot \imagof{\ma{z}_1} \\
     & + \realof{\ma{z}_1}^\trans \cdot \imagof{\ma{z}_2} \cdot \imagof{\ma{z}_2}^\trans \cdot \realof{\ma{z}_1} \\
     & + \imagof{\ma{z}_1}^\trans \cdot \realof{\ma{z}_2} \cdot \imagof{\ma{z}_2}^\trans \cdot \realof{\ma{z}_1} \\
     & + \realof{\ma{z}_1}^\trans \cdot \imagof{\ma{z}_2} \cdot \realof{\ma{z}_2}^\trans \cdot \imagof{\ma{z}_1}       
     \end{split}
   \end{align}
   \label{eqn_app_proof_perf_imz1z2sq}
\fi
Using~\eqref{eqn_app_proof_pert_mse_ir1} in $\expvof{(\Delta \mu_k^{(r)})^2}$ and 
applying~\eqref{eqn_app_proof_perf_imz1z2sq} for $\ma{z}_1^\trans = \ma{r}_k^{(r)^\trans} \cdot \ma{W}_{\rm mat}$
and $\ma{z}_2 = \vecof{\ma{N}} = \ma{n}$ we find
\ifCLASSOPTIONdraftcls
\begin{align}
   & \expvof{(\Delta \mu_k^{(r)})^2}  \label{eqn_app_proof_pert_mse_ir2}
 \\ &   
 = \expvof{\imagof{\ma{z}_1^\trans} \cdot \realof{{\ma{n}}} \cdot \realof{{\ma{n}}}^\trans \cdot \imagof{\ma{z}_1}} 
 + \expvof{\realof{\ma{z}_1^\trans} \cdot \imagof{{\ma{n}}} \cdot \imagof{{\ma{n}}}^\trans \cdot \realof{\ma{z}_1}} \notag \\ &
 + \expvof{\imagof{\ma{z}_1^\trans} \cdot \realof{{\ma{n}}} \cdot \imagof{{\ma{n}}}^\trans \cdot \realof{\ma{z}_1}} 
 + \expvof{\realof{\ma{z}_1^\trans} \cdot \imagof{{\ma{n}}} \cdot \realof{{\ma{n}}}^\trans \cdot \imagof{\ma{z}_1}} 
 \notag
\end{align}
\else
\begin{align}
   \expvof{(\Delta \mu_k^{(r)})^2} & 
 = \expvof{\imagof{\ma{z}_1^\trans} \cdot \realof{{\ma{n}}} \cdot \realof{{\ma{n}}}^\trans \cdot \imagof{\ma{z}_1}} \notag \\ &
 + \expvof{\realof{\ma{z}_1^\trans} \cdot \imagof{{\ma{n}}} \cdot \imagof{{\ma{n}}}^\trans \cdot \realof{\ma{z}_1}} \notag \\ &
 + \expvof{\imagof{\ma{z}_1^\trans} \cdot \realof{{\ma{n}}} \cdot \imagof{{\ma{n}}}^\trans \cdot \realof{\ma{z}_1}} \notag \\ &
 + \expvof{\realof{\ma{z}_1^\trans} \cdot \imagof{{\ma{n}}} \cdot \realof{{\ma{n}}}^\trans \cdot \imagof{\ma{z}_1}} 
 \label{eqn_app_proof_pert_mse_ir2}
\end{align}
\fi
Since the only random quantity in~\eqref{eqn_app_proof_pert_mse_ir2} is the vector of noise samples $\ma{n}$,
we can move $\ma{z}_1$
out of the expectation operator.
We are then left with the covariance matrices of the real part and the imaginary part of the noise, respectively,
as well as with the cross-covariance matrix between the real and the imaginary part. 
To proceed we require the following lemma:
\begin{lem}\label{lem_cov_rv_noncirc}
Let $\ma{n}$ be a zero mean random vector with covariance matrix
$\ma{R}_{\rm nn} = \expvof{\ma{n} \cdot \ma{n}^\herm}$
and pseudo-covariance matrix
$\ma{C}_{\rm nn} = \expvof{\ma{n} \cdot \ma{n}^\trans}$.
Then, the covariance matrices of the real part of $\ma{n}$, the imaginary part of $\ma{n}$
and the cross-covariance between the real and the imaginary part of $\ma{n}$ are given by
	\begin{align}
	  \ma{R}_{\rm nn}^{(\rm R,R)} \eqdef
	 \expvof{\realof{{\ma{n}}} \cdot \realof{{\ma{n}}}^\trans} &= \frac{1}{2} \realof{\ma{R}_{\rm nn} + \ma{C}_{\rm nn}} 
	 \notag \\
	  \ma{R}_{\rm nn}^{(\rm I,I)} \eqdef
	 \expvof{\imagof{{\ma{n}}} \cdot \imagof{{\ma{n}}}^\trans} &= \frac{1}{2} \realof{\ma{R}_{\rm nn} - \ma{C}_{\rm nn}} 	 \notag \\
	  \ma{R}_{\rm nn}^{(\rm R,I)} \eqdef
	 \expvof{\realof{{\ma{n}}} \cdot \imagof{{\ma{n}}}^\trans} &= -\frac{1}{2} \imagof{\ma{R}_{\rm nn} -\ma{C}_{\rm nn}} 
	 \notag \\
	  \ma{R}_{\rm nn}^{(\rm I,R)} \eqdef
	 \expvof{\imagof{{\ma{n}}} \cdot \realof{{\ma{n}}}^\trans} &= \frac{1}{2} \imagof{\ma{R}_{\rm nn} + \ma{C}_{\rm nn}}.
	 \notag
	\end{align}
\end{lem}
\begin{proof}
   To prove this Lemma we expand 
   $\ma{R}_{\rm nn}$ and $\ma{C}_{\rm nn}$ by inserting
   $\ma{n} = \realof{\ma{n}} + \j \imagof{\ma{n}}$. 
   We then obtain 
   \begin{align}
      \ma{R}_{\rm nn} & = \ma{R}_{\rm nn}^{(\rm R,R)} + \ma{R}_{\rm nn}^{(\rm I,I)}
                       + \j \left(\ma{R}_{\rm nn}^{(\rm I,R)} - \ma{R}_{\rm nn}^{(\rm R,I)}\right) 
                       \label{eqn_app_proof_lemcovnonc_ir1} \\
      \ma{C}_{\rm nn} & = \ma{R}_{\rm nn}^{(\rm R,R)} - \ma{R}_{\rm nn}^{(\rm I,I)}
                       + \j \left(\ma{R}_{\rm nn}^{(\rm I,R)} + \ma{R}_{\rm nn}^{(\rm R,I)}\right).
                       \label{eqn_app_proof_lemcovnonc_ir2} 
   \end{align}
   Since $\ma{R}_{\rm nn}^{(\rm R,R)}$, $\ma{R}_{\rm nn}^{(\rm R,I)}$, 
   $\ma{R}_{\rm nn}^{(\rm I,R)}$, and $\ma{R}_{\rm nn}^{(\rm I,I)}$ are real-valued, the 
   solution of \eqref{eqn_app_proof_lemcovnonc_ir1} and \eqref{eqn_app_proof_lemcovnonc_ir2} 
   is straightforward.
\end{proof}


Using Lemma~\ref{lem_cov_rv_noncirc} in~\eqref{eqn_app_proof_pert_mse_ir2}
we obtain
for the mean square error
\ifCLASSOPTIONdraftcls
\begin{align}
   \expvof{(\Delta \mu_k^{(r)})^2} & 
 = \frac{1}{2} \Big(
   {\imagof{\ma{z}_1^\trans} \cdot \realof{\ma{R}_{\rm nn} + \ma{C}_{\rm nn}} \cdot \imagof{\ma{z}_1}} 
 + {\realof{\ma{z}_1^\trans} \cdot \realof{\ma{R}_{\rm nn} - \ma{C}_{\rm nn}} \cdot \realof{\ma{z}_1}} \notag \\ &
 + {\imagof{\ma{z}_1^\trans} \cdot \imagof{-\ma{R}_{\rm nn} +\ma{C}_{\rm nn}} \cdot \realof{\ma{z}_1}} 
 + {\realof{\ma{z}_1^\trans} \cdot \imagof{\ma{R}_{\rm nn} + \ma{C}_{\rm nn}} \cdot \imagof{\ma{z}_1}} \Big)
 \label{eqn_app_proof_pert_mse_ir3}
\end{align}
\else
\begin{align}
   \expvof{(\Delta \mu_k^{(r)})^2} & 
 = \frac{1}{2} \Big(
   {\imagof{\ma{z}_1^\trans} \cdot \realof{\ma{R}_{\rm nn} + \ma{C}_{\rm nn}} \cdot \imagof{\ma{z}_1}} \notag \\ &
 + {\realof{\ma{z}_1^\trans} \cdot \realof{\ma{R}_{\rm nn} - \ma{C}_{\rm nn}} \cdot \realof{\ma{z}_1}} \notag \\ &
 + {\imagof{\ma{z}_1^\trans} \cdot \imagof{-\ma{R}_{\rm nn} +\ma{C}_{\rm nn}} \cdot \realof{\ma{z}_1}} \notag \\ &
 + {\realof{\ma{z}_1^\trans} \cdot \imagof{\ma{R}_{\rm nn} + \ma{C}_{\rm nn}} \cdot \imagof{\ma{z}_1}} \Big)
 \label{eqn_app_proof_pert_mse_ir3}
\end{align}
\fi

Finally, \eqref{eqn_app_proof_pert_mse_ir3} can be expressed in more compact form
as
\begin{align}
   \expvof{(\Delta \mu_k^{(r)})^2} & = \frac{1}{2}
   \left( \ma{z}_1^\herm \cdot \ma{R}_{\rm nn}^\trans \cdot \ma{z}_1
   - \realof{\ma{z}_1^\trans \cdot \ma{C}_{\rm nn}^\trans \cdot \ma{z}_1} \right)
\end{align}
for $\ma{z}_1 = \ma{W}_{\rm mat}^\trans \cdot \ma{r}_k^{(r)}$,
which is the desired result. \qed

Note that Gaussianity is not needed for these properties to hold. Consequently, the MSE
expressions are still valid if the noise is not Gaussian. We only need it to be zero
mean.
Also, note that for the special case of circularly symmetric white noise we have
$\ma{R}_{\rm nn} = \sigma_{\rm n}^2 \cdot \ma{I}_{MN}$
and $\ma{C}_{\rm nn} = \ma{0}_{MN \times MN}$ and hence the MSE simplifies into

\ifCLASSOPTIONdraftcls
\begin{align}
   \expvof{(\Delta \mu_k^{(r)})^2} 
 & = \frac{\sigma_{\rm n}^2}{2} \cdot 
    \twonorm{\ma{z}_1}^2 =  \frac{\sigma_{\rm n}^2}{2} \cdot 
    \twonorm{\ma{W}_{\rm mat}^\trans \cdot \ma{r}_k^{(r)}}^2.
 \label{eqn_app_proof_pert_mse_whitenoise}
\end{align}
\else
\begin{align}
 &  \expvof{(\Delta \mu_k^{(r)})^2} 
    = \frac{\sigma_{\rm n}^2}{2} \cdot 
    \twonorm{\ma{z}_1}^2 =  \frac{\sigma_{\rm n}^2}{2} \cdot 
    \twonorm{\ma{W}_{\rm mat}^\trans \cdot \ma{r}_k^{(r)}}^2.
 \label{eqn_app_proof_pert_mse_whitenoise}
\end{align}
\fi

The procedure for 2-D Standard Tensor-ESPRIT is in fact quite similar. 
The first step is to express the estimation error in $\mu_k^{(r)}$ in terms
of the perturbation 
$\ma{n} = \vecof{\ma{N}} = \vecof{\unf{N}{3}^\trans}$ {(cf.~\eqref{eqn_subsb_dm_matrix_rd})}. This
expression takes the form
\ifCLASSOPTIONdraftcls
\begin{align}
    \Delta \mu_k^{(r)} = \imagof{
       \ma{r}_k^{(r)^\trans}
       \cdot
       \vecof{
       \unfnot{\Delta\sig{\ten{\hat{U}}}}{R+1}^\trans
       }
    } + \bigO{\Delta^2}
    =  \imagof{
       \ma{r}_k^{(r)^\trans}
       \cdot
       \ma{W}_{\rm ten} \cdot 
       \vecof{\ma{N}}
    } + \bigO{\Delta^2} \label{eqn_app_proof_pert_mste_ir1}
\end{align}
\else
\begin{align}
    \Delta \mu_k^{(r)} & = \imagof{
       \ma{r}_k^{(r)^\trans}
       \cdot
       \vecof{
       \unfnot{\Delta\sig{\ten{\hat{U}}}}{R+1}^\trans
       }
    } + \bigO{\Delta^2} \notag \\
    & =  \imagof{
       \ma{r}_k^{(r)^\trans}
       \cdot
       \ma{W}_{\rm ten} \cdot 
       \vecof{\ma{N}}
    } + \bigO{\Delta^2} \label{eqn_app_proof_pert_mste_ir1}
\end{align}
\fi
since $\unfnot{\Delta\sig{\ten{\hat{U}}}}{R+1}^\trans$ depends linearly on 
$\vecof{\ma{N}}$. Due to the fact that~\eqref{eqn_app_proof_pert_mste_ir1} has the same form
as~\eqref{eqn_app_proof_pert_mse_ir1}, the second step to expand the MSE expressions
follows the same lines as for $R$-D Standard ESPRIT, which immediately shows that the
MSE becomes 
\begin{align}
   \expvof{(\Delta \mu_k^{(r)})^2} & = \frac{1}{2}
   \left( \ma{z}_1^\herm \cdot \ma{R}_{\rm nn}^\trans \cdot \ma{z}_1
   - \realof{\ma{z}_1^\trans \cdot \ma{C}_{\rm nn}^\trans \cdot \ma{z}_1} \right)
\end{align}
for $\ma{z}_1 = \ma{W}_{\rm ten}^\trans \cdot \ma{r}_k^{(r)}$.
%
%
Therefore, the final step is finding an explicit expression for $\ma{W}_{\rm ten}$
which satisfies
\begin{align}
   \unfnot{\Delta\sig{\ten{\hat{U}}}}{3}^\trans
   = 
   \ma{W}_{\rm ten} \cdot 
       \vecof{\ma{N}}
   + \bigO{\Delta^2}.
\end{align}
Recall from Theorem~\ref{thm_perf_exptensub} that
for $R=2$, the HOSVD-based \reva{signal} subspace estimation error
$\unfnot{\Delta\sig{\ten{\hat{U}}}}{R+1}$ can be expanded into
\ifCLASSOPTIONdraftcls
	  \begin{align}
        \unfnot{\Delta\sig{\ten{\hat{U}}}}{3}^\trans
        = 
           \left(
           {\ma{T}}_1 \kron {\ma{T}}_2 
           \right) \cdot \Delta \ma{{U}}_{\rm s}
            + 
           \left(
           \left[\Delta\sig{\ma{U}}_1 \cdot \sigH{\ma{U}}_1\right]
           \kron {\ma{T}}_2 
           \right) \cdot \ma{{U}}_{\rm s} 
            +
           \left(
           {\ma{T}}_1           
           \kron \left[\Delta\sig{\ma{U}}_2 \cdot \sigH{\ma{U}}_2\right]
           \right) \cdot \ma{{U}}_{\rm s} 
            + \bigO{\Delta^2}, \label{eqn_app_proof_perf_exptensubr2}
    \end{align}
\else
    \begin{align}
        \unfnot{\Delta\sig{\ten{\hat{U}}}}{3}^\trans
        = &
           \left(
           {\ma{T}}_1 \kron {\ma{T}}_2 
           \right) \cdot \Delta \ma{{U}}_{\rm s}
            + 
           \left(
           \left[\Delta\sig{\ma{U}}_1 \cdot \sigH{\ma{U}}_1\right]
           \kron {\ma{T}}_2 
           \right) \cdot \ma{{U}}_{\rm s} \notag \\
           & +
           \left(
           {\ma{T}}_1           
           \kron \left[\Delta\sig{\ma{U}}_2 \cdot \sigH{\ma{U}}_2\right]
           \right) \cdot \ma{{U}}_{\rm s} 
            + \bigO{\Delta^2}, \label{eqn_app_proof_perf_exptensubr2}
    \end{align}
\fi    
where $\Delta \ma{{U}}_{\rm s}$, $\Delta\sig{\ma{U}}_1$, and $\Delta\sig{\ma{U}}_2$
are given by
\ifCLASSOPTIONdraftcls
   \begin{align}
     \Delta \ma{U}_{\rm s} & =      
   {\ma{U}}_{\rm n}  \cdot {\ma{U}}_{\rm n}^\herm \cdot \ma{N} \cdot \ma{V}_{\rm s} \cdot \ma{\Sigma}_{\rm s}^{-1} 
   = \noiC{\ma{V}}_3 \cdot \noiT{\ma{V}}_3 \cdot \ma{N} \cdot \sigC{\ma{U}}_3 \cdot \siginv{\ma{\Sigma}}_3 
   \quad \mbox{and} \notag \\
     \Delta \sig{\ma{U}}_r & = \noi{\ma{U}}_r \cdot \noiH{\ma{U}}_r \cdot \unf{N}{r} \cdot \sig{\ma{V}}_r \cdot \siginv{\ma{\Sigma}}_r \quad \mbox{for $r=1, 2$}. \label{eqn_app_proof_perf_exptensubdelurs}
   \end{align}
\else
   \begin{align}
     \Delta \ma{U}_{\rm s} & =      
   {\ma{U}}_{\rm n}  \cdot {\ma{U}}_{\rm n}^\herm \cdot \ma{N} \cdot \ma{V}_{\rm s} \cdot \ma{\Sigma}_{\rm s}^{-1} 
   \notag \\
   & = \noiC{\ma{V}}_3 \cdot \noiT{\ma{V}}_3 \cdot \ma{N} \cdot \sigC{\ma{U}}_3 \cdot \siginv{\ma{\Sigma}}_3 
   \quad \mbox{and} \notag \\
     \Delta \sig{\ma{U}}_r & = \noi{\ma{U}}_r \cdot \noiH{\ma{U}}_r \cdot \unf{N}{r} \cdot \sig{\ma{V}}_r \cdot \siginv{\ma{\Sigma}}_r \quad \mbox{for $r=1, 2$}. \label{eqn_app_proof_perf_exptensubdelurs}
   \end{align}
\fi    
The first term in~\eqref{eqn_app_proof_perf_exptensubr2} is easily vectorized 
by applying property~\eqref{eqn_veckron} which yields the first term
of $\ma{W}_{\rm ten}$ as
\ifCLASSOPTIONdraftcls
 	\begin{align}
	\vecof{\left({\ma{T}}_1 \kron {\ma{T}}_2 \right) \cdot \Delta \ma{{U}}_{\rm s}}
	& = 
	\vecof{\left({\ma{T}}_1 \kron {\ma{T}}_2 \right) \cdot 
	\noiC{\ma{V}}_3 \cdot \noiT{\ma{V}}_3 \cdot \ma{N} \cdot \sigC{\ma{U}}_3 \cdot \siginv{\ma{\Sigma}}_3}
	 \notag \\
	& = 
	\left(\sigC{\ma{U}}_3 \cdot \siginv{\ma{\Sigma}}_3\right)^\trans \kron 
	\left[
	\left({\ma{T}}_1 \kron {\ma{T}}_2 \right) \cdot 
	\noiC{\ma{V}}_3 \cdot \noiT{\ma{V}}_3 \right] \reva{\cdot \vecof{\ma{N}}} \notag \\
	& = 
	\left(\siginv{\ma{\Sigma}}_3 \cdot \sigH{\ma{U}}_3 \right) \kron 
	\left[
	\left({\ma{T}}_1 \kron {\ma{T}}_2 \right) \cdot 
	\noiC{\ma{V}}_3 \cdot \noiT{\ma{V}}_3 \right] \reva{\cdot \vecof{\ma{N}}}.
	\notag
	\end{align} 
\else
	\begin{align}
	& \vecof{\left({\ma{T}}_1 \kron {\ma{T}}_2 \right) \cdot \Delta \ma{{U}}_{\rm s}} \notag \\
	& = 
	\vecof{\left({\ma{T}}_1 \kron {\ma{T}}_2 \right) \cdot 
	\noiC{\ma{V}}_3 \cdot \noiT{\ma{V}}_3 \cdot \ma{N} \cdot \sigC{\ma{U}}_3 \cdot \siginv{\ma{\Sigma}}_3}
	 \notag \\
	& = 
	\left(\sigC{\ma{U}}_3 \cdot \siginv{\ma{\Sigma}}_3\right)^\trans \kron 
	\left[
	\left({\ma{T}}_1 \kron {\ma{T}}_2 \right) \cdot 
	\noiC{\ma{V}}_3 \cdot \noiT{\ma{V}}_3 \right] \reva{\cdot \vecof{\ma{N}}} \notag \\
	& = 
	\left(\siginv{\ma{\Sigma}}_3 \cdot \sigH{\ma{U}}_3 \right) \kron 
	\left[
	\left({\ma{T}}_1 \kron {\ma{T}}_2 \right) \cdot 
	\noiC{\ma{V}}_3 \cdot \noiT{\ma{V}}_3 \right] \reva{\cdot \vecof{\ma{N}}}.
	\notag
	\end{align} 
\fi
However, for the second term in~\eqref{eqn_app_proof_perf_exptensubr2} we get
\ifCLASSOPTIONdraftcls
\begin{align}
  & \vecof{  \left(
           \left[\noi{\ma{U}}_1 \cdot \noiH{\ma{U}}_1 \cdot \unf{N}{1} \cdot \sig{\ma{V}}_1 \cdot \siginv{\ma{\Sigma}}_1 \cdot \sigH{\ma{U}}_1\right]
           \kron {\ma{T}}_2 
           \right) \cdot \ma{{U}}_{\rm s} } \notag \\
   = & \left( \ma{{U}}_{\rm s}^\trans \kron \ma{I}_M \right)
     \cdot 
    \vecof{
           \left[\noi{\ma{U}}_1 \cdot \noiH{\ma{U}}_1 \cdot \unf{N}{1} \cdot \sig{\ma{V}}_1 \cdot \siginv{\ma{\Sigma}}_1 \cdot \sigH{\ma{U}}_1\right]
           \kron {\ma{T}}_2 
           } \notag
\end{align}
\else
\begin{align}
  & \vecof{  \left(
           \left[\noi{\ma{U}}_1 \cdot \noiH{\ma{U}}_1 \cdot \unf{N}{1} \cdot \sig{\ma{V}}_1 \cdot \siginv{\ma{\Sigma}}_1 \cdot \sigH{\ma{U}}_1\right]
           \kron {\ma{T}}_2 
           \right) \cdot \ma{{U}}_{\rm s} } \notag \\
   = & \left( \ma{{U}}_{\rm s}^\trans \kron \ma{I}_M \right)
     \cdot  \notag \\ &
    \vecof{
           \left[\noi{\ma{U}}_1 \cdot \noiH{\ma{U}}_1 \cdot \unf{N}{1} \cdot \sig{\ma{V}}_1 \cdot \siginv{\ma{\Sigma}}_1 \cdot \sigH{\ma{U}}_1\right]
           \kron {\ma{T}}_2 
           } \notag
\end{align}
\fi
\reva{by inserting~\eqref{eqn_app_proof_perf_exptensubdelurs} for $\Delta \sig{\ma{U}}_1$.}
To proceed we need to rewrite the vectorization of a Kronecker product. 
After straightforward calculations we obtain
\ifCLASSOPTIONdraftcls
\begin{align}
& \left( \ma{{U}}_{\rm s}^\trans \kron \ma{I}_M \right)
     \cdot 
    \vecof{
           \left[\noi{\ma{U}}_1 \cdot \noiH{\ma{U}}_1 \cdot \unf{N}{1} \cdot \sig{\ma{V}}_1 \cdot \siginv{\ma{\Sigma}}_1 \cdot \sigH{\ma{U}}_1\right]
           \kron {\ma{T}}_2 
           } \notag \\ 
 = & \left( \ma{{U}}_{\rm s}^\trans \kron \ma{I}_M \right)
     \cdot {\ma{\bar{T}}}_2 \cdot
    \vecof{
           \left[\noi{\ma{U}}_1 \cdot \noiH{\ma{U}}_1 \cdot \unf{N}{1} \cdot \sig{\ma{V}}_1 \cdot \siginv{\ma{\Sigma}}_1 \cdot \sigH{\ma{U}}_1\right]
           } \\
 = & \left( \ma{{U}}_{\rm s}^\trans \kron \ma{I}_M \right)
     \cdot {\ma{\bar{T}}}_2 \cdot
     \left[
     \left( \sig{\ma{V}}_1 \cdot \siginv{\ma{\Sigma}}_1 \cdot \sigH{\ma{U}}_1 \right)^\trans
     \kron
     \left(\noi{\ma{U}}_1 \cdot \noiH{\ma{U}}_1 \right)
     \right]
    \vecof{\unf{N}{1}} \notag
\end{align}
\else
\begin{align}
& \left( \ma{{U}}_{\rm s}^\trans \kron \ma{I}_M \right)
     \cdot \notag \\ &
    \vecof{
           \left[\noi{\ma{U}}_1 \cdot \noiH{\ma{U}}_1 \cdot \unf{N}{1} \cdot \sig{\ma{V}}_1 \cdot \siginv{\ma{\Sigma}}_1 \cdot \sigH{\ma{U}}_1\right]
           \kron {\ma{T}}_2 
           } \notag \\ 
 = & \left( \ma{{U}}_{\rm s}^\trans \kron \ma{I}_M \right)
     \cdot {\ma{\bar{T}}}_2 \cdot \notag \\ &
    \vecof{
           \left[\noi{\ma{U}}_1 \cdot \noiH{\ma{U}}_1 \cdot \unf{N}{1} \cdot \sig{\ma{V}}_1 \cdot \siginv{\ma{\Sigma}}_1 \cdot \sigH{\ma{U}}_1\right]
           } \\
 = & \left( \ma{{U}}_{\rm s}^\trans \kron \ma{I}_M \right)
     \cdot {\ma{\bar{T}}}_2 \cdot \notag \\ &
     \left[
     \left( \sig{\ma{V}}_1 \cdot \siginv{\ma{\Sigma}}_1 \cdot \sigH{\ma{U}}_1 \right)^\trans
     \kron
     \left(\noi{\ma{U}}_1 \cdot \noiH{\ma{U}}_1 \right)
     \right]
    \vecof{\unf{N}{1}} \notag
\end{align}
\fi
where the matrix $\ma{\bar{T}}_2$ is constructed from the columns of $\ma{T}_2$ given by
$\ma{t}_{2,m}$ for $m=1, 2, \ldots, M_2$ in the following manner
\begin{align}
\ma{\bar{T}}_2  = 
\ma{I}_{M_1} \kron 
\begin{bmatrix}
\ma{I}_{M_1} \kron \ma{t}_{2,1}  \\
\vdots \\
\ma{I}_{M_1} \kron \ma{t}_{2,M_2}
\end{bmatrix}.
\end{align}
The final step is to rearrange the elements of $\vecof{\unf{N}{1}}$
so that they appear in the same order as in $\vecof{\ma{N}}$. 
However, since $\ma{N} = \unf{N}{3}^\trans$, this can easily be achieved
in the following manner
   \begin{align}
       \vecof{\unf{N}{1}} = \ma{K}_{M_2 \times (M_1\cdot N)} \cdot \vecof{\ma{N}}
   \end{align}
where $\ma{K}_{M_2 \times (M_1\cdot N)}$ is the commutation matrix (cf. equation~\eqref{eqn_def_commat}).
This completes the derivation of the second term of $\ma{W}_{\rm ten}$. The third
term is obtained in a similar manner. In this case, no permutation is needed, since
$\vecof{\unf{N}{2}} = \vecof{\unf{N}{3}^\trans} = \vecof{\ma{N}}$ .
\qed

\section{Proof of Theorem~\ref{thm_perf_mse_fba}}\label{sec_app_proof_perf_mse_fba}

As pointed out in Section~\ref{subsec_subsp_perf_fba}, the inclusion of Forward-Backward-Averaging
leads to a very similar model, where all quantities originating from the noise-free
observation $\ma{X}_0$ (or $\ten{X}_0$) are replaced by the corresponding quantities
for $\fba{\ma{X}}_0$ (or $\fba{\ten{X}}_0$). 
\revA{Since for Theorem~\ref{thm_perf_mse} it was only assumed that the desired signal
component is superimposed by a zero mean noise contribution,
it is directly applicable.

The only point we need to derive are the covariance matrix and the pseudo-covariance
matrix of the forward-backward averaged noise
$\fba{\ma{n}} \eqdef \vecof{\fba{\ma{N}}}$, which are needed for the MSE expressions. }
To this end, we can express $\fba{\ma{n}}$ as
\ifCLASSOPTIONdraftcls
\begin{align}
   \vecof{\fba{\ma{N}}}
   & = 
   \vecof{\begin{bmatrix} \ma{N}, & \ma{\Pi}_M \cdot \ma{N}^\conj \cdot \ma{\Pi}_N\end{bmatrix}} 
   \notag \\
   & = 
   \begin{bmatrix}
       \vecof{\ma{N}} \\
       \left( \ma{\Pi}_N \kron \ma{\Pi}_M \right) \cdot \vecof{\ma{N}^\conj} 
   \end{bmatrix} 
    = 
   \begin{bmatrix}
       {\ma{n}} \\
       \ma{\Pi}_{NM}  \cdot \ma{n}^\conj
   \end{bmatrix} 
   \label{eqn_app_proof_msefba_exp1}
\end{align}
\else
\begin{align}
   \vecof{\fba{\ma{N}}}
   & = 
   \vecof{\begin{bmatrix} \ma{N}, & \ma{\Pi}_M \cdot \ma{N}^\conj \cdot \ma{\Pi}_N\end{bmatrix}} 
   \notag \\
   & = 
   \begin{bmatrix}
       \vecof{\ma{N}} \\
       \left( \ma{\Pi}_N \kron \ma{\Pi}_M \right) \cdot \vecof{\ma{N}^\conj} 
   \end{bmatrix} \notag \\ &
    = 
   \begin{bmatrix}
       {\ma{n}} \\
       \ma{\Pi}_{NM}  \cdot \ma{n}^\conj
   \end{bmatrix} 
   \label{eqn_app_proof_msefba_exp1}
\end{align}
\fi
Equation~\eqref{eqn_app_proof_msefba_exp1} allows us to express the 
\revA{covariance matrix and the pseudo-covariance matrix of $\fba{\ma{n}}$
via the covariance matrix and the pseudo-covariance matrix of $\ma{n}$.
We obtain
\ifCLASSOPTIONdraftcls
\begin{align}
   \expvof{\fba{\ma{n}} \cdot \fbaH{\ma{n}}}
   & = 
   \begin{bmatrix}
     \expvof{\ma{n} \cdot \ma{n}^\herm} &
        \expvof{\ma{n} \cdot \ma{n}^\trans} \cdot \ma{\Pi}_{MN} \\
     \ma{\Pi}_{MN} \cdot  \expvof{\ma{n}^\conj \cdot \ma{n}^\herm} & 
        \ma{\Pi}_{MN} \cdot  \expvof{\ma{n}^\conj \cdot \ma{n}^\trans} \cdot \ma{\Pi}_{MN}
   \end{bmatrix} 
   \notag \\
   & = 
   \begin{bmatrix}
     \ma{R}_{\rm nn} &
        \ma{C}_{\rm nn} \cdot \ma{\Pi}_{MN} \\
     \ma{\Pi}_{MN} \cdot \ma{C}_{\rm nn}^\conj & 
        \ma{\Pi}_{MN} \cdot  \ma{R}_{\rm nn}^\conj \cdot \ma{\Pi}_{MN}
   \end{bmatrix} 
   \notag \\
   \expvof{\fba{\ma{n}} \cdot \fbaT{\ma{n}}}   
   & = \begin{bmatrix}
     \expvof{\ma{n} \cdot \ma{n}^\trans} &
        \expvof{\ma{n} \cdot \ma{n}^\herm} \cdot \ma{\Pi}_{MN} \\
     \ma{\Pi}_{MN} \cdot  \expvof{\ma{n}^\conj \cdot \ma{n}^\trans} & 
        \ma{\Pi}_{MN} \cdot  \expvof{\ma{n}^\conj \cdot \ma{n}^\herm} \cdot \ma{\Pi}_{MN}
   \end{bmatrix} 
   \notag \\ 
   &  = 
   \begin{bmatrix}
     \ma{C}_{\rm nn} &
        \ma{R}_{\rm nn} \cdot \ma{\Pi}_{MN} \\
     \ma{\Pi}_{MN} \cdot \ma{R}_{\rm nn}^\conj & 
        \ma{\Pi}_{MN} \cdot  \ma{C}_{\rm nn}^\conj \cdot \ma{\Pi}_{MN}
   \end{bmatrix} \notag
  \end{align}
\else
\begin{align}
   & \expvof{\fba{\ma{n}} \cdot \fbaH{\ma{n}}}  \notag \\ &
   = \begin{bmatrix}
     \expvof{\ma{n} \cdot \ma{n}^\herm} &
        \expvof{\ma{n} \cdot \ma{n}^\trans} \cdot \ma{\Pi}_{MN} \\
     \ma{\Pi}_{MN} \cdot  \expvof{\ma{n}^\conj \cdot \ma{n}^\herm} & 
        \ma{\Pi}_{MN} \cdot  \expvof{\ma{n}^\conj \cdot \ma{n}^\trans} \cdot \ma{\Pi}_{MN}
   \end{bmatrix} \notag \\ &
   = 
   \begin{bmatrix}
     \ma{R}_{\rm nn} &
        \ma{C}_{\rm nn} \cdot \ma{\Pi}_{MN} \\
     \ma{\Pi}_{MN} \cdot \ma{C}_{\rm nn}^\conj & 
        \ma{\Pi}_{MN} \cdot  \ma{R}_{\rm nn}^\conj \cdot \ma{\Pi}_{MN}
   \end{bmatrix} \notag \\
   & \expvof{\fba{\ma{n}} \cdot \fbaT{\ma{n}}}  \notag \\ &
   = \begin{bmatrix}
     \expvof{\ma{n} \cdot \ma{n}^\trans} &
        \expvof{\ma{n} \cdot \ma{n}^\herm} \cdot \ma{\Pi}_{MN} \\
     \ma{\Pi}_{MN} \cdot  \expvof{\ma{n}^\conj \cdot \ma{n}^\trans} & 
        \ma{\Pi}_{MN} \cdot  \expvof{\ma{n}^\conj \cdot \ma{n}^\herm} \cdot \ma{\Pi}_{MN}
   \end{bmatrix} \notag \\ &
   = 
   \begin{bmatrix}
     \ma{C}_{\rm nn} &
        \ma{R}_{\rm nn} \cdot \ma{\Pi}_{MN} \\
     \ma{\Pi}_{MN} \cdot \ma{R}_{\rm nn}^\conj & 
        \ma{\Pi}_{MN} \cdot  \ma{C}_{\rm nn}^\conj \cdot \ma{\Pi}_{MN}
   \end{bmatrix} \notag
\end{align}
\fi}

This completes the proof of the theorem. \qed

\section{Proof of Theorem~\ref{thm_perf_mse_sls}}\label{sec_app_proof_perf_sls}

Without regularization, the cost function for 1-D Structured Least Squares can
be expressed as~\cite{Haa:97}
\begin{align}
\ma{\hat{\Psi}}_{\rm SLS}
= \argmin_{\ma{\Psi},\Delta\ma{\overline{U}}_{\rm s}}
\fronorm{
       \ma{J}_1 \cdot \left(\ma{\hat{U}}_{\rm s} + \Delta\ma{\overline{U}}_{\rm s} \right) \cdot \ma{\Psi} 
    -  \ma{J}_2 \cdot \left(\ma{\hat{U}}_{\rm s} + \Delta\ma{\overline{U}}_{\rm s} \right)
}^2 \label{eqn_app_proof_perf_sls_cf}.
\end{align}
where we have used $\Delta\ma{\overline{U}}_{\rm s}$ only to avoid confusion
with the $\Delta \ma{U}_{\rm s}$ associated to the estimation error in $\ma{\hat{U}}_{\rm s}$.
Note that the cost function is solved in an iterative manner starting with 
$\reva{\Delta\ma{\overline{U}}_{\rm s}} = \ma{0}_{M \times d}$ and with ${\ma{\Psi}} = \ma{\Psi}_{\rm LS}$,
where $\ma{\Psi}_{\rm LS} = \left(\ma{J}_1 \cdot\ma{\hat{U}}_{\rm s}\right)^\pinv \cdot
\left(\ma{J}_2 \cdot\ma{\hat{U}}_{\rm s}\right)$ represents the LS solution to the shift
invariance equation. As we compute only a single iteration we find one update term for
$\ma{\hat{U}}_{\rm s}$ and one for $\ma{\Psi}_{\rm LS}$ which we denote as 
$\Delta \ma{U}_{\rm s,SLS}$ and $\Delta \ma{\Psi}_{\rm SLS}$ \reva{(i.e.,
$\Delta \ma{U}_{\rm s,SLS}$ represents the $\Delta\ma{\overline{U}}_{\rm s}$
which minimizes the linearized version of~\eqref{eqn_app_proof_perf_sls_cf})}.
In other words, the cost function
becomes
\ifCLASSOPTIONdraftcls
\begin{align}
\ma{\hat{\Psi}}_{\rm SLS} & = \ma{\hat{\Psi}}_{\rm LS} + \Delta \ma{{\Psi}}_{\rm SLS} \quad \mbox{where} \\
\Delta \ma{{\Psi}}_{\rm SLS}&  = \argmin_{\Delta \ma{\Psi},{\Delta\ma{\overline{U}}_{\rm s}}}
\fronorm{
 \ma{J}_1 \cdot \left(\ma{\hat{U}}_{\rm s}+{\Delta\ma{\overline{U}}_{\rm s}}\right) \cdot 
 \left({\ma{\Psi}}_{\rm LS} + \Delta \ma{\Psi}\right)  -  \ma{J}_2 \cdot\left(\ma{\hat{U}}_{\rm s} + {\Delta\ma{\overline{U}}_{\rm s}}\right) 
}^2 \notag \\
& = \argmin_{\Delta \ma{\Psi},{\Delta\ma{U}_{\rm s}}}
\fronorm{
 \ma{R}_{\rm LS} 
 + \ma{J}_1 \cdot \Delta\ma{\overline{U}}_{\rm s} \cdot \ma{\Psi}_{\rm LS}
 + \ma{J}_1 \cdot \ma{\hat{U}}_{\rm s} \cdot \Delta \ma{\Psi}
 - \ma{J}_2 \cdot \Delta\ma{\overline{U}}_{\rm s} 
 + \bigO{\Delta^2} }^2 \label{eqn_app_proof_subsp_sls_cf1}
\end{align}
\else
\begin{align}
\ma{\hat{\Psi}}_{\rm SLS} = & \ma{\hat{\Psi}}_{\rm LS} + \Delta \ma{{\Psi}}_{\rm SLS} \quad \mbox{where} \\
\Delta \ma{{\Psi}}_{\rm SLS}&  = \argmin_{\Delta \ma{\Psi},{\Delta\ma{\overline{U}}_{\rm s}}}
\froonormnlr{
 \ma{J}_1 \cdot \left(\ma{\hat{U}}_{\rm s}+{\Delta\ma{\overline{U}}_{\rm s}}\right) \cdot 
 \left({\ma{\Psi}}_{\rm LS} + \Delta \ma{\Psi}\right) \notag\\& \quad -  \ma{J}_2 \cdot\left(\ma{\hat{U}}_{\rm s} + {\Delta\ma{\overline{U}}_{\rm s}}\right) 
}{\Big}^2 \notag \\
& = \argmin_{\Delta \ma{\Psi},{\Delta\ma{U}_{\rm s}}}
\froonormnlr{
 \ma{R}_{\rm LS} 
 + \ma{J}_1 \cdot \Delta\ma{\overline{U}}_{\rm s} \cdot \ma{\Psi}_{\rm LS}
 + \ma{J}_1 \cdot \ma{\hat{U}}_{\rm s} \cdot \Delta \ma{\Psi} \notag \\ & \quad
 - \ma{J}_2 \cdot \Delta\ma{\overline{U}}_{\rm s} 
 + \bigO{\Delta^2} }{\Big}^2 \label{eqn_app_proof_subsp_sls_cf1}
\end{align}
\fi
where we have defined the matrix $\ma{R}_{\rm LS} =
\ma{J}_1 \cdot \ma{\hat{U}}_{\rm s} \cdot \ma{\Psi}_{\rm LS} - \ma{J}_2 \cdot \ma{\hat{U}}_{\rm s}$ which
contains the residual error in the shift invariance equation after the LS fit.
Since~\eqref{eqn_app_proof_subsp_sls_cf1} \reva{is} linearized \reva{by skipping the
quadratic terms in $\bigO{\Delta^2}$}, it is easily solved by
an LS fit. To express the result in closed-form we vectorize~\eqref{eqn_app_proof_subsp_sls_cf1}
using the fact that $\fronorm{\ma{A}} = \twonorm{\vecof{\ma{A}}}$ and obtain
\ifCLASSOPTIONdraftcls
\begin{align}
   \Delta \ma{{\psi}}_{\rm SLS}&  = \argmin_{\vecof{\Delta \ma{\Psi}},\vecof{{\Delta\ma{\overline{U}}_{\rm s}}}}
    \twonormnlr{
    \ma{r}_{\rm LS} 
    + \left(\ma{\hat{\Psi}}_{\rm LS}^\trans \kron \ma{J}_1  \right) \cdot \vecof{\Delta\ma{\overline{U}}_{\rm s}}
    \notag \\ &
    + \left(\ma{I}_d \kron \left( \ma{J}_1 \cdot \ma{\hat{U}}_{\rm s} \right)\right) \cdot \vecof{\Delta \ma{\Psi}} 
    - \left(\ma{I}_d \kron \ma{J}_2  \right) \cdot \vecof{\Delta\ma{\overline{U}}_{\rm s}}    
    + \bigO{\Delta^2}
    }{\Big}^2 \\
    \Rightarrow 
    \Delta \ma{{\psi}}_{\rm SLS}&  = \argmin_{\vecof{\Delta \ma{\Psi}},\vecof{{\Delta\ma{\overline{U}}_{\rm s}}}}
    \twonorm{
    \ma{r}_{\rm LS} 
    + \ma{\hat{F}}_{\rm SLS} \cdot 
    \begin{bmatrix}
        \vecof{\Delta \ma{\Psi}} \\
        \vecof{\Delta\ma{\overline{U}}_{\rm s}}
    \end{bmatrix}
    + \bigO{\Delta^2}
    }^2 \label{eqn_app_proof_subsp_sls_cfir1} \\
 \Rightarrow 
 \begin{bmatrix}
        \Delta \ma{\psi}_{\rm SLS} \\
        \Delta\ma{{u}}_{\rm s,SLS}
    \end{bmatrix}
    & = 
    - \ma{\hat{F}}_{\rm SLS}^\pinv \cdot \ma{r}_{\rm LS}   \label{eqn_app_proof_subsp_sls_solcf}
\end{align}
\else
\begin{align}
   & \Delta \ma{{\psi}}_{\rm SLS}  = \argmin_{\vecof{\Delta \ma{\Psi}},\vecof{{\Delta\ma{\overline{U}}_{\rm s}}}}
    \twonormnlr{
    \ma{r}_{\rm LS} 
    + \left(\ma{\hat{\Psi}}_{\rm LS}^\trans \kron \ma{J}_1  \right) \cdot \vecof{\Delta\ma{\overline{U}}_{\rm s}}
    \notag \\ & \quad\quad\quad
    + \left(\ma{I}_d \kron \left( \ma{J}_1 \cdot \ma{\hat{U}}_{\rm s} \right)\right) \cdot \vecof{\Delta \ma{\Psi}} 
    \notag \\ & \quad\quad\quad
    - \left(\ma{I}_d \kron \ma{J}_2  \right) \cdot \vecof{\Delta\ma{\overline{U}}_{\rm s}}    
    + \bigO{\Delta^2}
    }{\Big}^2 \\
    & \Rightarrow 
    \Delta \ma{{\psi}}_{\rm SLS}  = \argmin_{\vecof{\Delta \ma{\Psi}},\vecof{{\Delta\ma{\overline{U}}_{\rm s}}}}
    \twonormnlr{
    \ma{r}_{\rm LS} \notag \\ & \quad\quad\quad
    + \ma{\hat{F}}_{\rm SLS} \cdot 
    \begin{bmatrix}
        \vecof{\Delta \ma{\Psi}} \\
        \vecof{\Delta\ma{\overline{U}}_{\rm s}}
    \end{bmatrix}
    + \bigO{\Delta^2}
    }{\Big}^2 \label{eqn_app_proof_subsp_sls_cfir1} \\
 & \Rightarrow 
 \begin{bmatrix}
        \Delta \ma{\psi}_{\rm SLS} \\
        \Delta\ma{{u}}_{\rm s,SLS}
    \end{bmatrix}
    = 
    - \ma{\hat{F}}_{\rm SLS}^\pinv \cdot \ma{r}_{\rm LS}   \label{eqn_app_proof_subsp_sls_solcf}
\end{align}
\fi
where we have introduced the vectorized quantities $\ma{r}_{\rm LS}  = \vecof{\ma{R}_{\rm LS}}$, 
$\Delta \ma{\psi}_{\rm SLS} = \vecof{\Delta \ma{\Psi}_{\rm SLS}}$, and 
$\Delta\ma{{u}}_{\rm s,SLS} = \vecof{\Delta\ma{{U}}_{\rm s,SLS}}$, respectively.
\reva{We have also skipped the quadratic terms $\bigO{\Delta^2}$ in~\eqref{eqn_app_proof_subsp_sls_cfir1}
for the solution in~\eqref{eqn_app_proof_subsp_sls_solcf}.}
Moreover, the matrix $\ma{\hat{F}}_{\rm SLS} \in \compl^{(M-1)d \times (d^2+M\cdot d)}$ becomes
\begin{align}
   \ma{\hat{F}}_{\rm SLS} = \left[
 \ma{I}_d \kron \left( \ma{J}_1 \cdot \ma{\hat{U}}_{\rm s} \right), \;
 \left(\ma{\hat{\Psi}}_{\rm LS}^\trans \kron \ma{J}_1  \right)
-  \left(\ma{I}_d \kron \ma{J}_2\right)
\right].
\end{align}
Therefore, our next goal is to find a first order expansion of $\Delta \ma{\psi}_{\rm SLS}$
in~\eqref{eqn_app_proof_subsp_sls_solcf}. This looks difficult at first
sight as it involves an expansion of a pseudo-inverse due to $\ma{\hat{F}}_{\rm SLS}$. However,
this step simplifies significantly by realizing that $\ma{\hat{F}}_{\rm SLS}$ can be expressed
as $\ma{\hat{F}}_{\rm SLS} = \ma{{F}}_{\rm SLS} + \Delta \ma{\hat{F}}_{\rm SLS}$, where
\begin{align}
    \ma{F}_{\rm SLS} & = {\left[  
         \ma{I}_d \kron \left(\ma{J}_1 \cdot \ma{U}_{\rm s}\right), \;
         \left(\ma{\Psi}^\trans \kron \ma{J}_1\right) - \left(\ma{I}_d \kron \ma{J}_2\right)
         \right]}\label{eqn_app_proof_subsp_slsfsls} \\
   \Delta \ma{F}_{\rm SLS}
 & = {\left[  
         \ma{I}_d \kron \left(\ma{J}_1 \cdot {\Delta\ma{U}_{\rm s}}\right), \;
         \left({\Delta\ma{\Psi}_{\rm LS}^\trans} \kron \ma{J}_1\right) 
         \right]} \notag         
\end{align}
where $\ma{\hat{\Psi}}_{\rm LS} = \ma{\Psi} + \Delta\ma{\Psi}_{\rm LS}$.
Since $ \ma{F}_{\rm SLS}$ is not random (i.e., only dependent on $\ma{X}_0$ but not on $\ma{N}$)
and $\Delta\ma{\Psi}_{\rm LS} = \ma{0}_{d \times d} + \bigO{\Delta}$, i.e., at least linear in the perturbation,
we have
\begin{align}
    \ma{\hat{F}}_{\rm SLS}^\pinv = \ma{F}_{\rm SLS}^\pinv + \bigO{\Delta}. \label{eqn_app_proof_subsp_sls_fpinv}
\end{align}
This relation only describes the ``zero-th'' term of the expansion of $\ma{\hat{F}}_{\rm SLS}^\pinv$.
However, as we see below, the linear term is not needed for a first order expansion
of \reva{$\Delta \ma{\psi}_{\rm SLS}$}. 
Continuing with~\eqref{eqn_app_proof_subsp_sls_solcf}, the second term \reva{of the right-hand side} 
is given by $\ma{r}_{\rm LS}$
for which we can write
\ifCLASSOPTIONdraftcls
\begin{align}
    \ma{r}_{\rm LS} & = \vecof{\ma{J}_1 \cdot \ma{\hat{U}}_{\rm s}  \cdot
    \ma{\hat{\Psi}}_{\rm LS}- \ma{J}_2 \cdot\ma{\hat{U}}_{\rm s} } \notag \\
     & = \vecof{\ma{J}_1 \cdot \left(\ma{{U}}_{\rm s} + \Delta \ma{{U}}_{\rm s}\right) \cdot
    \left(\ma{\Psi} + \Delta\ma{\Psi}_{\rm LS}\right)
     - \ma{J}_2 \cdot \left(\ma{{U}}_{\rm s} + \Delta \ma{{U}}_{\rm s}\right)} \notag\\
      & = \vecof{ \ma{J}_1 \cdot \Delta \ma{{U}}_{\rm s} \cdot \ma{\Psi}
                + \ma{J}_1 \cdot \ma{{U}}_{\rm s} \cdot \Delta\ma{\Psi}_{\rm LS}    
                - \ma{J}_2 \cdot \Delta \ma{{U}}_{\rm s}
                + \bigO{\Delta^2}}  \label{eqn_app_proof_subsp_sls_rls}
\end{align}
\else
\begin{align}
    \ma{r}_{\rm LS}  = & \vecof{\ma{J}_1 \cdot \ma{\hat{U}}_{\rm s}  \cdot
    \ma{\hat{\Psi}}_{\rm LS}- \ma{J}_2 \cdot\ma{\hat{U}}_{\rm s} } \notag \\
      = & \vecofnlr{\ma{J}_1 \cdot \left(\ma{{U}}_{\rm s} + \Delta \ma{{U}}_{\rm s}\right) \cdot
    \left(\ma{\Psi} + \Delta\ma{\Psi}_{\rm LS}\right) \notag \\ & \quad\quad
     - \ma{J}_2 \cdot \left(\ma{{U}}_{\rm s} + \Delta \ma{{U}}_{\rm s}\right)}{} \notag\\
       = & \vecofnlr{ \ma{J}_1 \cdot \Delta \ma{{U}}_{\rm s} \cdot \ma{\Psi}
                + \ma{J}_1 \cdot \ma{{U}}_{\rm s} \cdot \Delta\ma{\Psi}_{\rm LS}    
                 \notag \\ & \quad\quad
                - \ma{J}_2 \cdot \Delta \ma{{U}}_{\rm s}
                + \bigO{\Delta^2}}{}  \label{eqn_app_proof_subsp_sls_rls}
\end{align}
\fi
Moreover, as shown in~\cite{LLV:93}, $\Delta\ma{\Psi}_{\rm LS}$ can be expressed
in terms of $\Delta \ma{{U}}_{\rm s}$ via 
\ifCLASSOPTIONdraftcls
\begin{align}
\Delta\ma{\Psi}_{\rm LS} = (\ma{J}_1 \cdot \ma{U}_{\rm s})^\pinv \cdot \ma{J}_2 \cdot \Delta \ma{U}_{\rm s}
- (\ma{J}_1 \cdot \ma{U}_{\rm s})^\pinv \cdot \ma{J}_1 \cdot \Delta \ma{U}_{\rm s} \cdot \ma{\Psi} + \bigO{\Delta^2}.
\label{eqn_app_proof_subsp_sls_psils}
\end{align}
\else
\begin{align}
\Delta\ma{\Psi}_{\rm LS} = & (\ma{J}_1 \cdot \ma{U}_{\rm s})^\pinv \cdot \ma{J}_2 \cdot \Delta \ma{U}_{\rm s}
 \notag \\ & \quad
- (\ma{J}_1 \cdot \ma{U}_{\rm s})^\pinv \cdot \ma{J}_1 \cdot \Delta \ma{U}_{\rm s} \cdot \ma{\Psi} + \bigO{\Delta^2}.
\label{eqn_app_proof_subsp_sls_psils}
\end{align}

\fi
Using this expansion in~\eqref{eqn_app_proof_subsp_sls_rls} we obtain
\ifCLASSOPTIONdraftcls
\begin{align}
\ma{r}_{\rm LS} & = 
    \vecof{ \ma{J}_1 \cdot \Delta \ma{{U}}_{\rm s} \cdot \ma{\Psi}}
    + \vecof{\ma{J}_1 \cdot \ma{{U}}_{\rm s} \cdot (\ma{J}_1 \cdot \ma{U}_{\rm s})^\pinv \cdot \ma{J}_2 \cdot \Delta \ma{U}_{\rm s}} \notag \\ &
    - \vecof{\ma{J}_1 \cdot \ma{{U}}_{\rm s} \cdot (\ma{J}_1 \cdot \ma{U}_{\rm s})^\pinv \cdot \ma{J}_1 \cdot \Delta \ma{U}_{\rm s} \cdot \ma{\Psi}}
    - \vecof{\ma{J}_2 \cdot \Delta \ma{{U}}_{\rm s}}
                + \bigO{\Delta^2} \notag \\
   & = \ma{W}_{\rm R,U}\cdot \vecof{\Delta\ma{{U}}_{\rm s}} + \bigO{\Delta^2} \quad \mbox{where} \label{eqn_app_proof_subsp_sls_rlsexp} \\
   \ma{W}_{\rm R,U} & = 
    \left(
 \ma{\Psi}^\trans \kron \ma{J}_1 \right)
     + \ma{I}_d \kron \left(\ma{J}_1\cdot \ma{U}_{\rm s} \left(\ma{J}_1\cdot \ma{U}_{\rm s}\right)^+\cdot\ma{J}_2\right)
      - \ma{\Psi}^\trans \kron \left(\ma{J}_1\cdot \ma{U}_{\rm s} \left(\ma{J}_1\cdot \ma{U}_{\rm s}\right)^+\cdot\ma{J}_1\right)  - {\left(  \ma{I}_d \kron \ma{J}_2
        \right)}. \notag
\end{align}
\else
\begin{align}
\ma{r}_{\rm LS} & = 
    \vecof{ \ma{J}_1 \cdot \Delta \ma{{U}}_{\rm s} \cdot \ma{\Psi}} \notag \\ &
    + \vecof{\ma{J}_1 \cdot \ma{{U}}_{\rm s} \cdot (\ma{J}_1 \cdot \ma{U}_{\rm s})^\pinv \cdot \ma{J}_2 \cdot \Delta \ma{U}_{\rm s}} \notag \\ &
    - \vecof{\ma{J}_1 \cdot \ma{{U}}_{\rm s} \cdot (\ma{J}_1 \cdot \ma{U}_{\rm s})^\pinv \cdot \ma{J}_1 \cdot \Delta \ma{U}_{\rm s} \cdot \ma{\Psi}} \notag \\ &
    - \vecof{\ma{J}_2 \cdot \Delta \ma{{U}}_{\rm s}}
                + \bigO{\Delta^2} \notag \\
   & = \ma{W}_{\rm R,U}\cdot \vecof{\Delta\ma{{U}}_{\rm s}} + \bigO{\Delta^2} \quad \mbox{where} \label{eqn_app_proof_subsp_sls_rlsexp} \\
   \ma{W}_{\rm R,U} & = 
    \left(
 \ma{\Psi}^\trans \kron \ma{J}_1 \right)
     + \ma{I}_d \kron \left(\ma{J}_1\cdot \ma{U}_{\rm s} \left(\ma{J}_1\cdot \ma{U}_{\rm s}\right)^+\cdot\ma{J}_2\right)
     \notag \\ &
      - \ma{\Psi}^\trans \kron \left(\ma{J}_1\cdot \ma{U}_{\rm s} \left(\ma{J}_1\cdot \ma{U}_{\rm s}\right)^+\cdot\ma{J}_1\right)  - {\left(  \ma{I}_d \kron \ma{J}_2
        \right)}. \notag
\end{align}
\fi
Using~\eqref{eqn_app_proof_subsp_sls_rlsexp} and~\eqref{eqn_app_proof_subsp_sls_fpinv} 
in~\eqref{eqn_app_proof_subsp_sls_solcf} we find that
\ifCLASSOPTIONdraftcls
\begin{align}
\begin{bmatrix}
        \Delta \ma{\psi}_{\rm SLS} \\
        \Delta\ma{{u}}_{\rm s,SLS}
    \end{bmatrix}
    & = 
    - \ma{\hat{F}}_{\rm SLS}^\pinv \cdot \ma{r}_{\rm LS} 
    = - \left(\ma{{F}}_{\rm SLS}^\pinv + \bigO{\Delta}\right) 
        \cdot
        \left(\ma{W}_{\rm R,U}\cdot \vecof{\Delta\ma{{U}}_{\rm s}} + \bigO{\Delta^2}\right) \notag \\
   & = - \ma{{F}}_{\rm SLS}^\pinv \cdot \ma{W}_{\rm R,U}\cdot \vecof{\Delta\ma{{U}}_{\rm s}} + \bigO{\Delta^2}.
   \label{eqn_app_proof_subsp_lssol_ir1}
\end{align}
\else
\begin{align}
   & \begin{bmatrix}
        \Delta \ma{\psi}_{\rm SLS} \\
        \Delta\ma{{u}}_{\rm s,SLS}
    \end{bmatrix}
     = 
    - \ma{\hat{F}}_{\rm SLS}^\pinv \cdot \ma{r}_{\rm LS} \notag \\ & \quad\quad
    = - \left(\ma{{F}}_{\rm SLS}^\pinv + \bigO{\Delta}\right) 
        \cdot
        \left(\ma{W}_{\rm R,U}\cdot \vecof{\Delta\ma{{U}}_{\rm s}} + \bigO{\Delta^2}\right) 
        \notag \\ & \quad\quad
    = - \ma{{F}}_{\rm SLS}^\pinv \cdot \ma{W}_{\rm R,U}\cdot \vecof{\Delta\ma{{U}}_{\rm s}} + \bigO{\Delta^2}.
   \label{eqn_app_proof_subsp_lssol_ir1}
\end{align}
\fi
Note that from~\eqref{eqn_app_proof_subsp_slsfsls} it follows that $\ma{F}_{\rm SLS}$ has full row-rank
and hence its pseudo-inverse can be expressed as $\ma{F}_{\rm SLS}^\pinv = \ma{F}_{\rm SLS}^\herm \cdot(
\ma{F}_{\rm SLS}\cdot\ma{F}_{\rm SLS}^\herm)^{-1}$. This allows us to extract 
$ \Delta \ma{\psi}_{\rm SLS}$ from~\eqref{eqn_app_proof_subsp_lssol_ir1} since $\Delta\ma{{u}}_{\rm s,SLS}$
is not explicitly needed as long as only one SLS iteration is performed. We obtain
\ifCLASSOPTIONdraftcls
\begin{align}
   \Delta \ma{\psi}_{\rm SLS}
   = - \left( \ma{I}_d \kron \left(\ma{J}_1 \cdot \ma{U}_{\rm s}\right)^\herm \right)
   \cdot (\ma{F}_{\rm SLS}\cdot\ma{F}_{\rm SLS}^\herm)^{-1} \cdot 
   \ma{W}_{\rm R,U}\cdot \vecof{\Delta\ma{{U}}_{\rm s}} + \bigO{\Delta^2} \label{eqn_app_proof_subsp_sls_psisls}
\end{align}
\else
\begin{align}
   \Delta \ma{\psi}_{\rm SLS}
    = & - \left( \ma{I}_d \kron \left(\ma{J}_1 \cdot \ma{U}_{\rm s}\right)^\herm \right)
   \cdot (\ma{F}_{\rm SLS}\cdot\ma{F}_{\rm SLS}^\herm)^{-1} \notag \\ & \cdot 
   \ma{W}_{\rm R,U}\cdot \vecof{\Delta\ma{{U}}_{\rm s}} + \bigO{\Delta^2} \label{eqn_app_proof_subsp_sls_psisls}
\end{align}
\fi
The final step in SLS-based ESPRIT is to replace the LS-based estimate $\ma{\hat{\Psi}}_{\rm LS}$
by $\ma{\hat{\Psi}}_{\rm SLS} =\ma{\hat{\Psi}}_{\rm LS} + \Delta \ma{\Psi}_{\rm SLS}
= \ma{\Psi} + \Delta \ma{{\Psi}}_{\rm LS} + \Delta \ma{\Psi}_{\rm SLS}$. Following the
first-order expansion for Standard ESPRIT from~\cite{LLV:93} we obtain
\ifCLASSOPTIONdraftcls
\begin{align}
    \Delta \mu_k & = \imagof{\ma{p}_k^\trans \cdot (\Delta \ma{{\Psi}}_{\rm LS} + \Delta \ma{\Psi}_{\rm SLS})
     \cdot \ma{q}_k} / \expof{\j \mu_k} + \bigO{\Delta^2} \notag \\
     & = 
     \imagof{\ma{p}_k^\trans \cdot \Delta \ma{{\Psi}}_{\rm LS}\cdot \ma{q}_k} / \expof{\j \mu_k}
     + \imagof{\left(\ma{q}_k^\trans \kron \ma{p}_k^\trans\right) \cdot \Delta \ma{\psi}_{\rm SLS}} 
     / \expof{\j \mu_k}
      + \bigO{\Delta^2}
\end{align}
\else
\begin{align}
    \Delta \mu_k & = \imagof{\ma{p}_k^\trans \cdot (\Delta \ma{{\Psi}}_{\rm LS} + \Delta \ma{\Psi}_{\rm SLS})
     \cdot \ma{q}_k} / \expof{\j \mu_k} + \bigO{\Delta^2} \notag \\
     & = 
     \imagof{\ma{p}_k^\trans \cdot \Delta \ma{{\Psi}}_{\rm LS}\cdot \ma{q}_k} / \expof{\j \mu_k}
     \notag \\ &
     + \imagof{\left(\ma{q}_k^\trans \kron \ma{p}_k^\trans\right) \cdot \Delta \ma{\psi}_{\rm SLS}} 
     / \expof{\j \mu_k}
      + \bigO{\Delta^2}
\end{align}
\fi
Since the first term is exactly the same as for LS-based 1-D Standard ESPRIT, we can
use the result from~\cite{LLV:93}.
Inserting the expansion 
for $\Delta \ma{\psi}_{\rm SLS}$ from~\eqref{eqn_app_proof_subsp_sls_psisls} we obtain
\ifCLASSOPTIONdraftcls
\begin{align}
  \Delta \mu_k & =
    \imagof{
       \ma{p}_k^\trans \cdot \left( \ma{J}_1 \cdot \ma{U}_{\rm s} \right)^+
       \cdot
       \left( \frac{\ma{J}_2}{\expof{\j \mu_k}} -\ma{J}_1 \right)
       \cdot
       \Delta \ma{U}_{\rm s} \cdot \ma{q}_k
    } 
   \notag \\
& - \imagof{\left(\ma{q}_k^\trans \kron 
     \ma{p}_k^\trans \cdot \left(\ma{J}_1 \cdot \ma{U}_{\rm s}\right)^\herm \right)
   \cdot (\ma{F}_{\rm SLS}\cdot\ma{F}_{\rm SLS}^\herm)^{-1} \cdot 
   \ma{W}_{\rm R,U}\cdot \vecof{\Delta\ma{{U}}_{\rm s}}}
     / \expof{\j \mu_k} + \bigO{\Delta^2}
\end{align}
\else
\begin{align}
  \Delta \mu_k & =
    \imagof{
       \ma{p}_k^\trans \cdot \left( \ma{J}_1 \cdot \ma{U}_{\rm s} \right)^+
       \cdot
       \left( \frac{\ma{J}_2}{\expof{\j \mu_k}} -\ma{J}_1 \right)
       \cdot
       \Delta \ma{U}_{\rm s} \cdot \ma{q}_k
    } 
   \notag \\
& - \imagofnlr{\left(\ma{q}_k^\trans \kron 
     \ma{p}_k^\trans \cdot \left(\ma{J}_1 \cdot \ma{U}_{\rm s}\right)^\herm \right)
   \cdot (\ma{F}_{\rm SLS}\cdot\ma{F}_{\rm SLS}^\herm)^{-1} \notag \\ & \quad \cdot 
   \ma{W}_{\rm R,U}\cdot \vecof{\Delta\ma{{U}}_{\rm s}}}{\Big}
     / \expof{\j \mu_k} + \bigO{\Delta^2}
\end{align}
\fi
Finally, rearranging the first term as 
$\ma{a}^\trans \cdot \ma{B} \cdot \ma{c} = (\ma{c}^\trans \kron \ma{a}^\trans) \cdot \vecof{\ma{B}}$
we have
\ifCLASSOPTIONdraftcls
\begin{align}
  \Delta \mu_k & =
    \imagof{
       \ma{q}_k^\trans \kron \left(
       \ma{p}_k^\trans \cdot \left( \ma{J}_1 \cdot \ma{U}_{\rm s} \right)^+
       \cdot
       \left( \frac{\ma{J}_2}{\expof{\j \mu_k}} -\ma{J}_1 \right)\right)
       \cdot
       \vecof{\Delta \ma{U}_{\rm s}} 
    } 
   \notag \\
& - \imagof{\left(\ma{q}_k^\trans \kron 
     \ma{p}_k^\trans \cdot  \frac{\left(\ma{J}_1 \cdot \ma{U}_{\rm s}\right)^\herm}{\expof{\j \mu_k}} \right)
   \cdot (\ma{F}_{\rm SLS}\cdot\ma{F}_{\rm SLS}^\herm)^{-1} \cdot 
   \ma{W}_{\rm R,U}\cdot \vecof{\Delta\ma{{U}}_{\rm s}}}
     + \bigO{\Delta^2} \notag \\
     & = 
     \imagof{\ma{r}_{k,{\rm SLS}}^\trans \cdot \vecof{\Delta\ma{{U}}_{\rm s}}} + \bigO{\Delta^2}, \quad \mbox{where} \notag \\
     \ma{r}_{k,{\rm SLS}}^\trans & =   \ma{q}_k^\trans \kron \left[\ma{p}_k^\trans \cdot \left(\ma{J}_1 \cdot \ma{U}_s\right)^+
                                  \cdot \left( \frac{\ma{J}_2}{\expof{\j \mu_k}} - \ma{J}_1 \right) \right] 
                       -  \left(\ma{q}_k^\trans \kron \left[ \ma{p}_k^\trans\cdot 
                       \frac{(\ma{J}_1 \cdot \ma{U}_{\rm s})^\herm}{\expof{\j \mu_k}} \right] \right)
                                     \cdot \left(\ma{F}_{\rm SLS}\cdot\ma{F}_{\rm SLS}^\herm\right)^{-1}
                                     \cdot \ma{W}_{\rm R,U},
\end{align}
\else
\begin{align}
  \Delta \mu_k & =
    \imagofnlr{
       \ma{q}_k^\trans \kron \left(
       \ma{p}_k^\trans \cdot \left( \ma{J}_1 \cdot \ma{U}_{\rm s} \right)^+
       \cdot
       \left( \frac{\ma{J}_2}{\expof{\j \mu_k}} -\ma{J}_1 \right)\right)
       \notag \\ & \quad\quad
       \cdot
       \vecof{\Delta \ma{U}_{\rm s}} 
    }{\Big}
   \notag \\
& - \imagofnlr{\Big(\ma{q}_k^\trans \kron 
     \ma{p}_k^\trans \cdot  \frac{\left(\ma{J}_1 \cdot \ma{U}_{\rm s}\right)^\herm}{\expof{\j \mu_k}} \Big)
   \cdot (\ma{F}_{\rm SLS}\cdot\ma{F}_{\rm SLS}^\herm)^{-1} 
   \notag \\ & \quad\quad
   \cdot 
   \ma{W}_{\rm R,U}\cdot \vecof{\Delta\ma{{U}}_{\rm s}}}{\Big}
     + \bigO{\Delta^2} \notag \\
     & = 
     \imagof{\ma{r}_{k,{\rm SLS}}^\trans \cdot \vecof{\Delta\ma{{U}}_{\rm s}}} + \bigO{\Delta^2}, \quad \mbox{where} \notag \\
     \ma{r}_{k,{\rm SLS}}^\trans & =   \ma{q}_k^\trans \kron \left[\ma{p}_k^\trans \cdot \left(\ma{J}_1 \cdot \ma{U}_s\right)^+
                                  \cdot \left( \frac{\ma{J}_2}{\expof{\j \mu_k}} - \ma{J}_1 \right) \right] 
                                  \notag \\ & \!\!\!\!\!\!\!\!\!\!\!\!
                       -  \left(\ma{q}_k^\trans \kron \left[ \ma{p}_k^\trans\cdot 
                       \frac{(\ma{J}_1 \cdot \ma{U}_{\rm s})^\herm}{\expof{\j \mu_k}} \right] \right)
                                     \cdot \left(\ma{F}_{\rm SLS}\cdot\ma{F}_{\rm SLS}^\herm\right)^{-1}
                                     \cdot \ma{W}_{\rm R,U},
\end{align}
\fi
which is the desired result~\eqref{eqn_subsp_perf_sls_u}. 
Equation~\eqref{eqn_subsp_perf_sls} follows from~\eqref{eqn_subsp_perf_sls_u} by inserting
the first order expansion for $\vecof{\Delta\ma{U}_{\rm s}}$ in terms of $\vecof{\ma{N}}$ as shown in Appendix~\ref{sec_app_proof_perf_mse}. \qed

\section{Proof of Theorem~\ref{thm_perf_mse_singsrc_1d}}\label{sec_app_proof_perf__mse_singsrc_1d}

This theorem consists of several parts which we address in separate subsections. 

\subsection{MSE for Standard ESPRIT}

We start by simplifying the MSE expression for 1-D Standard ESPRIT.
In the case of a single source we can write
\begin{align}
    \ma{X}_0 = \ma{a}(\mu) \cdot \ma{s}^\trans, \label{eqn_app_proof_perf_singsrc_1d_x0}
\end{align}
where $\ma{a} \in \compl^{M \times 1}$ is the array steering vector and
$\ma{s} \in \compl^{N \times 1}$ contains the source symbols. 
Let $\hat{P}_{\rm T} = \twonorm{\ma{s}}^2 / N$ be the empirical source power.
Furthermore, since we assume a ULA of isotropic elements, $\ma{a}(\mu)$
is given by $\ma{a}(\mu) = [1, \expof{\j \mu}, \expof{2 \j \mu}, \ldots, \expof{(M-1) \j \mu}]$.
Note that $\twonorm{\ma{a}(\mu)}^2 = M$.
For notational convenience, we drop the explicit dependence of $\ma{a}$ on $\mu$
and write just $\ma{a}(\mu) = \ma{a}$ in the sequel.
The selection matrices $\ma{J}_1$ and $\ma{J}_2$ are then chosen as
\begin{align}
    \ma{J}_1 = \begin{bmatrix} \ma{I}_{M-1} & \ma{0}_{M-1 \times 1} \end{bmatrix}
    \quad
    \ma{J}_2 = \begin{bmatrix} \ma{0}_{M-1 \times 1} & \ma{I}_{M-1} \end{bmatrix}
\end{align}
for maximum overlap, i.e., $\subsel{M} = M-1$.
Since~\eqref{eqn_app_proof_perf_singsrc_1d_x0} is a rank-one matrix, we can directly
relate the subspaces to the array steering vector and the source symbol matrix, namely
\begin{align}
    \ma{U}_{\rm s} & = \ma{u}_{\rm s} = \frac{\ma{a}}{\twonorm{\ma{a}}} 
    = \frac{1}{\sqrt{M}} \cdot \ma{a} 
    \\
    \ma{V}_{\rm s} & = \ma{v}_{\rm s} = \frac{\ma{s}^\conj}{\twonorm{\ma{s}}} 
    = \frac{1}{\sqrt{\hat{P}_{\rm T}\cdot N}} \cdot \ma{s}^\conj 
    \\
    \ma{\Sigma}_{\rm s} & = \sigma_{\rm s} = \sqrt{M \cdot N \cdot \hat{P}_{\rm T}}.
\end{align}
%
For the MSE expression from Theorem~\ref{thm_perf_mse} we also require the quantity 
$\ma{U}_{\rm n} \cdot \ma{U}_{\rm n}^\herm$, which resembles
a projection matrix on th noise subspace. However, since the signal subspace is spanned by 
$\ma{a}$ we can write 
$\ma{U}_{\rm n} \cdot \ma{U}_{\rm n}^\herm = \projp{\ma{a}} = \ma{I}_M - \frac{\ma{a} \cdot \ma{a}^\herm}{\twonorm{\ma{a}}^2}
= \ma{I}_M - \frac{1}{M} \cdot \ma{a} \cdot \ma{a}^\herm$.
The MSE expression for 1-D Standard ESPRIT
also include the eigenvectors of $\ma{\Psi}$ which for the special case
discussed here is scalar and given by $\ma{\Psi} = \expof{\j \mu}$. Consequently, we have $\ma{p}_k
= \ma{q}_k = 1$ for the eigenvectors.

Combining these expressions and inserting into
\revA{the special case of~\eqref{eqn_subsp_perf_mse_se} for white noise, which is shown
in equation \eqref{eqn_app_proof_pert_mse_whitenoise}}, we have
$\expvof{ (\Delta \mu_k^{(r)})^2} = \sigma_{\rm n}^2/2 \cdot \twonorm{ \ma{W}_{\rm mat}^\trans \cdot \ma{r} }^2
= \sigma_{\rm n}^2/2 \cdot \twonorm{ \ma{r}^\trans \cdot \ma{W}_{\rm mat} }^2$ with
\ifCLASSOPTIONdraftcls
\begin{align}
  \ma{r}_k^{(r)} & = \ma{r} =
         \left[ \left(\ma{{J}}_1  \frac{\ma{a}}{\sqrt{M}}  \right)^+ 
                \left(\ma{{J}}_2/\expof{\j \cdot \mu} - \ma{{J}}_1\right)
         \right]^\trans  \label{eqn_app_proof_subsp_d1_rk_ls} \\
        \ma{W}_{\rm mat} & =    
   \left(\frac{1}{\sqrt{M \cdot N \cdot \hat{P}_{\rm T}}} \cdot \frac{\ma{s}^\herm}{\sqrt{\hat{P}_{\rm T}\cdot N}}
    \right) \kron \projp{\ma{a}}    \label{eqn_app_proof_subsp_d1_wmat} 
\end{align}
\else
\begin{align}
  \ma{r}_k^{(r)} & = \ma{r} =
         \left[ \left(\ma{{J}}_1  \frac{\ma{a}}{\sqrt{M}}  \right)^+ 
                \left(\ma{{J}}_2/\expof{\j \cdot \mu} - \ma{{J}}_1\right)
         \right]^\trans  \label{eqn_app_proof_subsp_d1_rk_ls} \\
        \ma{W}_{\rm mat} & =    
   \left(\frac{1}{\sqrt{M \cdot N \cdot \hat{P}_{\rm T}}} \cdot \frac{\ma{s}^\herm}{\sqrt{\hat{P}_{\rm T}\cdot N}}
     \right) 
    \kron \projp{\ma{a}}    \label{eqn_app_proof_subsp_d1_wmat} 
\end{align}
\fi
Note that $\ma{W}_{\rm mat}$ is the Kronecker product of a $1 \times N$ vector and an $M \times M$
matrix. Hence, $\ma{r}^\trans \cdot \ma{W}_{\rm mat}$ can be written as
$\ma{r}^\trans \cdot \ma{W}_{\rm mat} = \tilde{\ma{s}}^\trans \kron \tilde{\ma{a}}^\trans$,
where
\begin{align}
\tilde{\ma{s}}^\trans & = 
\frac{1}{\sqrt{M \cdot N \cdot \hat{P}_{\rm T}}} \cdot \frac{\ma{s}^\herm}{\sqrt{\hat{P}_{\rm T}\cdot N}}  \\
\tilde{\ma{a}}^\trans & = 
 \left(\ma{{J}}_1  \frac{\ma{a}}{\sqrt{M}} \right)^+ 
                \left(\frac{\ma{{J}}_2}{\expof{\j \cdot \mu}} - \ma{{J}}_1\right)
          \cdot 
         \projp{\ma{a}}
 \label{eqn_app_proof_subsp_d1_defatst}
\end{align}
Therefore, the MSE can be expressed as
$\expvof{ (\Delta \mu_k^{(r)})^2} = \sigma_{\rm n}^2/2 \cdot \twonorm{\tilde{\ma{s}}^\trans \kron \tilde{\ma{a}}^\trans}^2$,
which is equal to 
$\expvof{ (\Delta \mu_k^{(r)})^2} = \sigma_{\rm n}^2/2 \cdot \twonorm{\tilde{\ma{s}}^\trans}^2 \cdot \twonorm{\tilde{\ma{a}}^\trans}^2$.

Since $\ma{\tilde{s}}^\trans$ is a scaled version 
of $\ma{s}^\herm$ and $\twonorm{\ma{s}^\herm}^2 = N \cdot \hat{P}_{\rm T}$
we find that the first term in the MSE expression can conveniently be expressed as
\begin{align}
 \twonorm{\tilde{\ma{s}}^\trans}^2 = \frac{1}{M \cdot N \cdot \hat{P}_{\rm T}} \cdot \frac{\hat{P}_{\rm T}\cdot N}{\hat{P}_{\rm T}\cdot N} = \frac{1}{M \cdot N \cdot \hat{P}_{\rm T}}. \label{eqn_app_proof_mse_d1se_norms}
\end{align}
Next, we proceed to simplify $\ma{\tilde{a}}^\trans$ further. To this end, we expand the pseudo-inverse
of $\ma{J}_1 \cdot \ma{a}$ using the rule $\ma{x}^\pinv = \ma{x}^\herm / \twonorm{\ma{x}}^2$
and multiply the brackets out. After straightforward algebraic manipulations we obtain
\begin{align}
\ma{\tilde{a}}^\trans
& = \frac{\sqrt{M}}{M-1} \Big(\ma{\tilde{a}}_1^\trans - \ma{\tilde{a}}_2^\trans\Big), \quad \mbox{where} \notag \\
\ma{\tilde{a}}_1^\trans & = \ma{a}^\herm \cdot \ma{{J}}_1^\herm \cdot \ma{{J}}_2/\expof{\j \cdot \mu} 
\quad \mbox{and} \quad
\ma{\tilde{a}}_1^\trans  = \ma{a}^\herm \cdot \ma{{J}}_1^\herm \cdot \ma{{J}}_1.
 \label{eqn_app_proof_subsp_d1se_atilde}
\end{align}
%
Since we have $\ma{a}^\herm = [1, \expof{-\j \mu}, \expof{-2 \j \mu}, \ldots, \expof{-(M-1) \j \mu}]$
it is easy to see that
\begin{align}
   \ma{\tilde{a}}_1^\trans 
   & = [0, \expof{-\j \mu}, \expof{-2 \j \mu}, \ldots, \expof{-(M-2) \j \mu}, \expof{-(M-1) \j \mu}] 
   \notag \\
   \ma{\tilde{a}}_2^\trans & = [1, \expof{-\j \mu}, \expof{-2 \j \mu}, \ldots, \expof{-(M-2) \j \mu}, 0]  \quad \mbox{and hence} \notag \\
   \ma{\tilde{a}}_1^\trans - \ma{\tilde{a}}_2^\trans & = [-1, 0, \ldots, 0, \expof{-(M-1) \j \mu}].
   \label{eqn_app_proof_subsp_d1se_a12tilde}
\end{align}
Consequently, we find $\twonorm{\ma{\tilde{a}}^\trans}^2 = \frac{M}{(M-1)^2} \cdot 2$.
Combining this result with~\eqref{eqn_app_proof_mse_d1se_norms} we finally have
\ifCLASSOPTIONdraftcls
\begin{align}
\expvof{ (\Delta \mu_k^{(r)})^2} & = 
\frac{\sigma_{\rm n}^2}{2} \cdot \twonorm{\tilde{\ma{s}}^\trans}^2 \cdot \twonorm{\tilde{\ma{a}}^\trans}^2 \label{eqn_app_proof_perf_singsrc_1d_afr}
= \frac{\sigma_{\rm n}^2}{2} \cdot \frac{1}{M \cdot N \cdot \hat{P}_{\rm T}} \cdot 
2 \cdot \frac{M}{(M-1)^2}  \\
& = \frac{\sigma_{\rm n}^2}{N \cdot \hat{P}_{\rm T}} \cdot \frac{1}{(M-1)^2} 
\label{eqn_app_proof_subsp_d1se_res}
\end{align}
\else
\begin{align}
\expvof{ (\Delta \mu_k^{(r)})^2} & = 
\frac{\sigma_{\rm n}^2}{2} \cdot \twonorm{\tilde{\ma{s}}^\trans}^2 \cdot \twonorm{\tilde{\ma{a}}^\trans}^2 \label{eqn_app_proof_perf_singsrc_1d_afr} \\
& = \frac{\sigma_{\rm n}^2}{2} \cdot \frac{1}{M \cdot N \cdot \hat{P}_{\rm T}} \cdot 
2 \cdot \frac{M}{(M-1)^2}  \notag\\
& = \frac{\sigma_{\rm n}^2}{N \cdot \hat{P}_{\rm T}} \cdot \frac{1}{(M-1)^2} 
\label{eqn_app_proof_subsp_d1se_res}
\end{align}
\fi
which is the desired result. \qed

\subsection{MSE for Unitary ESPRIT}\label{sec_app_proof_perf__mse_singsrc_1due}

The second part of the theorem is to show that for a single source, the MSE
for Unitary ESPRIT is the same as the MSE for Standard ESPRIT. Firstly, we expand
$\fba{\ma{X}}_0$ and find
\ifCLASSOPTIONdraftcls
\begin{align}
    \fba{\ma{X}}_0 = 
       \begin{bmatrix}
           \ma{a} \cdot \ma{s}^\trans, &
            \ma{\Pi}_M \cdot \ma{a}^\conj \cdot \ma{s}^\herm \cdot \ma{\Pi}_N
       \end{bmatrix}
       = \ma{a} \cdot \begin{bmatrix} \ma{s}^\trans & \expof{-\j \mu (M-1)} \cdot \ma{s}^\herm \cdot \ma{\Pi}_N \end{bmatrix}
       = \ma{a} \cdot \overline{\ma{s}}^\trans \label{eqn_app_proof_subsp_dmfba}
\end{align}
\else
\begin{align}
    \fba{\ma{X}}_0 & = 
       \begin{bmatrix}
           \ma{a} \cdot \ma{s}^\trans, &
            \ma{\Pi}_M \cdot \ma{a}^\conj \cdot \ma{s}^\herm \cdot \ma{\Pi}_N
       \end{bmatrix} \notag \\
       & = \ma{a} \cdot \begin{bmatrix} \ma{s}^\trans & \expof{-\j \mu (M-1)} \cdot \ma{s}^\herm \cdot \ma{\Pi}_N \end{bmatrix}
       = \ma{a} \cdot \overline{\ma{s}}^\trans \label{eqn_app_proof_subsp_dmfba}
\end{align}
\fi
where we have used the fact that for our ULA 
we have $\ma{\Pi}_M \cdot \ma{a}^\conj = \ma{a} \cdot \expof{-\j \mu (M-1)}$
and we have defined $\overline{\ma{s}}$ to be
\begin{align}
    \overline{\ma{s}} = \begin{bmatrix} \ma{s} \\ \expof{-\j \mu (M-1)}\cdot \ma{\Pi}_N\cdot \ma{s}^\conj \end{bmatrix}
\end{align}
Note that $\overline{\ma{s}}^\herm \overline{\ma{s}} = \ma{s}^\herm \cdot \ma{s} + \ma{s}^\trans \cdot \ma{\Pi}_N \cdot \ma{\Pi}_N \cdot \ma{s}^\conj = 2 \cdot \ma{s}^\herm \cdot \ma{s}$.
As for Standard ESPRIT, we relate~\eqref{eqn_app_proof_subsp_dmfba} to its
SVD and obtain
\ifCLASSOPTIONdraftcls
\begin{align}
   \fba{\ma{u}_{\rm s}}  = \frac{\ma{a}}{\sqrt{M}} = \ma{u}_{\rm s} , \quad 
   \fba{\ma{v}_{\rm s}}  = \frac{\overline{\ma{s}}^\conj}{\sqrt{2 \cdot N \cdot \hat{P}_{\rm T}}}, \quad
   \fba{\sigma_{\rm s}}  = \sqrt{2 \cdot M \cdot N \cdot \hat{P}_{\rm T}}. \label{eqn_app_proof_subsp_usfba}
\end{align}
\else
\begin{align}
   \fba{\ma{u}_{\rm s}}  &= \frac{\ma{a}}{\sqrt{M}} = \ma{u}_{\rm s} , \quad 
   \fba{\ma{v}_{\rm s}}  = \frac{\overline{\ma{s}}^\conj}{\sqrt{2 \cdot N \cdot \hat{P}_{\rm T}}}, \notag \\ 
   \fba{\sigma_{\rm s}}  &= \sqrt{2 \cdot M \cdot N \cdot \hat{P}_{\rm T}}. \label{eqn_app_proof_subsp_usfba}
\end{align}
\fi
An important consequence we can draw from~\eqref{eqn_app_proof_subsp_usfba} is that the column space
$\ma{u}_{\rm s}$ remains unaffected from the 
forward-backward-averaging. Therefore we also have $\fba{\ma{U}_{\rm n}} = \ma{U}_{\rm n}$ and hence
$\fba{\ma{U}_{\rm n}}\fbaH{\ma{U}_{\rm n}} = \ma{I}_M - \frac{\ma{a}\cdot \ma{a}^\herm}{M}$.
However, the equivalence of Standard ESPRIT and Unitary ESPRIT is still not obvious since
Forward-Backward Averaging destroys the circular symmetry of the noise, which leads to an additional
term in the MSE expressions.
Following the lines of the derivation for 1-D Standard ESPRIT we can show that
\begin{align}
    \fbaT{\ma{r}} \cdot \fba{\ma{W}}_{\rm mat} = \ma{\tilde{\overline{s}}}^\trans \kron \ma{\tilde{a}}^\trans,
    \label{eqn_app_proof_subsp_d1ue_ir1}
\end{align}
where $\ma{\tilde{\overline{s}}}$ is given by 
\begin{align}
   \ma{\tilde{\overline{s}}} = \frac{1}{\sqrt{2 \cdot M \cdot N \cdot \hat{P}_{\rm T}}} \cdot
   \frac{\overline{\ma{s}}^\conj}{\sqrt{2 \cdot N \cdot \hat{P}_{\rm T}}} 
   \label{eqn_app_proof_subsp_d1ue_stildebar}
\end{align}
and $\ma{\tilde{a}}$ is the same as in the derivation for 1-D Standard ESPRIT (cf. equation~\eqref{eqn_app_proof_subsp_d1_defatst}).

According to Theorem~\ref{thm_perf_mse_fba}, the MSE for Unitary ESPRIT can be computed
as 
\begin{align}
    \frac{\sigma_{\rm n}^2}{2} \cdot 
    \left( \ma{z}^\trans \cdot \ma{z}^\conj
    - \realof{\ma{z}^\trans \cdot \ma{\Pi}_{2MN} \cdot \ma{z}}\right),
\end{align}
for $\ma{z}^\trans = \fbaT{\ma{r}} \cdot \fba{\ma{W}}_{\rm mat}$,
\revA{where we have inserted $\fba{\ma{R}_{\rm nn}} = \sigma_{\rm n}^2 \cdot \ma{I}_{2MN}$
and $\fba{\ma{C}_{\rm nn}} = \sigma_{\rm n}^2 \cdot \ma{\Pi}_{2MN}$
since we are considering the special case of circularly symmetric white noise.}
Using~\eqref{eqn_app_proof_subsp_d1ue_ir1} and the fact
that $\ma{\Pi}_{2MN} = \ma{\Pi}_{2N} \kron \ma{\Pi}_M$, this expression can be written into
\begin{align}
\frac{\sigma_{\rm n}^2}{2} \cdot 
    \left( \twonorm{\ma{\tilde{\overline{s}}}}^2
    \twonorm{\ma{\tilde{a}}}^2
    -  \ma{\tilde{\overline{s}}}^\trans \cdot \ma{\Pi}_{2N} \cdot \ma{\tilde{\overline{s}}}^\conj
    \cdot \ma{\tilde{a}}^\trans \cdot \ma{\Pi}_{M} \cdot \ma{\tilde{a}}\right).
    \label{eqn_app_proof_subsp_d1ue_ir2} 
\end{align}
Since $\ma{\tilde{a}}$ is the same as in~\eqref{eqn_app_proof_subsp_d1se_atilde}
we know that $\twonorm{\ma{\tilde{a}}}^2 = \frac{2 M}{(M-1)^2}$. Moreover, 
$\twonorm{\ma{\tilde{\overline{s}}}}^2 = \frac{1}{2 \cdot M \cdot N \cdot \hat{P}_{\rm T}}$ follows
directly from~\eqref{eqn_app_proof_subsp_d1ue_stildebar}.
For the second term in~\eqref{eqn_app_proof_subsp_d1ue_ir2} we have
\begin{align}
   & \ma{\tilde{\overline{s}}}^\trans \cdot \ma{\Pi}_{2N} \cdot \ma{\tilde{\overline{s}}}^\conj 
    = 
   \frac{1}{2  M  N  \hat{P}_{\rm T}} \cdot \frac{1}{2  N  \hat{P}_{\rm T}}
   \cdot
    \ma{\overline{s}}^\herm \cdot \ma{\Pi}_{2N} \cdot \ma{\overline{s}} \notag \\
   & =    
   \frac{1}{2  M  N  \hat{P}_{\rm T}} \cdot \frac{1}{2  N  \hat{P}_{\rm T}}
   \cdot
   \left[ \ma{s}^\herm \cdot \ma{s} \cdot \expof{\j \mu (M-1)} + 
   \ma{s}^\trans \cdot \ma{s}^\conj \cdot \expof{\j \mu (M-1)}
   \right] \notag \\   
   & =
   \frac{1}{2  M  N  \hat{P}_{\rm T}}
   \cdot \expof{\j \mu (M-1)} 
   \label{eqn_app_proof_subsp_d1ue_ir3}
\end{align}
Similarly we can simplify $\ma{\tilde{a}}^\trans \cdot \ma{\Pi}_{M} \cdot \ma{\tilde{a}}$ 
by using~\eqref{eqn_app_proof_subsp_d1se_atilde} and~\eqref{eqn_app_proof_subsp_d1se_a12tilde}.
We obtain
\begin{align}
   \ma{\tilde{a}}^\trans \cdot \ma{\Pi}_{M} \cdot \ma{\tilde{a}}
   & = - \frac{2 M}{(M-1)^2} \cdot \expof{-(M-1) \j \mu}.
   \label{eqn_app_proof_subsp_d1ue_ir4}
\end{align}
Combining the results from~\eqref{eqn_app_proof_subsp_d1ue_ir3} and~\eqref{eqn_app_proof_subsp_d1ue_ir4} 
into~\eqref{eqn_app_proof_subsp_d1ue_ir2} we finally obtain for the MSE
\ifCLASSOPTIONdraftcls
\begin{align}
&
\frac{\sigma_{\rm n}^2}{2} \cdot 
    \left(  \frac{1}{2  M  N  \hat{P}_{\rm T}} \cdot
    \frac{2 M}{(M-1)^2}
    +  \frac{1}{2  M  N  \hat{P}_{\rm T}}
   \cdot \expof{\j \mu (M-1)}  \cdot 
    \frac{2 M}{(M-1)^2} \cdot \expof{-(M-1) \j \mu}
    \right) \notag \\
 = &
\frac{\sigma_{\rm n}^2}{2} \cdot 
    \left(  \frac{1}{ N  \hat{P}_{\rm T}} \cdot
    \frac{1}{(M-1)^2}
    +  \frac{1}{ N  \hat{P}_{\rm T}}
   \cdot \frac{1}{(M-1)^2} \cdot 
    \right)    \\
= &
\frac{\sigma_{\rm n}^2}{N \cdot \hat{P}_{\rm T}} \cdot  \frac{1}{(M-1)^2},
\end{align}
\else
\begin{align}
&
\frac{\sigma_{\rm n}^2}{2} \cdot 
    \Big(  \frac{1}{2  M  N  \hat{P}_{\rm T}} \cdot
    \frac{2 M}{(M-1)^2} \notag \\ & 
    +  \frac{1}{2  M  N  \hat{P}_{\rm T}}
   \cdot \expof{\j \mu (M-1)}  \cdot 
    \frac{2 M}{(M-1)^2} \cdot \expof{-(M-1) \j \mu}
    \Big) \notag \\
 = &
\frac{\sigma_{\rm n}^2}{2} \cdot 
    \left(  \frac{1}{ N  \hat{P}_{\rm T}} \cdot
    \frac{1}{(M-1)^2}
    +  \frac{1}{ N  \hat{P}_{\rm T}}
   \cdot \frac{1}{(M-1)^2} \cdot 
    \right)    \\
= &
\frac{\sigma_{\rm n}^2}{N \cdot \hat{P}_{\rm T}} \cdot  \frac{1}{(M-1)^2},
\end{align}
\fi
which is equal to the result for 1-D Standard ESPRIT from~\eqref{eqn_app_proof_subsp_d1se_res}
and hence proves this part
of the theorem. \qed

\subsection{Cramér-Rao Bound} \label{app_proof_d11d_crb}

The third part of the theorem is to simplify the deterministic Cramér-Rao
Bound (CRB) for the special case of a single source. To this end, a closed-form
expression for the deterministic CRB for this setting is given by~\cite{SN:89}
\begin{align}
    \ma{C} = \frac{\sigma_{\rm n}^2}{2\cdot N} \cdot
    \realof{
      \left[ \ma{D}^\herm \cdot \projp{\ma{A}} \cdot \ma{D}
      \right] \odot \ma{\hat{R}}_{\rm S}^\trans
    }^{-1}
\end{align}
where $\ma{\hat{R}}_{\rm S} = \frac{1}{N} \cdot \ma{S} \cdot \ma{S}^\herm$ is the sample covariance
matrix of the source symbols, $\ma{D} \in \compl^{M \times d}$ is the matrix of partial
derivatives of the array steering vectors with respect to the parameters of interest,
and $\projp{\ma{A}} = \ma{I}_M - \ma{A} \cdot \left(\ma{A}^\herm \cdot \ma{A} \right)^{-1} \cdot \ma{A}^\herm$.
In the case $d=1$, we have $\ma{\hat{R}}_{\rm S} = \twonorm{\ma{s}}^2 / N = \hat{P}_{\rm T}$
and the CRB expression simplifies into
\begin{align}
    C & = \frac{\sigma_{\rm n}^2}{2\cdot N\cdot \hat{P}_{\rm T}} \cdot
    \realof{
      \ma{d}^\herm \cdot \left(
      \ma{I}_M - \frac{\ma{a} \cdot \ma{a}^\herm}{M}
      \right) \cdot \ma{d}      
    }^{-1} \\
    & = 
    \frac{1}{2 \cdot \hat{\rho}} \cdot
    \realof{
       \ma{d}^\herm \cdot \ma{d} - 
      \frac{1}{M}\cdot
      \ma{d}^\herm \cdot \ma{a} \cdot \ma{a}^\herm \cdot \ma{d} 
    }^{-1}  \\
    & = 
    \frac{1}{2 \cdot \hat{\rho}} \cdot
      \left[
      \ma{d}^\herm \cdot \ma{d} - 
      \frac{1}{M}\cdot
      |\ma{d}^\herm \cdot \ma{a}|^2
      \right] ^{-1}  \label{eqn_app_proof_subsp_d1crb_ir1}
\end{align}
Since for a ULA the array steering vector can be expressed as
$ \ma{a} = \begin{bmatrix} 1 & \expof{\j \mu} & \expof{2 \j \mu} & \ldots & \expof{(M-1) \j \mu} \end{bmatrix}$
we have 
\begin{align}
   \ma{d} = \frac{\partial \ma{a}}{\partial \mu} = 
   \j \cdot \begin{bmatrix} 0 & \expof{\j \mu} & 2\cdot \expof{2 \j \mu} & \ldots & 
       (M-1) \cdot \expof{(M-1) \j \mu} \end{bmatrix}.
\end{align}
Consequently the terms $\ma{d}^\herm \cdot \ma{d}$ and $\ma{d}^\herm \cdot \ma{a}$ become
\begin{align}
   \ma{d}^\herm \cdot \ma{d} = \sum_{m=0}^{M-1} m^2 = \frac{1}{6} \cdot (M-1) \cdot M \cdot (2M-1) \\
   \ma{d}^\herm \cdot \ma{a} = -\j \sum_{m=0}^{M-1} m = -\j \cdot \frac{1}{2} \cdot (M-1) \cdot M 
\end{align}
Using these expressions in~\eqref{eqn_app_proof_subsp_d1crb_ir1}, we obtain
\ifCLASSOPTIONdraftcls
\begin{align}
   C = & \frac{1}{2 \cdot \hat{\rho}} \cdot
      \left[
      \frac{1}{6} \cdot (M-1) \cdot M \cdot (2M-1) - 
      \frac{1}{M}\cdot \left| -\j \cdot \frac{1}{2} \cdot (M-1) \cdot M\right|^2
      \right]^{-1}  \notag \\ 
   = & \frac{1}{2 \cdot \hat{\rho}} \cdot
      \left[
      \frac{1}{12} \cdot (M-1) \cdot M \cdot (M+1)
      \right]^{-1} \\ 
   = & \frac{1}{\hat{\rho}} \cdot \frac{6}{(M-1) \cdot M \cdot (M+1)},
\end{align}
\else
\begin{align}
   C = & \frac{1}{2 \cdot \hat{\rho}} \cdot
      \Big[
      \frac{1}{6} \cdot (M-1) \cdot M \cdot (2M-1)  \notag \\ & \quad - 
      \frac{1}{M}\cdot \Big| -\j \cdot \frac{1}{2} \cdot (M-1) \cdot M\Big|^2
      \Big]^{-1} \notag \\ 
   = & \frac{1}{2 \cdot \hat{\rho}} \cdot
      \left[
      \frac{1}{12} \cdot (M-1) \cdot M \cdot (M+1)
      \right]^{-1} \\ 
   = & \frac{1}{\hat{\rho}} \cdot \frac{6}{(M-1) \cdot M \cdot (M+1)},
\end{align}
\fi
which is the desired result. \qed

\section{Proof of Theorem~\ref{thm_perf_mse_singsrc_1d_sls}}\label{sec_app_proof_perf__mse_singsrc_1d_sls}

As shown in Theorem \ref{thm_perf_mse_sls}, the MSE for SLS in the special
case of circularly symmetric white noise can be expressed as
\begin{align}
   \expvof{\left(\Delta \mu_{k,{\rm SLS}}\right)^2 }
   = \frac{\sigma_{\rm n}^2}{2} \cdot \twonorm{\ma{W}_{\rm mat} \cdot \ma{r}_{k,{\rm SLS}}}^2,
\end{align}
where $\ma{r}_{k,{\rm SLS}} = \ma{r}_{k,{\rm LS}} - \Delta\ma{r}_{k,{\rm SLS}}$,
\ifCLASSOPTIONdraftcls
    \begin{align}
     \Delta\ma{r}_{k,{\rm SLS}}^\trans & =  \left(\ma{q}_k^\trans \kron \left[ \ma{p}_k^\trans\cdot 
                       \frac{(\ma{J}_1 \cdot \ma{U}_{\rm s})^\herm}{\expof{\j \mu_k}} \right] \right)
                                     \cdot \left(\ma{F}_{\rm SLS}\cdot\ma{F}_{\rm SLS}^\herm\right)^{-1}
                                     \cdot \ma{W}_{\rm R,U}, \notag
    \end{align}
\else
        \begin{align}
     \Delta\ma{r}_{k,{\rm SLS}}^\trans & =  \left(\ma{q}_k^\trans \kron \left[ \ma{p}_k^\trans\cdot 
                       \frac{(\ma{J}_1 \cdot \ma{U}_{\rm s})^\herm}{\expof{\j \mu_k}} \right] \right)
                       \notag \\ & \quad
                                     \cdot \left(\ma{F}_{\rm SLS}\cdot\ma{F}_{\rm SLS}^\herm\right)^{-1}                                                    \cdot \ma{W}_{\rm R,U}, \notag
        \end{align}                                     
\fi
and $\ma{W}_{\rm R,U}$ as well as $\ma{F}_{\rm SLS}$ are defined in
Theorem \ref{thm_perf_mse_sls}.
For a single source, we have $\ma{p}_k = \ma{q}_k = 1$, $\ma{\Psi} = \Psi = \expof{\j \mu}$,
and $\ma{U}_{\rm s} = \ma{a} / \sqrt{M}$, and therefore
$\Delta\ma{r}_{k,{\rm SLS}}^\trans = \Delta\ma{r}_{{\rm SLS}}^\trans$ simplifies to
\ifCLASSOPTIONdraftcls
\begin{align}
     \Delta\ma{r}_{{\rm SLS}}^\trans & =  
      \frac{(\ma{J}_1 \cdot \ma{a})^\herm}{\sqrt{M} \cdot \expof{\j \mu}} 
                                     \cdot \left(\ma{F}_{\rm SLS}\cdot\ma{F}_{\rm SLS}^\herm\right)^{-1}
                                     \cdot \ma{W}_{\rm R,U} \label{eqn_app_proof_perf_singsrc_sls_deltarls} \\
\ma{W}_{\rm R,U} & = 
    \left(
 \expof{\j \mu} \cdot \ma{J}_1 \right)
     + \left(\ma{J}_1\cdot \ma{a} \left(\ma{J}_1\cdot \ma{a} \right)^+\cdot\ma{J}_2\right)
      - \expof{\j \mu} \cdot \left(\ma{J}_1\cdot \ma{a} \left(\ma{J}_1\cdot \ma{a}\right)^+\cdot\ma{J}_1\right)  - \ma{J}_2 \notag \\
    \ma{F}_{\rm SLS} & = {\left[  
          \ma{J}_1 \cdot \frac{\ma{a}}{\sqrt{M}}, \;
         \expof{\j \mu} \cdot \ma{J}_1 -  \ma{J}_2
         \right]} \notag
\end{align}
\else
\begin{align}
     \Delta\ma{r}_{{\rm SLS}}^\trans & =  
      \frac{(\ma{J}_1 \cdot \ma{a})^\herm}{\sqrt{M} \cdot \expof{\j \mu}} 
                                     \cdot \left(\ma{F}_{\rm SLS}\cdot\ma{F}_{\rm SLS}^\herm\right)^{-1}
                                     \cdot \ma{W}_{\rm R,U} \label{eqn_app_proof_perf_singsrc_sls_deltarls} \\
\ma{W}_{\rm R,U} & = 
    \left(
 \expof{\j \mu} \cdot \ma{J}_1 \right)
     + \left(\ma{J}_1\cdot \ma{a} \left(\ma{J}_1\cdot \ma{a} \right)^+\cdot\ma{J}_2\right)
     \notag \\ &
      - \expof{\j \mu} \cdot \left(\ma{J}_1\cdot \ma{a} \left(\ma{J}_1\cdot \ma{a}\right)^+\cdot\ma{J}_1\right)  - \ma{J}_2 \notag \\
    \ma{F}_{\rm SLS} & = {\left[  
          \ma{J}_1 \cdot \frac{\ma{a}}{\sqrt{M}}, \;
         \expof{\j \mu} \cdot \ma{J}_1 -  \ma{J}_2
         \right]} \notag
\end{align}
\fi
We can write $\ma{W}_{\rm R,U} $ as 
\begin{align}
	\ma{W}_{\rm R,U}  = 
	   \left( \ma{J}_1\cdot \ma{a} \left(\ma{J}_1\cdot \ma{a} \right)^+  - \ma{I}_{M-1}\right)
	   \cdot
	   \left( \ma{J}_2 - \expof{\j \mu} \cdot \ma{J}_1 \right)  \label{eqn_app_proof_perf_singsrc_sls_wru}
\end{align}
Moreover, we need to simplify the term $\left(\ma{F}_{\rm SLS}\cdot\ma{F}_{\rm SLS}^\herm\right)^{-1}$.
It is easily verified that $\ma{F}_{\rm SLS}\cdot\ma{F}_{\rm SLS}^\herm$ can be written as
\ifCLASSOPTIONdraftcls
\begin{align}
   \ma{F}_{\rm SLS}\cdot\ma{F}_{\rm SLS}^\herm
  & = \diagof{\ma{J}_1 \cdot \ma{a}} \cdot \ma{G} \cdot \diagof{\ma{J}_1 \cdot \ma{a}}^\herm \label{eqn_app_proof_perf_singsrc_sls_ffh} \quad \mbox{where}  \\
  \ma{G} & = \frac{1}{M} \cdot \ma{1}_{(M-1)\times (M-1)} + 2 \cdot \ma{I}_{M-1} - \ma{J}_1 \cdot \ma{J}_2^\herm - \ma{J}_2 \cdot \ma{J}_1^\herm    \label{eqn_app_proof_perf_singsrc_sls_Gdef}
\end{align}
\else
\begin{align}
  &  \ma{F}_{\rm SLS}\cdot\ma{F}_{\rm SLS}^\herm
   = \diagof{\ma{J}_1 \cdot \ma{a}} \cdot \ma{G} \cdot \diagof{\ma{J}_1 \cdot \ma{a}}^\herm \label{eqn_app_proof_perf_singsrc_sls_ffh}  \\
  & \ma{G} = \frac{1}{M} \cdot \ma{1}_{(M-1)\times (M-1)} + 2 \cdot \ma{I}_{M-1} - \ma{J}_1 \cdot \ma{J}_2^\herm - \ma{J}_2 \cdot \ma{J}_1^\herm    \label{eqn_app_proof_perf_singsrc_sls_Gdef}
\end{align}

\fi
%
Equation~\eqref{eqn_app_proof_perf_singsrc_sls_ffh} shows that 
the inverse of $\ma{F}_{\rm SLS}\cdot\ma{F}_{\rm SLS}^\herm$ can be expressed as
$\diagof{\ma{J}_1 \cdot \ma{a}} \cdot \ma{G}^{-1}       
      \cdot \diagof{\ma{J}_1 \cdot \ma{a}}^\herm$.
To proceed further, we require the following Lemma:
\begin{lem} \label{lem_app_proof_perf_ginv}
   The inverse of the matrix $\ma{G}$
   defined in \eqref{eqn_app_proof_perf_singsrc_sls_Gdef}
   is given by the following expression
\ifCLASSOPTIONdraftcls   
   \begin{align}
       \matelem{\ma{G}^{-1}}{m_1}{m_2} = 
       &\begin{cases}
       \frac{1}{M} \cdot \left((M-m_1)\cdot m_2 - 3 \cdot \frac{m_1 \cdot (M - m_1) \cdot m_2 \cdot (M - m_2)}{M^2+11} \right)
           & m_1 \geq m_2 \\
			 \frac{1}{M} \cdot \left(m_1\cdot (M-m_2) - 3 \cdot \frac{m_1 \cdot (M - m_1) \cdot m_2 \cdot (M - m_2)}{M^2+11} \right)
           & m_1 < m_2
       \end{cases}
       \label{eqn_app_proof_perf_singsrc_sls_ginv} \\
       & m_1, m_2= 1, 2, \ldots, M-1. \notag
   \end{align}
\else
   \begin{align}
       \matelem{\ma{G}^{-1}}{m_1}{m_2} = 
       &\begin{cases}
       \frac{1}{M} \cdot \Big((M-m_1)\cdot m_2 &  \\ - 3 \cdot \frac{m_1 \cdot (M - m_1) \cdot m_2 \cdot (M - m_2)}{M^2+11} \Big)
           & m_1 \geq m_2 \\
			 \frac{1}{M} \cdot \Big(m_1\cdot (M-m_2) & \\ - 3 \cdot \frac{m_1 \cdot (M - m_1) \cdot m_2 \cdot (M - m_2)}{M^2+11} \Big)
           & m_1 < m_2
       \end{cases}
       \label{eqn_app_proof_perf_singsrc_sls_ginv} \\
       & m_1, m_2= 1, 2, \ldots, M-1. \notag
   \end{align}
\fi   
\end{lem}
\begin{proof}
 To prove this Lemma it is sufficient to multiply $\ma{G}^{-1}$ in \eqref{eqn_app_proof_perf_singsrc_sls_ginv}
 with $\ma{G}$ defined in \eqref{eqn_app_proof_perf_singsrc_sls_Gdef}
 and show that the result is an identity matrix. 
\end{proof}
Collecting our intermediate results from~\eqref{eqn_app_proof_perf_singsrc_sls_deltarls},
\eqref{eqn_app_proof_perf_singsrc_sls_ffh},
and \eqref{eqn_app_proof_perf_singsrc_sls_wru}
we have for $\Delta \ma{r}_{\rm SLS}^\trans$
\ifCLASSOPTIONdraftcls   
\begin{align}
    \Delta \ma{r}_{\rm SLS}^\trans = & \frac{(\ma{J}_1 \ma{a})^\herm}{\sqrt{M} \cdot \expof{\j \mu}}
    \cdot \diagof{\ma{J}_1 \cdot \ma{a}} \cdot \ma{G}^{-1} \cdot \diagof{\ma{J}_1 \cdot \ma{a}}^\herm
    \cdot \left(\ma{J}_1\cdot \ma{a} \left(\ma{J}_1\cdot \ma{a} \right)^+ - \ma{I}_{M-1} \right)
	   \cdot
	   \left(  \ma{J}_2 - \expof{\j \mu} \cdot \ma{J}_1 \right) \notag \\
     = & \frac{\gamma(M)}{\sqrt{M}}     
     \cdot \left(\ma{J}_1\cdot \ma{a} \right)^+ \cdot
	   \left( \ma{J}_2 / \expof{\j \mu} - \ma{J}_1 \right) 
      - \frac{\ma{g}_{\rm D}^\trans}{\sqrt{M}}     
	    \cdot \diagof{\ma{a}}^\herm
      \label{eqn_app_proof_perf_singsrc_sls_deltarls_ir2}
\end{align}
\else
\begin{align}
    \Delta \ma{r}_{\rm SLS}^\trans = & \frac{(\ma{J}_1 \ma{a})^\herm}{\sqrt{M} \cdot \expof{\j \mu}}
    \cdot \diagof{\ma{J}_1 \cdot \ma{a}} \cdot \ma{G}^{-1} \cdot \diagof{\ma{J}_1 \cdot \ma{a}}^\herm
    \notag \\ 
    & 
    \cdot \left(\ma{J}_1\cdot \ma{a} \left(\ma{J}_1\cdot \ma{a} \right)^+ - \ma{I}_{M-1} \right)
	   \cdot
	   \left(  \ma{J}_2 - \expof{\j \mu} \cdot \ma{J}_1 \right) \notag \\
     = & \frac{\gamma(M)}{\sqrt{M}}     
     \cdot \left(\ma{J}_1\cdot \ma{a} \right)^+ \cdot
	   \left( \ma{J}_2 / \expof{\j \mu} - \ma{J}_1 \right) 
      - \frac{\ma{g}_{\rm D}^\trans}{\sqrt{M}}     
	    \cdot \diagof{\ma{a}}^\herm
      \label{eqn_app_proof_perf_singsrc_sls_deltarls_ir2}
\end{align}
\fi
where the scalar $\gamma(M)$ and the row-vector $\ma{g}_{\rm D}^\trans$ are defined
as
\begin{align}
   \gamma(M) & = \ma{1}_{1 \times (M-1)} \cdot \ma{G}^{-1} \cdot \ma{1}_{(M-1) \times 1} \\
   \ma{g}_{\rm D}^\trans & = \ma{1}_{1 \times (M-1)} \cdot \ma{G}^{-1} \cdot \left( \ma{J}_2  - \ma{J}_1 \right)
   \in \real^{1 \times M}
\end{align}
For $\gamma(M)$ we can show via Lemma~\ref{lem_app_proof_perf_ginv}
%
\begin{align}
   \gamma(M) & = 
   \sum_{m_1=1}^{M-1}  \sum_{m_2=1}^{M-1} \matelem{\ma{G}^{-1}}{m_1}{m_2}
   = \frac{(M-1) M (M+1)}{M^2+11}. \notag
\end{align}
%
Moreover, for 
 the $m$-th element of 
the vector $\ma{g}_{\rm D}^\trans$, which we denote as $g_{{\rm D},m}$
we can show
\begin{align}
  g_{{\rm D},m}
 & =  \frac{6}{M^2+11}
      \cdot (2m-M-1), \quad m=1, 2, \ldots, M
\end{align}
Collecting our intermediate results, we have shown that $\ma{r}_{\rm SLS}^\trans$
can be written as
\ifCLASSOPTIONdraftcls   
\begin{align}
   \ma{r}_{{\rm SLS}}^\trans & = \ma{r}_{{\rm LS}}^\trans - \Delta\ma{r}_{{\rm SLS}}^\trans \notag \\
   & = 
   \sqrt{M} \cdot \left(
    \left(1 - \frac{\gamma(M)}{M}\right)\cdot (\ma{J}_1 \ma{a})^+ \cdot \left(\frac{\ma{J}_2}{\expof{\j \mu}} - \ma{J}_1\right)
    + \frac{1}{M} \cdot\ma{g}_{\rm D}^\trans \cdot \diagof{\ma{a}}^\herm\right)
    \label{eqn_app_proof_perf_singsrc_sls_deltarls_rslsT1}
\end{align}
\else
\begin{align}
   \ma{r}_{{\rm SLS}}^\trans & = \ma{r}_{{\rm LS}}^\trans - \Delta\ma{r}_{{\rm SLS}}^\trans \notag \\
   & = 
   \sqrt{M} \cdot \Big(
    \Big(1 - \frac{\gamma(M)}{M}\Big)\cdot (\ma{J}_1 \ma{a})^+ \cdot \Big(\frac{\ma{J}_2}{\expof{\j \mu}} - \ma{J}_1\Big) \notag \\ &
    + \frac{1}{M} \cdot\ma{g}_{\rm D}^\trans \cdot \diagof{\ma{a}}^\herm\Big)
    \label{eqn_app_proof_perf_singsrc_sls_deltarls_rslsT1}
\end{align}
\fi
where $\ma{r}_{\rm LS}$ has been taken from~\eqref{eqn_app_proof_subsp_d1_rk_ls}.
The next step to computing the mean square error is to calculate the squared norm
of the vector
$\ma{W}_{\rm mat}^\trans \cdot \ma{r}_{{\rm SLS}}$. The first few steps in computing this
product are very similar to the LS case. Following~\eqref{eqn_app_proof_subsp_d1_defatst}
we find that in the SLS case, the result is again equal to the product of
the squared norm the same vector $\ma{\tilde{s}}^\trans$ and 
the squared norm of a modified vector $\ma{\tilde{a}}^\trans_{\rm SLS}$,
i.e., $\twonorm{\ma{r}_{{\rm SLS}}^\trans \cdot \ma{W}_{\rm mat}}^2  = \twonorm{\ma{\tilde{s}}^\trans}^2 \cdot
\twonorm{\ma{\tilde{a}}^\trans_{\rm SLS}}^2$, where
\begin{align}
\ma{\tilde{a}}^\trans_{\rm SLS}  & = 
    \ma{r}_{{\rm SLS}}^\trans \cdot \projp{\ma{a}} 
\end{align}
and $\ma{r}_{{\rm SLS}}^\trans$ has been computed in \eqref{eqn_app_proof_perf_singsrc_sls_deltarls_rslsT1}.
Applying similar 
arguments as in~\eqref{eqn_app_proof_subsp_d1se_atilde}, $\ma{\tilde{a}}^\trans_{\rm SLS}$
can be further simplified into
\ifCLASSOPTIONdraftcls   
\begin{align}
\ma{\tilde{a}}^\trans_{\rm SLS}
  = & \phantom{+} 12 \cdot \frac{\sqrt{M}}{(M-1)(M^2+11)} \cdot 
[-1, 0, \ldots, 0, \expof{-(M-1) \j \mu}] 
 + 
\frac{\sqrt{M}}{M} \cdot
     \frac{6}{M^2+11}
      \cdot \notag \\
      & \phantom{+}[(-M+1), (-M+3)\expof{-\j\mu}, \ldots, (M-1)\cdot\expof{-(M-1) \j \mu}] \notag
\end{align}
\else
\begin{align}
& \ma{\tilde{a}}^\trans_{\rm SLS}
  =  12 \cdot \frac{\sqrt{M}}{(M-1)(M^2+11)} \cdot 
[-1, 0, \ldots, 0, \expof{-(M-1) \j \mu}] \notag \\
 &\quad + 
\frac{\sqrt{M}}{M} \cdot
     \frac{6}{M^2+11}
      \cdot \notag \\ & \quad
      [(-M+1), (-M+3)\expof{-\j\mu}, \ldots, (M-1)\cdot\expof{-(M-1) \j \mu}] \notag
\end{align}
\fi
We conclude that the two vectors $\ma{\tilde{a}}^\trans_{\rm SLS}$ consists of have the same
phase in each element and can hence be conveniently combined. When computing the squared norm of 
$\ma{\tilde{a}}^\trans_{\rm SLS}$ by summing the squared magnitude of all elements the phase
terms cancel which also confirms the intuition the the result should be independent
of the particular position $\mu$. 
We obtain
\ifCLASSOPTIONdraftcls   
\begin{align}
   \twonorm{\ma{\tilde{a}}^\trans_{\rm SLS}}^2
   = & \phantom{+} \sqrt{M}^2  \cdot \left(
   \frac{-12}{(M-1) (M^2+11)} + \frac{6\cdot(-M+1)}{M(M^2+11)}\right)^2 \notag \\
     & + \sum_{m=2}^{M-1}
     \frac{36}{M(M^2+11)^2} \cdot (-M+2m-1)^2 \notag \\
     & + \sqrt{M}^2  \cdot \left(
   \frac{12}{(M-1) (M^2+11)} + \frac{6\cdot(M-1)}{M(M^2+11)}\right)^2 \notag \\
  = & 12\cdot \frac{M^4 - 2 M^3 + 24 M^2 - 22 M + 23}{(M^2+11)^2(M-1)^2} 
  \label{eqn_app_proof_perf_singsrc_sls_normatsls}
\end{align}
\else
\begin{align}
    \twonorm{\ma{\tilde{a}}^\trans_{\rm SLS}}^2
   & =  \sqrt{M}^2  \cdot \left(
   \frac{-12}{(M-1) (M^2+11)} + \frac{6\cdot(-M+1)}{M(M^2+11)}\right)^2 \notag \\
     & + \sum_{m=2}^{M-1}
     \frac{36}{M(M^2+11)^2} \cdot (-M+2m-1)^2 \notag \\
     & + \sqrt{M}^2  \cdot \left(
   \frac{12}{(M-1) (M^2+11)} + \frac{6\cdot(M-1)}{M(M^2+11)}\right)^2 \notag \\
  = & 12\cdot \frac{M^4 - 2 M^3 + 24 M^2 - 22 M + 23}{(M^2+11)^2(M-1)^2} 
  \label{eqn_app_proof_perf_singsrc_sls_normatsls}
\end{align}
\fi
The mean square error is given by (cf. equation~\eqref{eqn_app_proof_perf_singsrc_1d_afr})
$   \expvof{(\Delta\mu_{\rm SLS})^2} = \frac{\sigma_{\rm n}^2}{2} \cdot \twonorm{\ma{\tilde{s}}^\trans}^2 \cdot \twonorm{\ma{\tilde{a}}^\trans_{\rm SLS}}^2$.
Inserting~\eqref{eqn_app_proof_mse_d1se_norms} and~\eqref{eqn_app_proof_perf_singsrc_sls_normatsls}
we have
\ifCLASSOPTIONdraftcls   
\begin{align}
   \expvof{(\Delta\mu_{\rm SLS})^2} & = \frac{\sigma_{\rm n}^2}{2} \cdot \frac{12}{M  N  \hat{P}_{\rm T}} \cdot \frac{M^4 - 2 M^3 + 24 M^2 - 22 M + 23}{(M^2+11)^2(M-1)^2}  \notag \\
   & = \frac{\sigma_{\rm n}^2}{N \cdot \hat{P}_{\rm T}} \cdot 6 \cdot 
   \frac{M^4 - 2 M^3 + 24 M^2 - 22 M + 23}{M(M^2+11)^2(M-1)^2},
\end{align}
\else
\begin{align}
   & \expvof{(\Delta\mu_{\rm SLS})^2} \notag \\
   &= \frac{\sigma_{\rm n}^2}{2} \cdot \frac{12}{M  N  \hat{P}_{\rm T}} \cdot \frac{M^4 - 2 M^3 + 24 M^2 - 22 M + 23}{(M^2+11)^2(M-1)^2} \notag \\
   & = \frac{\sigma_{\rm n}^2}{N \cdot \hat{P}_{\rm T}} \cdot 6 \cdot 
   \frac{M^4 - 2 M^3 + 24 M^2 - 22 M + 23}{M(M^2+11)^2(M-1)^2},
\end{align}
\fi
which is the desired result. \qed

\section{Proof of Theorem~\ref{thm_perf_mse_singsrc_2d}}\label{sec_app_proof_perf__mse_singsrc_2d}

\subsection{\texorpdfstring{$R$-D Standard ESPRIT}{R-D Standard ESPRIT}}

The proof for the $R$-D extension is in fact quite similar to the proof for the 1-D case
provided in Section~\ref{sec_app_proof_perf__mse_singsrc_1d}. 
In fact,~\eqref{eqn_app_proof_perf_singsrc_1d_x0} is still valid, the only difference being
that $\ma{a}(\mu)$ becomes $\ma{a}(\mu^{(1)}) \kron \ldots \kron \ma{a}(\mu^{(R)}) = \ma{a}(\ma{\mu})$. Therefore,
the first steps of the derivation can still be performed in the very same way. We obtain
the MSE for $R$-D Standard ESPRIT as
\begin{align}
    \expvof{\left(\Delta\mu^{(r)}\right)^2} = \frac{\sigma_{\rm n}^2}{2}
    \cdot \twonorm{\ma{r}^{(r)^\trans} \cdot \ma{W}_{\rm mat}}^2
    = \frac{\sigma_{\rm n}^2}{2} \cdot \twonorm{\ma{\tilde{\ma{s}}}}^2 \cdot \twonorm{\tilde{\ma{a}}^{(r)}}^2,
\end{align}
where $\tilde{\ma{s}}$ is the same as in the 1-D case (cf. equation~\eqref{eqn_app_proof_subsp_d1_defatst})
and $\tilde{\ma{a}}^{(r)}$ is given by
\begin{align}
\tilde{\ma{a}}^{(r)^\trans} = 
\sqrt{M}
\frac{\ma{a}^\herm \ma{\tilde{J}}_1^{(r)^\herm}}{\twonorm{\ma{\tilde{J}}_1^{(r)} \ma{a}}^2}
                \left(\frac{\ma{\tilde{J}}_2^{(r)}}{\expof{\j \cdot \mu^{(r)}}} - \ma{\tilde{J}}_1^{(r)}\right)
          \cdot 
         \projp{\ma{a}}. \label{eqn_app_proof_subsp_d1rd_at}
\end{align}
Since $\ma{\tilde{J}}_1^{(r)}$ selects the $M_r-1$ out of $M_r$ elements in the $r$-th mode,
we have $\twonorm{\ma{\tilde{J}}_1^{(r)} \ma{a}}^2 = \frac{M}{M_r} \cdot (M_r-1)$.
Moreover, 
multiplying \eqref{eqn_app_proof_subsp_d1rd_at} out and using the fact that
$\ma{a}$ satisfies the shift invariance equation in the $r$-th mode, we
obtain
\ifCLASSOPTIONdraftcls
\begin{align}
\tilde{\ma{a}}^{(r)^\trans} = &
\frac{\sqrt{M}\cdot M_r}{M \cdot (M_r-1)}
\cdot \left(  \ma{a}^\herm \cdot \ma{\tilde{J}}_1^{(r)^\herm} \cdot     
                \ma{\tilde{J}}_2^{(r)}/\expof{\j \cdot \mu^{(r)}} 
            - \ma{a}^\herm \cdot \ma{\tilde{J}}_1^{(r)^\herm} \cdot     
                \ma{\tilde{J}}_1^{(r)}\right)      \label{eqn_app_proof_subsp_d1rd_ir1}      
\end{align}
\else
\begin{align}
\tilde{\ma{a}}^{(r)^\trans} = &
\frac{\sqrt{M}\cdot M_r}{M \cdot (M_r-1)}
\cdot \Big(  \ma{a}^\herm \cdot \ma{\tilde{J}}_1^{(r)^\herm} \cdot     
                \ma{\tilde{J}}_2^{(r)}/\expof{\j \cdot \mu^{(r)}}  \notag \\ &
            - \ma{a}^\herm \cdot \ma{\tilde{J}}_1^{(r)^\herm} \cdot     
                \ma{\tilde{J}}_1^{(r)}\Big)      \label{eqn_app_proof_subsp_d1rd_ir1}      
\end{align}
\fi
Since the the array steering vector $\ma{a}$ and the selection matrices $\ma{\tilde{J}}^{(r)}_\ell$ 
can be factored into Kronecker products according to
$\ma{a} = \ma{a}^{(1)} \kron \ldots \ma{a}^{(R)}$ and
$\ma{\tilde{J}}_\ell^{(r)} = \ma{I}_{\prod_{n=1}^{r-1} M_n} \kron 
\ma{J}_\ell^{(r)} \kron \ma{I}_{\prod_{n=r+1}^R M_n}$, for $\ell=1,2$ and $r=1, 2, \ldots, R$,
all ``unaffected'' modes can be factored out of~\eqref{eqn_app_proof_subsp_d1rd_ir1} and
we have
\ifCLASSOPTIONdraftcls
\begin{align}
   \tilde{\ma{a}}^{(r)^\trans}  = 
   \frac{\sqrt{M}\cdot M_r}{M \cdot (M_r-1)}
   \cdot \left( \ma{a}^{(1)} \kron \ldots \kron \ma{a}^{(r-1)}\right)^\herm
   \kron \left(\ma{\tilde{a}}^{(r)}_1 - \ma{\tilde{a}}^{(r)}_2 \right)^\trans
   \kron \left( \ma{a}^{(r+1)} \kron \ldots \kron \ma{a}^{(R)}\right)^\herm \label{eqn_app_proof_subsp_d1rd_atrk}
\end{align}
\else
\begin{align}
   \tilde{\ma{a}}^{(r)^\trans} & = 
   \frac{\sqrt{M}\cdot M_r}{M \cdot (M_r-1)}
   \cdot \left( \ma{a}^{(1)} \kron \ldots \kron \ma{a}^{(r-1)}\right)^\herm \notag \\ &
   \kron \left(\ma{\tilde{a}}^{(r)}_1 - \ma{\tilde{a}}^{(r)}_2 \right)^\trans
   \kron \left( \ma{a}^{(r+1)} \kron \ldots \kron \ma{a}^{(R)}\right)^\herm \label{eqn_app_proof_subsp_d1rd_atrk}
\end{align}
\fi
where $\ma{\tilde{a}}^{(r)}_1$ and $\ma{\tilde{a}}^{(r)}_2$ are given by
\ifCLASSOPTIONdraftcls
\begin{align}
    \ma{\tilde{a}}^{(r)^\trans}_1  = \ma{a}^{(r)^\herm} \cdot \ma{{J}}_1^{(r)^\herm} \cdot     
                \ma{{J}}_2^{(r)}/\expof{\j \cdot \mu^{(r)}} \quad \mbox{and} \quad
    \ma{\tilde{a}}^{(r)^\trans}_2  = \ma{a}^{(r)^\herm} \cdot \ma{{J}}_1^{(r)^\herm} \cdot     
                \ma{{J}}_1^{(r)}.
\end{align}
\else
\begin{align}
    \ma{\tilde{a}}^{(r)^\trans}_1 & = \ma{a}^{(r)^\herm} \cdot \ma{{J}}_1^{(r)^\herm} \cdot     
                \ma{{J}}_2^{(r)}/\expof{\j \cdot \mu^{(r)}} \quad \mbox{and} \notag\\ 
    \ma{\tilde{a}}^{(r)^\trans}_2  & = \ma{a}^{(r)^\herm} \cdot \ma{{J}}_1^{(r)^\herm} \cdot     
                \ma{{J}}_1^{(r)}.
\end{align}
\fi
Following the same reasoning as for~\eqref{eqn_app_proof_subsp_d1se_a12tilde} we
find
\begin{align}
   \ma{\tilde{a}}^{(r)^\trans}_1 - \ma{\tilde{a}}^{(r)^\trans}_2
   =  [-1, 0, \ldots, 0, \expof{-(M_r-1) \j \mu^{(r)}}]
   \label{eqn_app_proof_subsp_d1serd_a12tilde}
\end{align}
Consequently, the desired norm $\twonorm{\tilde{\ma{a}}^{(r)}}^2$ is directly
found to be
\ifCLASSOPTIONdraftcls
\begin{align}
  \twonorm{\tilde{\ma{a}}^{(r)}}^2
  = &
    \frac{M\cdot M_r^2}{M^2 \cdot (M_r-1)^2}
    \cdot
    \left(
    \prod_{n=1}^{r-1} \twonorm{\ma{a}^{(n)}}^2 
    \right)\cdot 2 \cdot     
    \left(
    \prod_{n=r+1}^{R} \twonorm{\ma{a}^{(n)}}^2
    \right)
%
 =  2 \cdot 
   \frac{M_r}{(M_r-1)^2}.    
\end{align}
\else
\begin{align}
  \twonorm{\tilde{\ma{a}}^{(r)}}^2
  = &
    \frac{M\cdot M_r^2}{M^2 \cdot (M_r-1)^2}
    \cdot
    \left(
    \prod_{n=1}^{r-1} \twonorm{\ma{a}^{(n)}}^2 
    \right)\cdot 2 \cdot     
    \prod_{n=r+1}^{R} \twonorm{\ma{a}^{(n)}}^2
%
\notag \\  
 = & 2 \cdot 
   \frac{M_r}{(M_r-1)^2}.    
\end{align}
\fi
Therefore, the MSE expression for $R$-D Standard ESPRIT
is given by
\ifCLASSOPTIONdraftcls
\begin{align}
   \expvof{\left(\Delta\mu^{(r)}\right)^2}
   & = 
    \frac{\sigma_{\rm n}^2}{2} \cdot
    \frac{1}{M \cdot N \cdot \hat{P}_{\rm T}}
    \cdot
    2 \cdot 
   \frac{M_r}{(M_r-1)^2}    
  = \frac{\sigma_{\rm n}^2}{ N \cdot \hat{P}_{\rm T}} \cdot
   \frac{M_r}{M \cdot(M_r-1)^2},    
\end{align}
\else
\begin{align}
   \expvof{\left(\Delta\mu^{(r)}\right)^2}
   & = 
    \frac{\sigma_{\rm n}^2}{2} \cdot
    \frac{1}{M \cdot N \cdot \hat{P}_{\rm T}}
    \cdot
    2 \cdot 
   \frac{M_r}{(M_r-1)^2}    \notag \\ &
  = \frac{\sigma_{\rm n}^2}{ N \cdot \hat{P}_{\rm T}} \cdot
   \frac{M_r}{M \cdot(M_r-1)^2},    
\end{align}
\fi
which proofs the first part of the theorem. \qed

\subsection{\texorpdfstring{$R$-D Unitary ESPRIT}{R-D Unitary ESPRIT}}

The second part of the theorem is to prove that in the $R$-D case the performance
of $R$-D Unitary ESPRIT and $R$-D Standard ESPRIT are the same as long
as a single source is present. However, for this part, no changes have
to be made compared to Appendix~\ref{sec_app_proof_perf__mse_singsrc_1due}: As
it was shown there, Forward-Backward-Averaging only affects $\ma{v}_{\rm s}$
and has no effect on $\ma{u}_{\rm s}$ or $\ma{U}_{\rm n}$. Applying the same
steps here immediately proves this part of the theorem.

\subsection{\texorpdfstring{Cram\'er-Rao Bound}{Cramer-Rao Bound}}

The third part is the simplification of the Cramér-Rao Bound.
In the $R$-D case, the CRB is given by
\begin{align}
    \ma{C} = \frac{\sigma_{\rm n}^2}{2\cdot N} \cdot
    \realof{
      \left[ \ma{D}^{(R)^\herm} \cdot 
      \projp{\ma{A}} \cdot \ma{D}^{(R)}
      \right] \odot \left( \ma{1}_{R \times R} \kron \ma{\hat{R}}_{\rm S}^\trans\right)
    }^{-1} \notag
\end{align}
where $\projp{\ma{A}} = \ma{I}_M - \ma{A} \cdot \left(\ma{A}^\herm \cdot \ma{A} \right)^{-1} \cdot \ma{A}^\herm$ and $\ma{D}^{(R)} \in \compl^{M \times (d\cdot R)}$ contains the partial derivatives
of the array steering vectors $\ma{a}_n$ with respect to $\mu_n^{(r)}$ for $n=1, 2, \ldots, d$
and $r=1, 2, \ldots, R$. For the special case of a single source, the CRB simplifies into
(cf.~Appendix \ref{app_proof_d11d_crb})
\begin{align}
    \ma{C} & =
    \frac{\sigma_{\rm n}^2}{2\cdot N \cdot \hat{P}_{\rm T}} \cdot
    \realof{\ma{J}}^{-1}, \notag \\
    \ma{J} & = 
       \ma{D}^{(R)^\herm} \cdot \left(
      \ma{I}_M - \frac{1}{M} \cdot \ma{a} \cdot \ma{a}^\herm
      \right) \cdot \ma{D}^{(R)}           
\end{align}
The columns of $\ma{D}^{(R)} \in \compl^{M \times R}$ are given by
$\ma{\tilde{d}}^{(r)} = \frac{\partial \ma{a}}{\partial \mu^{(r)}} \in \compl^{M \times 1}$. Using
the fact that $\ma{a} = \ma{a}^{(1)} \kron \ldots \kron \ma{a}^{(R)}$ we
obtain
\begin{align}
\ma{\tilde{d}}^{(r)} = \ma{a}^{(1)} \kron \ldots \kron \ma{a}^{(r-1)}
                   \kron \ma{d}^{(r)} \kron
                        \ma{a}^{(r+1)} \kron \ldots \kron \ma{a}^{(R)} \label{eqn_app_proof_subsp_d1rdcrb_ir1}
\end{align}
where $\ma{d}^{(r)} = \frac{\partial \ma{a}^{(r)}}{\partial \mu^{(r)}} \in \compl^{M_r \times 1}
= \j \cdot [0, \expof{\j \mu^{(r)}}, 2\cdot \expof{2 \j \mu^{(r)}}, \ldots, 
(M-1) \expof{(M-1) \j \mu^{(r)}}]$.
Therefore, the elements of the matrix $\ma{J}$ are given by
\begin{align}
   \matelem{\ma{J}}{r_1}{r_2} = \ma{\tilde{d}}^{(r_1)^\herm} \cdot \ma{\tilde{d}}^{(r_2)} - \frac{1}{M} \cdot \ma{\tilde{d}}^{(r_1)^\herm}\cdot \ma{a}\cdot
\ma{a}^{\herm} \cdot \ma{\tilde{d}}^{(r_1)}
\end{align}
With the help of~\eqref{eqn_app_proof_subsp_d1rdcrb_ir1} we find for the diagonal
elements ($r_1 = r_2 = r$)
\ifCLASSOPTIONdraftcls
\begin{align}
   \ma{\tilde{d}}^{(r)^\herm} \cdot \ma{\tilde{d}}^{(r)}    
   & = \frac{M}{M_r}
   \cdot \left(\sum_{m=0}^{M_r-1} m^2\right)  
    = \frac{1}{6} \cdot M \cdot (M_r-1)  \cdot (2M_r-1)
\end{align}
\else
\begin{align}
   \ma{\tilde{d}}^{(r)^\herm} \cdot \ma{\tilde{d}}^{(r)}    
   & = \frac{M}{M_r}
   \cdot \left(\sum_{m=0}^{M_r-1} m^2\right)  \notag \\
   & = \frac{1}{6} \cdot M \cdot (M_r-1)  \cdot (2M_r-1)
\end{align}
\fi
and similarly
\ifCLASSOPTIONdraftcls
\begin{align}
   \ma{\tilde{d}}^{(r)^\herm} \cdot \ma{a}
   & = \frac{M}{M_r}
   \cdot \left(-\j \sum_{m=0}^{M_r-1} m\right)
    = -\j \cdot M \cdot \frac{1}{2} \cdot (M_r-1).
\end{align}
\else
\begin{align}
   \ma{\tilde{d}}^{(r)^\herm} \cdot \ma{a}
   & = \frac{M}{M_r}
   \cdot \left(-\j \sum_{m=0}^{M_r-1} m\right) \notag \\
   & = -\j \cdot M \cdot \frac{1}{2} \cdot (M_r-1).
\end{align}
\fi
Combining these two results we have for $\matelem{\ma{J}}{r}{r}$
\begin{align}
    \matelem{\ma{J}}{r}{r}
    & = \frac{1}{12} \cdot M \cdot (M_r-1)  (M_r+1).
\end{align}
On the other hand, for the off-diagonal elements we obtain
\begin{align}
   \ma{\tilde{d}}^{(r_1)^\herm} \cdot \ma{\tilde{d}}^{(r_2)} 
   & =  \frac{1}{4} M\cdot (M_{r_1}-1) \cdot M_{r_1}\cdot (M_{r_2}-1) \cdot M_{r_2} \notag
\end{align}
and therefore for $\matelem{\ma{J}}{r_1}{r_2}$, $r_1 \neq r_2$
\ifCLASSOPTIONdraftcls
\begin{align}
 \matelem{\ma{J}}{r_1}{r_2}  = & 
 \ma{\tilde{d}}^{(r_1)^\herm} \cdot \ma{\tilde{d}}^{(r_2)} - 
 \ma{\tilde{d}}^{(r_1)^\herm} \cdot \ma{a}
 \cdot \ma{a}^\herm \cdot \ma{\tilde{d}}^{(r_2)^\herm} 
    = 0. \notag
\end{align}
\else
\begin{align}
 \matelem{\ma{J}}{r_1}{r_2}  = & 
 \ma{\tilde{d}}^{(r_1)^\herm} \cdot \ma{\tilde{d}}^{(r_2)} - 
 \ma{\tilde{d}}^{(r_1)^\herm} \cdot \ma{a}
 \cdot \ma{a}^\herm \cdot \ma{\tilde{d}}^{(r_2)^\herm}  = 0. \notag 
\end{align}
\fi
This shows that $\ma{J}$ is diagonal and real-valued. Consequently, the CRB
becomes
\ifCLASSOPTIONdraftcls
\begin{align}
\ma{C} & = \frac{\sigma_{\rm n}^2}{2\cdot N \cdot \hat{P}_{\rm T}} \cdot  \realof{\ma{J}}^{-1}  = 
   \diagof{\left[C^{(1)}, \; \ldots, C^{(R)}\right]} 
   \notag \\
C^{(r)} & = \frac{\sigma_{\rm n}^2}{2\cdot N \cdot \hat{P}_{\rm T}} \frac{12}{M \cdot (M_r-1) \cdot (M_r+1) }
= \frac{1}{\hat{\rho}} \cdot \frac{6}{M \cdot (M_r-1) \cdot (M_r+1) },
\end{align}
\else
\begin{align}
\ma{C} & = \frac{\sigma_{\rm n}^2}{2\cdot N \cdot \hat{P}_{\rm T}} \cdot  \realof{\ma{J}}^{-1}  = 
   \diagof{\left[C^{(1)}, \; \ldots, C^{(R)}\right]} 
   \notag \\
C^{(r)} & = \frac{\sigma_{\rm n}^2}{2\cdot N \cdot \hat{P}_{\rm T}} \frac{12}{M \cdot (M_r-1) \cdot (M_r+1) }
\notag\\&
= \frac{1}{\hat{\rho}} \cdot \frac{6}{M \cdot (M_r-1) \cdot (M_r+1) },
\end{align}
\fi
which is the desired result. \qed

\subsection{\texorpdfstring{$R$-D Standard Tensor-ESPRIT}{R-D Standard Tensor-ESPRIT}}

The fourth part of the theorem is to show that the MSE of $R$-D Standard
Tensor-ESPRIT is the same as the MSE for $R$-D Standard ESPRIT for $d=1$. Since we have
only shown the expressions for $R$-D Standard Tensor-ESPRIT in the special case $R=2$,
we will also assume this case here.

Note that the MSE expression for Tensor-ESPRIT is in fact quite similar to the one
for matrix-based ESPRIT with the only difference being that the matrix $\ma{W}_{\rm mat}$
is replaced by the matrix $\ma{W}_{\rm ten}$, cf.~\eqref{eqn_subsp_perf_wmat} and~\eqref{eqn_subsp_perf_wten_r2},
respectively.

To simplify this expression for the special case $d=1$, we express
the unfoldings of $\ten{X}_0$ as
\ifCLASSOPTIONdraftcls
\begin{align}   
   \unfnot{\ten{X}_0}{1} & = \ma{a}^{(1)} \cdot \left( \ma{a}^{(2)} \kron \ma{s}\right)^\trans,\quad 
   \unfnot{\ten{X}_0}{2}  = \ma{a}^{(2)} \cdot \left( \ma{s} \kron \ma{a}^{(1)}\right)^\trans,\quad 
   \unfnot{\ten{X}_0}{3}  = \ma{s} \cdot \left( \ma{a}^{(1)} \kron \ma{a}^{(2)}\right)^\trans \notag
\end{align}
\else
\begin{align}   
   \unfnot{\ten{X}_0}{1} & = \ma{a}^{(1)} \cdot \left( \ma{a}^{(2)} \kron \ma{s}\right)^\trans,\; 
   \unfnot{\ten{X}_0}{2}  = \ma{a}^{(2)} \cdot \left( \ma{s} \kron \ma{a}^{(1)}\right)^\trans, \notag \\
   \unfnot{\ten{X}_0}{3} & = \ma{s} \cdot \left( \ma{a}^{(1)} \kron \ma{a}^{(2)}\right)^\trans. \notag
\end{align}
\fi
Consequently, we can relate the necessary subspaces of the unfoldings of $\ten{X}_0$ to $\ma{s}$ and $\ma{a}^{(r)}$
via
\ifCLASSOPTIONdraftcls
\begin{align}
  \begin{split}
    \sig{\ma{u}}_1 & = \frac{\ma{a}^{(1)}}{\sqrt{M_1}}, \quad 
    \sig{\ma{u}}_2  = \frac{\ma{a}^{(2)}}{\sqrt{M_2}}, \quad 
    \sig{\ma{u}}_3   = \frac{\ma{s}}{\sqrt{N \cdot \hat{P}_{\rm T}}},\quad 
    \noi{\ma{U}}_1  = \projp{\ma{a}^{(1)}},  \quad
	  \noi{\ma{U}}_2  = \projp{\ma{a}^{(2)}}  \\
    \sig{\ma{\Sigma}}_1 &= \sig{\ma{\Sigma}}_2 = \sig{\ma{\Sigma}}_3 = \sqrt{M \cdot N \cdot \hat{P}_{\rm T}}\\
    \sig{\ma{v}}_1 &= \frac{\left( \ma{a}^{(2)} \kron \ma{s}\right)^\conj}{\sqrt{M_2\cdot N \cdot \hat{P}_{\rm T}}}, \quad 
    \sig{\ma{v}}_2  = \frac{\left( \ma{s}\kron \ma{a}^{(1)} \right)^\conj}{\sqrt{M_1\cdot N \cdot \hat{P}_{\rm T}}}, 
    \quad 
    \sig{\ma{v}}_3 = \ma{u}_{\rm s} = \frac{\ma{a}}{\sqrt{M}}     \\
    \noiC{\ma{V}}_3\cdot\noiT{\ma{V}}_3 & = \ma{U}_{\rm n} \cdot \ma{U}_{\rm n}^\herm = 
    \projp{\ma{a}}.
  \end{split}
  \label{eqn_app_proof_subsp_d1ste_sx_e}
\end{align}
\else
\begin{align}
 \begin{split}
  & \sig{\ma{u}}_1  = \frac{\ma{a}^{(1)}}{\sqrt{M_1}}, \quad 
    \sig{\ma{u}}_2  = \frac{\ma{a}^{(2)}}{\sqrt{M_2}}, \quad 
    \sig{\ma{u}}_3   = \frac{\ma{s}}{\sqrt{N \cdot \hat{P}_{\rm T}}} 
    \\
  & \noi{\ma{U}}_1   = \projp{\ma{a}^{(1)}}, \;
    \noi{\ma{U}}_2   = \projp{\ma{a}^{(2)}} \\
  & \sig{\ma{\Sigma}}_1 = \sig{\ma{\Sigma}}_2 = \sig{\ma{\Sigma}}_3 = \sqrt{M \cdot N \cdot \hat{P}_{\rm T}}\\
  & \sig{\ma{v}}_1 = \frac{\left( \ma{a}^{(2)} \kron \ma{s}\right)^\conj}{\sqrt{M_2\cdot N \cdot \hat{P}_{\rm T}}}, \quad 
    \sig{\ma{v}}_2  = \frac{\left( \ma{s}\kron \ma{a}^{(1)} \right)^\conj}{\sqrt{M_1\cdot N \cdot \hat{P}_{\rm T}}}, 
    \\ & \quad\quad 
    \sig{\ma{v}}_3 = \ma{u}_{\rm s} = \frac{\ma{a}}{\sqrt{M}}     \\
  & \noiC{\ma{V}}_3\cdot\noiT{\ma{V}}_3  = \ma{U}_{\rm n} \cdot \ma{U}_{\rm n}^\herm = 
    \projp{\ma{a}}. %
 \end{split}
 \label{eqn_app_proof_subsp_d1ste_sx_e}
\end{align}
\fi
Moreover, we have for $\ma{T}_r$
\begin{align}
    \ma{T}_r & =  \sig{\ma{u}}_r \cdot \sigH{\ma{u}}_r = 
    \proj{\ma{a}^{(r)}}
    \quad \mbox{for $r=1,2$ and thus} \notag \\
    \ma{T}_1 \kron \ma{T}_2 & = 
    \proj{\ma{a}^{(1)}} \kron \proj{\ma{a}^{(2)}}
     = \proj{\ma{a}}. \label{eqn_app_proof_subsp_d1rdste_ir2}
\end{align}
From~\eqref{eqn_app_proof_subsp_d1ste_sx_e} and~\eqref{eqn_app_proof_subsp_d1rdste_ir2}
it immediately follows that the first term in $\ma{W}_{\rm ten}$ cancels as it contains
$\left[\ma{T}_1 \kron \ma{T}_2\right] \cdot \noiC{\ma{V}}_3\cdot\noiT{\ma{V}}_3$.
We also find $\ma{t}_{r,m} = \ma{a}^{(r)} \cdot \expof{-\j \mu^{(r)} (m-1)} / M_r$ for $r=1,2$
and $m=1,2,\ldots,M_r$.
To simplify the remaining two terms in $\ma{W}_{\rm ten}$ we first simplify some of their
components. 
Using the identity $\ma{u}_{\rm s} = \ma{a}/\sqrt{M}$ and the explicit expression for $\ma{t}_{r,m} $
it is easy to show that
\begin{align}
   \left(\ma{u}_{\rm s}^\trans \kron \ma{I}_M\right) \cdot \ma{\bar{T}}_1
    & = 
    \frac{1}{\sqrt{M}} \left(\ma{a}^{(2)^\trans} \kron \ma{a}^{(1)} \kron \ma{I}_{M_2}\right)
   \\
   \left(\ma{u}_{\rm s}^\trans \kron \ma{I}_M\right) \cdot \ma{\bar{T}}_2 & = 
   \frac{1}{\sqrt{M}}\left(\ma{a}^{(1)^\trans} \kron \ma{I}_{M_1} \kron \ma{a}^{(2)} \right).
\end{align}
Moreover, using the relations from \eqref{eqn_app_proof_subsp_d1ste_sx_e} we
can rewrite $ \sigC{\ma{U}}_r \siginv{\ma{\Sigma}_r} \sigT{\ma{V}}_r
        \kron \noi{\ma{U}_r} \noiH{\ma{U}_r}$ as
%
\ifCLASSOPTIONdraftcls
	\begin{align}
	\left(\sigC{\ma{U}}_1 \siginv{\ma{\Sigma}_1} \sigT{\ma{V}}_1\right) \kron \left(\noi{\ma{U}_1} \noiH{\ma{U}_1}\right)
	& = 
	\frac{1}{M N \hat{P}_{\rm T}} \cdot \left(\ma{a}^{(1)^\conj} \cdot \left(\ma{a}^{(2)} \kron \ma{s}\right)^\herm\right)
	\kron \projp{\ma{a}^{(1)}} \notag \\
	\left(\sigC{\ma{U}}_2 \siginv{\ma{\Sigma}_2} \sigT{\ma{V}}_2\right) \kron \left(\noi{\ma{U}_2} \noiH{\ma{U}_2}\right)
	& = 
	\frac{1}{M N \hat{P}_{\rm T}} \cdot \left(\ma{a}^{(2)^\conj} \cdot \left(\ma{s} \kron \ma{a}^{(1)}\right)^\herm\right)
	\kron \projp{\ma{a}^{(2)}}. \notag
	\end{align}
\else
	\begin{align}
	& \left(\sigC{\ma{U}}_1 \siginv{\ma{\Sigma}_1} \sigT{\ma{V}}_1\right) \kron \left(\noi{\ma{U}_1} \noiH{\ma{U}_1}\right) \notag \\
	 = &
	\frac{1}{M N \hat{P}_{\rm T}} \cdot 
	  \Big(\ma{a}^{(1)^\conj} \cdot \big(\ma{a}^{(2)} \kron \ma{s}\big)^\herm\Big)
	\kron \projp{\ma{a}^{(1)}} \notag \\
	& \left(\sigC{\ma{U}}_2 \siginv{\ma{\Sigma}_2} \sigT{\ma{V}}_2\right) \kron \left(\noi{\ma{U}_2} \noiH{\ma{U}_2}\right) \notag \\
	 = &
	\frac{1}{M N \hat{P}_{\rm T}} \cdot 
	  \Big(\ma{a}^{(2)^\conj} \cdot \big(\ma{s} \kron \ma{a}^{(1)}\big)^\herm\Big)
	\kron \projp{\ma{a}^{(2)}}. \notag
	\end{align}
\fi
%
Combining these intermediate result, the third term in $\ma{W}_{\rm ten}$ can be expressed as
\begin{align}
 &  \left(\ma{u}_{\rm s}^\trans \kron \ma{I}_M\right) \cdot \ma{\bar{T}}_1
    \cdot
    \Big(
    \big(\sigC{\ma{U}}_2 \siginv{\ma{\Sigma}_2} \sigT{\ma{V}}_2\big) \kron 
    \big(\noi{\ma{U}_2} \noiH{\ma{U}_2}\big)\Big) \notag \\
  = &
\frac{M_2}{M N \hat{P}_{\rm T}\sqrt{M}} \left(
    \ma{s}^\herm \kron \ma{a}^{(1)^\herm}
   \kron
   \ma{a}^{(1)} \kron 
   \projp{\ma{a}^{(2)}}
   \right) \label{eqn_app_proof_subsp_d1ste_ir2}
\end{align}
With similar arguments, the second term in $\ma{W}_{\rm ten}$ can be simplified into
\ifCLASSOPTIONdraftcls
\begin{align}
 &  \left(\ma{u}_{\rm s}^\trans \kron \ma{I}_M\right) \cdot \ma{\bar{T}}_2
    \cdot
    \Big(
    \big(\sigC{\ma{U}}_1 \siginv{\ma{\Sigma}_1} \sigT{\ma{V}}_1\big) \kron 
    \big(\noi{\ma{U}_1} \noiH{\ma{U}_1}\big)\Big) \cdot \ma{K}_{M_2 \times (M_1\cdot N)}\notag \\    
  = &
\frac{M_1}{M N \hat{P}_{\rm T}\sqrt{M}} \left(
    \ma{a}^{(2)^\herm} \kron \ma{s}^\herm 
   \kron \projp{\ma{a}^{(1)}}
    \kron \ma{a}^{(2)}   
   \right)\cdot \ma{K}_{M_2 \times (M_1\cdot N)} \notag \\
   = &
\frac{ M_1}{M N \hat{P}_{\rm T}\sqrt{M}} \left(
    \ma{s}^\herm 
   \kron \projp{\ma{a}^{(1)}}
    \kron \ma{a}^{(2)}  \kron \ma{a}^{(2)^\herm}
   \right).\label{eqn_app_proof_subsp_d1ste_ir3}
\end{align}
\else
\begin{align}
 & \mbox{\small $\left(\ma{u}_{\rm s}^\trans \kron \ma{I}_M\right) \cdot \ma{\bar{T}}_2
    \cdot
    \Big(
    \big(\sigC{\ma{U}}_1 \siginv{\ma{\Sigma}_1} \sigT{\ma{V}}_1\big) \kron 
    \big(\noi{\ma{U}_1} \noiH{\ma{U}_1}\big)\Big) \cdot \ma{K}_{M_2 \times (M_1\cdot N)}$}\notag \\    
&  = 
\frac{M_1}{M N \hat{P}_{\rm T}\sqrt{M}} \left(
    \ma{a}^{(2)^\herm} \kron \ma{s}^\herm 
   \kron \projp{\ma{a}^{(1)}}
    \kron \ma{a}^{(2)}   
   \right)\cdot \ma{K}_{M_2 \times (M_1\cdot N)} \notag \\
&   = 
\frac{ M_1}{M N \hat{P}_{\rm T}\sqrt{M}} \left(
    \ma{s}^\herm 
   \kron \projp{\ma{a}^{(1)}}
    \kron \ma{a}^{(2)}  \kron \ma{a}^{(2)^\herm}
   \right).\label{eqn_app_proof_subsp_d1ste_ir3}
\end{align}
\fi
where the last step is a special case of Property~\eqref{eqn_commat_permkron} for commutation matrices.

Using~\eqref{eqn_app_proof_subsp_d1ste_ir2} and~\eqref{eqn_app_proof_subsp_d1ste_ir3} 
in~\eqref{eqn_subsp_perf_wten_r2}, we obtain
\begin{align}
   \ma{W}_{\rm ten} 
   & = \frac{1}{N \hat{P}_{\rm T}\sqrt{M}}  \cdot \ma{s}^\herm \kron
   \left(
    \projp{\ma{a}^{(1)}}
    \kron \proj{\ma{a}^{(2)}}
    +     
    \proj{\ma{a}^{(1)}} \kron  \projp{\ma{a}^{(2)}}
   \right).   \label{eqn_app_proof_subsp_d1ste_wten_a}
\end{align}
%
Comparing~\eqref{eqn_app_proof_subsp_d1ste_wten_a}
matrix-based counterpart
in~\eqref{eqn_app_proof_subsp_d1rd_at}
we find that for a single source, 
$\ma{W}_{\rm mat}$ and $\ma{W}_{\rm ten}$ are in fact quite similar, the only difference
being that $\projp{\ma{a}}$
is replaced by $\projp{\ma{a}^{(1)}} \kron \proj{\ma{a}^{(2)}}
    + \proj{\ma{a}^{(1)}} \kron  \projp{\ma{a}^{(2)}}$.
Therefore, to show the $R$-D Standard ESPRIT and $R$-D Standard Tensor-ESPRIT
have the same MSE for $d=1$, it is sufficient to show that the corresponding terms $\ma{\tilde{a}}^{(r)}$
are the same, i.e., that
\begin{align}
\ma{\bar{a}}^{(r)}  \cdot  \projp{\ma{a}}
= \ma{\bar{a}}^{(r)} \cdot 
   \left(\projp{\ma{a}^{(1)}} \kron \proj{\ma{a}^{(2)}}
       + \proj{\ma{a}^{(1)}} \kron  \projp{\ma{a}^{(2)}}\right) \label{eqn_app_proof_subsp_d1ste_ident}
\end{align}
where $\ma{\bar{a}}^{(r)} = \ma{a}^\herm \ma{\tilde{J}}_1^{(r)^\herm}
                \left(\ma{\tilde{J}}_2^{(r)}/\expof{\j \cdot \mu^{(r)}} - \ma{\tilde{J}}_1^{(r)}\right)$
for $r=1, 2$. Note that $\ma{\bar{a}}^{(r)}  \cdot  \projp{\ma{a}}$
was shown to be equal to (cf. equation~\eqref{eqn_app_proof_subsp_d1rd_atrk})
%
$ \Big(\ma{\tilde{a}}_1^{(1)} - \ma{\tilde{a}}_2^{(1)}\Big)^\trans \kron \ma{a}^{(2)^\trans} $
for $r=1$ and 
$ \ma{a}^{(1)^\trans} \kron \Big(\ma{\tilde{a}}_1^{(2)} - \ma{\tilde{a}}_2^{(2)}\Big)^\trans $
for $r=2$, where 
   $\ma{\tilde{a}}^{(r)^\trans}_1  = \ma{a}^{(r)^\herm} \cdot \ma{{J}}_1^{(r)^\herm} \cdot     
                \ma{{J}}_2^{(r)}/\expof{\j \cdot \mu^{(r)}}$
and $\ma{\tilde{a}}^{(r)^\trans}_2  = \ma{a}^{(r)^\herm} \cdot \ma{{J}}_1^{(r)^\herm} \cdot     
                \ma{{J}}_1^{(r)}$.
%
Expanding the corresponding right-hand side of~\eqref{eqn_app_proof_subsp_d1ste_ident} we have for $r=1$
%
\ifCLASSOPTIONdraftcls
\begin{align}
   \ma{\bar{a}}^{(r)} \cdot 
  \left(\projp{\ma{a}^{(1)}} \kron \proj{\ma{a}^{(2)}}
      + \proj{\ma{a}^{(1)}} \kron  \projp{\ma{a}^{(2)}}\right) & =    
  \left(\ma{a}^{(1)^\herm} \ma{{J}}_1^{(1)^\herm}
                \left(\ma{{J}}_2^{(1)}/\expof{\j \cdot \mu^{(1)}} - \ma{{J}}_1^{(1)}\right)
                \cdot \projp{\ma{a}^{(1)}}
  \right) \kron \ma{a}^{(2)^\herm} \notag \\
  & = 
  \left(\ma{\tilde{a}}^{(1)^\trans}  - \ma{\tilde{a}}^{(2)^\trans} \right) \kron \ma{a}^{(2)^\herm} \notag,
\end{align}
\else
\begin{align}
  & \ma{\bar{a}}^{(r)} \cdot 
  \left(\projp{\ma{a}^{(1)}} \kron \proj{\ma{a}^{(2)}}
      + \proj{\ma{a}^{(1)}} \kron  \projp{\ma{a}^{(2)}}\right) \notag \\
      = &
  \left(\ma{a}^{(1)^\herm} \ma{{J}}_1^{(1)^\herm}
                \left(\ma{{J}}_2^{(1)}/\expof{\j \cdot \mu^{(1)}} - \ma{{J}}_1^{(1)}\right)
                \cdot \projp{\ma{a}^{(1)}}
  \right) \kron \ma{a}^{(2)^\herm} \notag \\
  = &
  \left(\ma{\tilde{a}}^{(1)^\trans}  - \ma{\tilde{a}}^{(2)^\trans} \right) \kron \ma{a}^{(2)^\herm} \notag,
\end{align}
\fi
where we have used the fact that $\ma{a} = \ma{a}^{(1)} \kron \ma{a}^{(2)}$
from its definition,
$\ma{a}^{(2)^\herm} \cdot \proj{\ma{a}^{(2)}} = \ma{a}^{(2)^\herm}$
and $\ma{a}^{(2)^\herm} \cdot \projp{\ma{a}^{(2)}} = \ma{0}_{1 \times M_2}$
since $\proj{\ma{a}^{(2)}}$ and $\projp{\ma{a}^{(2)}}$ are projectors onto
$\ma{a}^{(2)}$ and its orthogonal complement,
and the identity
$\ma{{J}}_2^{(1)}/\expof{\j \cdot \mu^{(1)}}\ma{a}^{(1)} - \ma{{J}}_1^{(1)} \ma{a}^{(1)} = \ma{0}_{(M_2-1)\times 1}$
since $\ma{a}^{(1)}$ satisfies the shift invariance equation for $r=1$.
This shows that the left-hand side and the right-hand side of~\eqref{eqn_app_proof_subsp_d1ste_ident}
are equal for $r=1$. The proof for $r=2$ proceeds in an analogous fashion. 
Consequently, we have shown that for $d=1$
\begin{align}
    \ma{r}^{(r)^\trans} \cdot \ma{W}_{\rm mat} = 
    \ma{r}^{(r)^\trans} \cdot \ma{W}_{\rm ten}, \quad \mbox{for $r=1, 2$}
\end{align}
and hence the MSE for 2-D Standard ESPRIT and 2-D Standard Tensor-ESPRIT are in fact equal. \qed


\subsection{\texorpdfstring{$R$-D Unitary Tensor-ESPRIT}{R-D Unitary Tensor-ESPRIT}}

The fifth and final part of the theorem is to show that the MSE for $R$-D Unitary ESPRIT is
again equal to the MSE for $R$-D Standard ESPRIT in case of a single source. Again, there 
is no need to derive this in full detail. As it was shown in Appendix~\ref{sec_app_proof_perf__mse_singsrc_1due},
Forward-Backward-Averaging has no effect on $\ma{u}_s$ or $\ma{U}_{\rm n}$ but only affects
$\ma{v}_s$ and $\ma{V}_{\rm n}$. This carries over to the tensor case where only
the quantities involving the symbols are affected. However, since the ``symbol part'' and the ``array part''
can always be factorized (cf. equation~\eqref{eqn_app_proof_subsp_d1ste_wten_a}), the arguments from 
Appendix~\ref{sec_app_proof_perf__mse_singsrc_1due} can still be applied to prove this part
of the theorem. \qed


\smallskip
\bibliographystyle{IEEEsort}
\bibliography{refs}

\end{document}